\documentclass[10pt,journal,comsoc]{IEEEtran}

\ifCLASSINFOpdf
  \usepackage[pdftex]{graphicx}
\else
  \usepackage[dvips]{graphicx}
\fi
\usepackage[paper=letterpaper,left=0.65in,right=0.65in,top=0.7in,bottom=1.0in]{geometry}

\usepackage[dvipsnames]{xcolor}

\usepackage{enumitem}
\setlist[enumerate]{leftmargin=0.5 cm}
\setlist[itemize]{leftmargin=0.3 cm}

\usepackage{epstopdf}
\usepackage{graphicx}
\usepackage{capt-of}
\usepackage{subfigure}
\usepackage{amsmath}

\allowdisplaybreaks

\usepackage{amsthm}
\usepackage{algorithm}
\usepackage[noend]{algpseudocode}

\usepackage[short,nocomma]{optidef}

\usepackage{harpoon}

\usepackage{amssymb}
\usepackage[mathscr]{eucal}
\usepackage{url}
\usepackage{thmtools}
\usepackage{mathtools}
\usepackage{color}
\usepackage{multirow}
\usepackage{siunitx}
\usepackage{comment}
\usepackage{cite}
\usepackage{physics}
\usepackage{diagbox}

\usepackage{stfloats}

\usepackage{cuted}
\usepackage{cases}
\usepackage{float}

\usepackage{flushend}

\theoremstyle{definition}

\declaretheoremstyle[
  headfont=\normalfont\bfseries,
  numbered=unless unique,
  bodyfont=\normalfont,
  qed={$\blacksquare$}
]{exmpstyle2}

\newtheorem{definition}{Definition}

\newtheorem{lemma}{Lemma}
\newtheorem{theorem}{Theorem}
\newtheorem{corollary}{Corollary}

\DeclareMathAlphabet{\mathpzc}{OT1}{pzc}{m}{it}

\newcommand{\Title}{Decoherence-Aware Entangling and Swapping Strategy Optimization for Entanglement Routing in Quantum Networks}

\usepackage{acro}

\DeclareAcronym{QN}{short=QN,long=quantum network}
\DeclareAcronym{SN}{short=SN,long=social network}
\DeclareAcronym{SD}{short=SD,long=source-destination}
\DeclareAcronym{QC}{short=QC,long=quantum computer}
\DeclareAcronym{QEC}{short=QEC,long=quantum error correction}
\DeclareAcronym{LP}{short=LP,long=linear programming}
\DeclareAcronym{ILP}{short=ILP,long=integer linear programming}

\DeclareAcronym{DP}{short=DP,long=dynamic programming}

\DeclareAcronym{REI}{short=REI, long=resource efficiency index}

\DeclareAcronym{MIS}{short=MIS, long=maximum independent set}

\DeclareAcronym{QKD}{short=QKD, long=quantum key distribution}

\newcommand{\ProblemName}{Op\textbf{t}imized D\textbf{e}coherence-aware S\textbf{tr}ategy for Entangl\textbf{i}ng and \textbf{S}wapping Scheduling Problem}

\newcommand{\ProblemNameAbbr}{TETRIS}
\newcommand{\ProblemNameAbbrN}{TETRIS-N}

\newcommand{\AlgoNametwo}{\textbf{F}ractional \textbf{N}umerology \textbf{P}acking and \textbf{R}ounding Algorithm}
\newcommand{\AlgoNameAbbrtwo}{FNPR}

\newcommand{\AlgoNameone}{\textbf{F}idelity-\textbf{L}oad \textbf{T}rade-\textbf{O}ff Algorithm}
\newcommand{\AlgoNameAbbrone}{FLTO}

\newcommand{\Merge}{Nesting}
\newcommand{\Linear}{Linear}
\newcommand{\ASAP}{ASAP}
\newcommand{\UB}{UB}

\newcommand{\NMOutformLinear}{$78\%$}
\newcommand{\NMOutformMerge}{$60\%$}
\newcommand{\NMOutformASAP}{$63\%$}

\newcommand{\FLTOOutformLinear}{$77\%$}
\newcommand{\FLTOOutformMerge}{$60\%$}
\newcommand{\FLTOOutformASAP}{$62\%$}

\makeatletter 
\newcommand{\customlabel}[2]{%
\protected@write \@auxout {}{\string \newlabel {#1}{{#2}{}}}}
\makeatother

\renewcommand{\baselinestretch}{0.99}
\definecolor{orange}{rgb}{1,0.5,0}
\graphicspath{{graph/}{photo/}}

\IEEEoverridecommandlockouts
\def\BibTeX{{\rm B\kern-.05em{\sc i\kern-.025em b}\kern-.08em T\
kern-.1667em\lower.7ex\hbox{E}\kern-.125emX}}

\begin{document}

\title{\Title}%

\author{
Shao-Min Huang\IEEEauthorrefmark{6},
Cheng-Yang~Cheng\IEEEauthorrefmark{6},
Ming-Huang~Chien\IEEEauthorrefmark{6},
Jian-Jhih~Kuo,
and
Chih-Yu~Wang
\IEEEcompsocitemizethanks{
    \IEEEcompsocthanksitem S.-M. Huang is with the Research Center for Information Technology Innovation, Academia Sinica, Taiwan, and also with the Department of Computer Science and Information Engineering, National Chung Cheng University, Taiwan (e-mail: shaominhuang2019@alum.ccu.edu.tw).
    \IEEEcompsocthanksitem C.-Y.~Cheng is with the Institute of Data Science and Engineering, National Yang Ming Chiao Tung University, Taiwan, and also with the Department of Computer Science and Information Engineering, National Chung Cheng University, Taiwan (e-mail: yang890912.cs12@nycu.edu.tw).
    \IEEEcompsocthanksitem M.-H. Chien is with the Institute of Computer Science and Engineering, National Yang Ming Chiao Tung University, Taiwan, and also with the Department of Computer Science and Information Engineering, National Chung Cheng University, Taiwan (e-mail: qq11123334.cs12@nycu.edu.tw).
    \IEEEcompsocthanksitem J.-J. Kuo is with the Department of Computer Science and Information Engineering, National Chung Cheng University, Taiwan, and also with the Advanced Institute of Manufacturing with High-tech Innovations, National Chung Cheng University, Taiwan (e-mail: lajacky@cs.ccu.edu.tw).
    \IEEEcompsocthanksitem C.-Y. Wang is with the Research Center for Information Technology Innovation, Academia Sinica, Taiwan (e-mail: cywang@citi.sinica.edu.tw).
    \IEEEcompsocthanksitem \IEEEauthorrefmark{6} denotes the equal contributions. Corresponding author: Jian-Jhih Kuo}%
}%


\maketitle
\begin{abstract} 
Quantum teleportation enables high-security communications through end-to-end quantum entangled pairs. 
End-to-end entangled pairs are created by using swapping processes to consume short entangled pairs and generate long pairs. 
However, due to environmental interference, entangled pairs decohere over time, resulting in low fidelity. 
Thus, generating entangled pairs at the right time is crucial. Moreover, the swapping process also causes additional fidelity loss.
To this end, this paper presents a short time slot protocol, where a time slot can only accommodate a process.
It has a more flexible arrangement of entangling and swapping processes than the traditional long time slot protocol. 
It raises a new optimization problem {\ProblemNameAbbr} for finding strategies of entangling and swapping for each request to maximize the fidelity sum of all accepted requests.
To solve the {\ProblemNameAbbr}, we design two novel algorithms with different optimization techniques.
Finally, the simulation results manifest that our algorithms can outperform the existing methods by 
up to $60\sim 78\%$ in general, and by $20\sim 75\%$ even under low entangling probabilities.
\end{abstract}


\IEEEpeerreviewmaketitle


\section{Introduction}
\label{sec: introduction}

\acresetall

\IEEEPARstart{A}{s}
the frontiers of modern technology, \acp{QN} have been developed and implemented to evaluate their practicability for information transmission \cite{elliott2002building,chen2021integrated,daiss2021quantum,komar2014quantum,nielsen2002quantum}.
\acp{QN} connect quantum nodes to transmit quantum bits (qubits) using end-to-end entangled pairs \cite{nielsen2002quantum} (i.e., quantum teleportation). These networks serve as the foundation of quantum services such as \ac{QKD} \cite{chen2021integrated}, distributed quantum computing \cite{daiss2021quantum}, and clock synchronization \cite{komar2014quantum}.
As illustrated in Fig. \ref{fig: 6Nswap}, each quantum node in the \ac{QN} has a specific amount of quantum memory (blue squares) to store entangled qubits, mitigating rapid decoherence \cite{schlosshauer2005decoherence,sangouard2011quantum,van2014quantum,nielsen2002quantum}. 
Adjacent nodes are interconnected by optical fibers (black lines) that facilitate the entanglement process.
However, an end-to-end entangled pair is not directly available between two non-adjacent nodes.
In such a case, their end-to-end entangled pair can be achieved via one or more entanglement swapping processes, each of which consumes two short entangled pairs to create a long one. 
For example, at time slot $t_2$ in Fig. \ref{fig: 6Nswap}, there are two entangled pairs $(v_1,v_2)$ and $(v_2,v_3)$. After performing a swapping at the end of $t_2$, we obtain a long entangled pair $(v_1, v_3)$ at time slot $t_3$ and free two quantum memory units at $v_2$ \cite{van2014quantum}.

\begin{figure}[t]
\centering
    \subfigure[Entangling and swapping from time slot $t_1$ to $t_5$ in a network]{\label{fig: 6Nswap}\includegraphics[width= .48\textwidth]{graph/5Nswap.pdf}}
    \subfigure[Skewed strategy tree]{\label{fig: treeskew}\includegraphics[width= 0.23\textwidth]{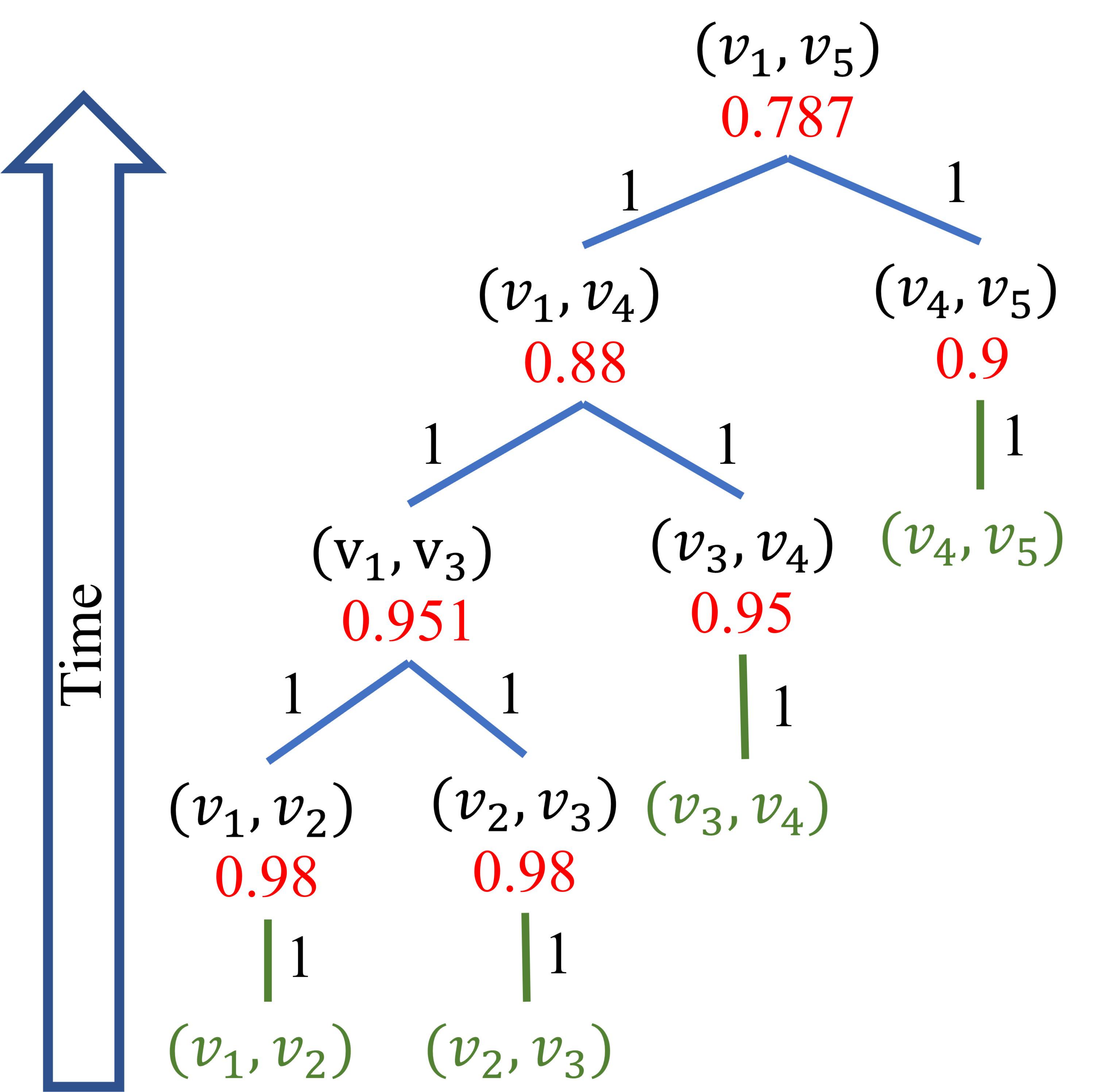}}
    \subfigure[Complete strategy tree]{\label{fig: treefull}\includegraphics[width= .202\textwidth]{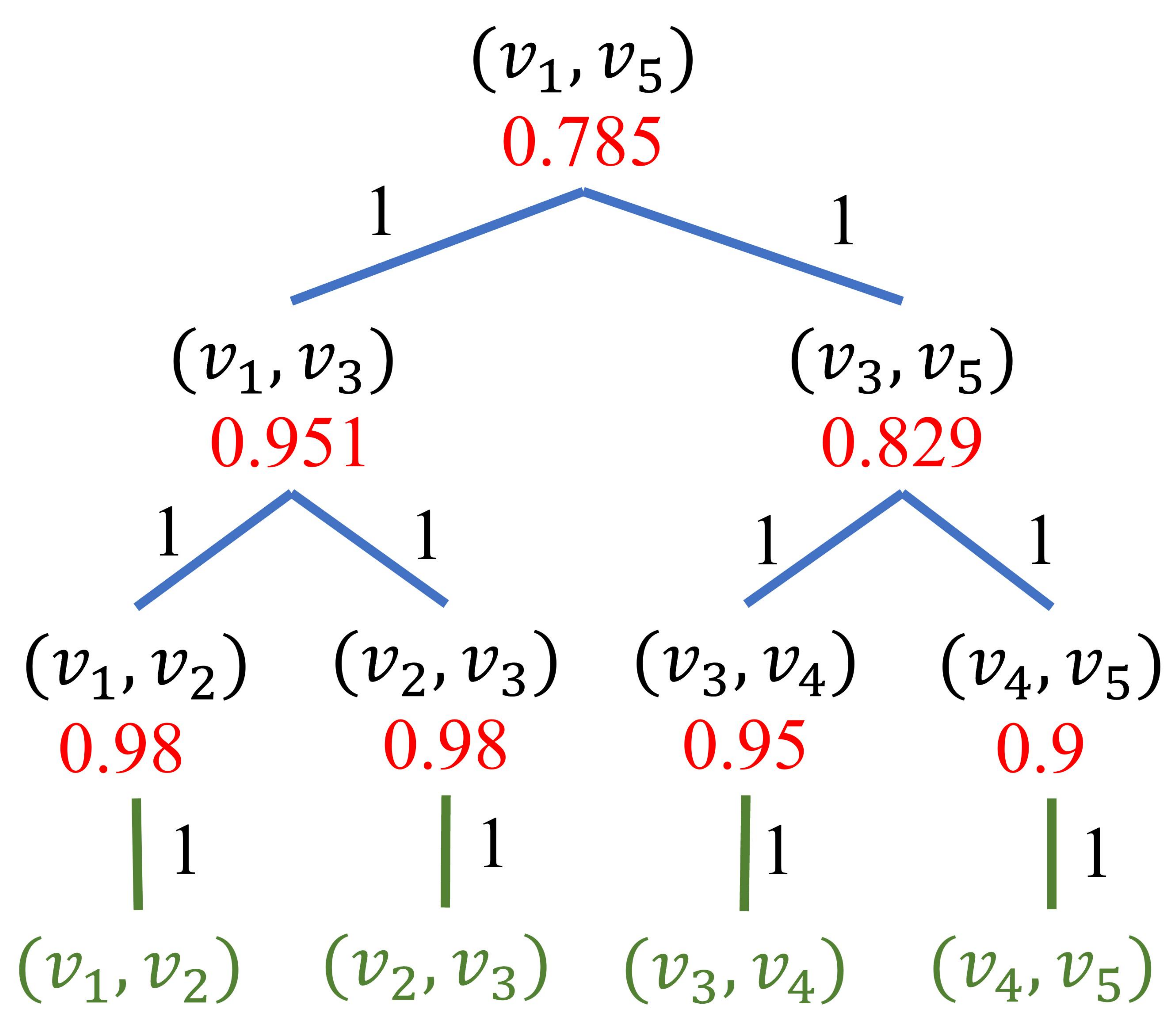}}
    \caption{Scheduling of entangling and swapping in \acp{QN}.}
    \label{fig: qexample}
\end{figure}

However, entangled pairs stored in the quantum memory decohere over time, resulting in a loss of quality, i.e., fidelity loss.
Thus, finding the perfect timing for entangling is crucial; an improper strategy may cause entangled pairs to idle and unnecessary fidelity loss. 
Moreover, swapping also causes fidelity loss as the entangled link generated by swapping exhibits lower fidelity than the two consumed entangled links, and the loss is exacerbated when the fidelity difference between the two pairs is significant.
Thus, the entangling and swapping strategy should be properly designed to improve fidelity.

Existing \ac{QN} protocols, studied in the literature, often employ a standard long time slot protocol setting \cite{pant2019routing, shi2020concurrent, zhao2021redundant, zhao2022e2e, chen2022heuristic, farahbakhsh2022opportunistic, zeng2023entanglement, zeng2023entanglement2, li2022fidelity, chakraborty2020entanglement, pouryousef2022quantum, zhao2023segmented, cicconetti2021request, ghaderibaneh2022efficient},
which may result in unnecessary fidelity loss when handling multiple requests. 
Specifically, in traditional long time slot protocols, the entangling processes of each request (i.e., \ac{SD} pair) take place during the entangling phase, while swapping processes occur during the swapping phase. As a result, shorter requests might have to wait for longer requests to complete because the latter involve more processes to perform.
Further, traditional long time slot protocols bind resources for the entire time slot, even though a swapping process frees two quantum memory units, which originally could be used for other requests but remain blocked by the protocol.

In light of the shortcomings, we propose a short time slot protocol where either entangling or swapping can occur in any given time slot.
Specifically, in our short time slot protocol, a node is free to perform one process within a single time slot, and the process could be entangling, swapping, teleporting, or idling. 
Thus, short requests no longer need to wait for other requests to establish their entangled pairs, as they can perform processes separately, while longer requests can utilize the resource released by the nodes performing swapping processes.

The adoption of a short time slot protocol enables more efficient entangling and swapping strategies, minimizing unnecessary fidelity loss by precisely timing entangled pair generation for subsequent swapping.
Take Figs. \ref{fig: 6Nswap} and \ref{fig: treeskew} for example, we have a request from $v_1$ to $v_5$.%
    \footnote{In this example, all the parameters follow the default settings in Section \ref{sec: evaluation}. In addition, the fidelity after decoherence over time and swapping is calculated using Eqs. (\ref{deco formula}) and (\ref{swapping formula}), respectively, in Section \ref{sec: background}.}
At time slot $t_1$, pairs $(v_1,v_2)$ and $(v_2,v_3)$ start entangling, taking one time slot.
At time slot $t_2$, both the generated pairs $(v_1, v_2)$ and $(v_2, v_3)$ have the initial fidelity $0.98$ and start decoherence while doing the swapping process. 
One may observe that we can choose to start the entangling of pair $(v_3,v_4)$ at either $t_1$ or $t_2$. 
Clearly, entangling at $t_2$ is better than $t_1$ since they will have less wait time for subsequent swapping processes, leading to less decoherence and better fidelity.
On the other hand, swapping processes cause extra fidelity loss, and the fidelity of $(v_1,v_3)$ is $0.951$, which is lower than the fidelity before swapping (i.e., the two links with fidelity $0.98$ after decoherence for $1$ time slot will have fidelity $0.975$). 
At time slot $t_3$, the generated pair $(v_1, v_3)$ with fidelity $0.951$ and $(v_3, v_4)$ with fidelity $0.95$ start decoherence while swapping, leading to pair $(v_1, v_4)$ with fidelity $0.88$ after swapping. 
Meanwhile, pair $(v_4,v_5)$ start entangling.
Finally, $(v_1, v_5)$ are generated at $t_5$ and has fidelity $0.787$.
Unlike the skewed strategy tree in Fig. \ref{fig: treeskew}, the strategy in Fig. \ref{fig: treefull} requires only four slots. Intuitively, it should have better fidelity because it suffers less decoherence time. 
However, the result is counter-intuitive; Fig. \ref{fig: treeskew} has better fidelity than Fig. \ref{fig: treefull} because most of the swapping processes in Fig. \ref{fig: treeskew} consume two entangled pairs with similar fidelity and thus suffer less fidelity loss due to swapping processes.%
    \footnote{More simulations and clarifications of the reasons behind this counter-intuitive example are provided in Appendix \ref{app: discussion1} of the supplementary material.}
From the perspective above, we can guess that if all the links on the path have similar initial fidelity, the complete strategy tree will perform better than others.
For example, if all the generated pairs $(v_1,v_2)$, $(v_2,v_3)$, $(v_3,v_4)$, and $(v_4,v_5)$ have initial fidelity $0.98$, then the final fidelity of the strategy trees in Figs. \ref{fig: treeskew} and \ref{fig: treefull} will be $0.889$ and $0.891$, respectively.

We have two observations from the above examples:
1) swapping two entangled pairs with similar fidelity performs better, and
2) the more skewed the strategy tree, the less memory is required within a time slot.
Moreover, a different strategy leads to varying fidelity and distribution of resource consumption.
In this paper, such a distribution is referred to as \emph{numerology}, representing how the resources (i.e., quantum memory in this paper) are consumed for the request.
The numerology provides a direct representation of the corresponding strategy to examine its feasibility. The quality of the strategy, which is the resulting fidelity of the request, can also be derived directly through the corresponding numerology. Thus, the entangling and swapping strategy formation can be transformed into the numerology selection problem.

In this paper, we aim to choose a numerology for each request to maximize the fidelity sum for all accepted requests while meeting the fidelity threshold within a batch of time slots, i.e., the {\ProblemName} (\ProblemNameAbbr). 
The {\ProblemNameAbbr} introduces four novel challenges:
1) Entangled pairs will decohere over time. How can we choose the right time to generate entangled pairs to reduce their wait time?
2) The swapping process may cause additional fidelity loss. How can we achieve an intelligent swapping strategy to minimize this loss?
3) Numerologies determine the resulting fidelity. How can we identify the feasible numerologies for requests to meet the fidelity threshold, guaranteeing the quality of teleportation?
4) Quantum memory is limited, and its availability will be affected by simultaneous requests. How can we find efficient numerologies to mitigate resource contention and satisfy more requests while ensuring a better fidelity sum?

It can be shown that the {\ProblemNameAbbr} with no fidelity threshold constraint (i.e., {\ProblemNameAbbrN}) is already NP-hard and cannot be approximated within any factor of $|I|^{1-\epsilon}$, where $|I|$ denotes the number of \ac{SD} pairs.
To conquer the above challenges, we propose two novel algorithms.
1) The first one is {\AlgoNametwo} (\AlgoNameAbbrtwo).
{\AlgoNameAbbrtwo} aims to maximize the number of accepted requests.
With \ac{LP} rounding, {\AlgoNameAbbrtwo} can achieve a bi-criteria approximation ratio for the {\ProblemNameAbbrN} as the actual duration of a time slot and the number of time slots are sufficiently small. Then, we extend {\AlgoNameAbbrtwo} to solve the {\ProblemNameAbbr}.
2) The second one is {\AlgoNameone} (\AlgoNameAbbrone).
{\AlgoNameAbbrone} utilizes two efficient \ac{DP} algorithms to find two candidate numerologies for each request with a given path. One maximizes fidelity, and the other considers resource utilization.
Then, {\AlgoNameAbbrone} iteratively invokes them to allocate the resource for each request based on a subtly-designed index called \ac{REI}, inspired by \cite{greedyIndicator}.

In sum, the contributions of this paper are as follows:
1) To the best of our knowledge, this is the first attempt to consider the decoherence of entangled pairs over time for entangling and swapping scheduling while applying a short time slot protocol that offers better resource allocation.
2) We prove that the {\ProblemNameAbbr} is an NP-hard problem, and even its special case, the {\ProblemNameAbbrN}, cannot be approximated within any factor of $|I|^{1-\epsilon}$.
3) We propose a combinatorial algorithm with a clever separation oracle to achieve a bi-criteria approximation. Meanwhile, we develop two \ac{DP} algorithms to find candidate numerologies with different goals and design an index to choose numerologies adaptively.

\section{Related Work}
\label{sec: related work}

Before discussing the related works on optimization problems for entangling and swapping, we review previous works that examine feasibility and realization of \acp{QN} \cite{elliott2002building, design_qrn, muralidharan2016optimal, Pirandola_2017, optimal_routing, zhao2023distributed}. 
Early foundational studies have laid the groundwork for secure communication protocols and robust architectures in quantum networks.
Elliott \emph{et al.} introduced \ac{QN} to realize secure communications \cite{elliott2002building}. Meter \emph{et al.} proposed a large \ac{QN} architecture with layered recursive repeaters, where repeaters may not trust each other, and then designed new protocol layers to support quantum sessions to ensure robustness and interoperable communication \cite{design_qrn}. 
Muralidharan \emph{et al.} presented new quantum nodes that can execute the \ac{QEC} processes and classified the theoretically feasible technologies of quantum nodes into three generations \cite{muralidharan2016optimal}.   
Pirandola \emph{et al.} discussed the limits of repeater-less quantum communications and provided general benchmarks for repeaters \cite{Pirandola_2017}.
Caleffi \emph{et al.} designed a routing protocol to ensure a high end-to-end entanglement rate between any two nodes \cite{optimal_routing}.
Zhao \emph{et al.} proposed two transport layer protocols for quantum data networks that achieve high throughput and fairness \cite{zhao2023distributed}.

Building on these foundational works, subsequent research has focused on path selection and entangling and swapping process optimization for long time slot system-based \acp{QN} \cite{pant2019routing, shi2020concurrent, zhao2021redundant, zhao2022e2e, chen2022heuristic, farahbakhsh2022opportunistic, zeng2023entanglement, zeng2023entanglement2, li2022fidelity, chakraborty2020entanglement, pouryousef2022quantum, zhao2023segmented, cicconetti2021request, ghaderibaneh2022efficient}.
Pant \emph{et al.} presented a greedy algorithm to determine the paths for each request \cite{pant2019routing}.
Shi \emph{et al.} designed a routing method Q-CAST based on the Dijkstra algorithm to find primary paths and recovery paths to mitigate the effect of entanglement failures \cite{shi2020concurrent}.
Zhao \emph{et al.} presented an LP-based algorithm REPS to maximize throughput in SDN-based QNs \cite{zhao2021redundant}.
Zhao \emph{et al.} considered entangled links' fidelity and exploited quantum purification to enhance link fidelity \cite{zhao2022e2e}.
Chen \emph{et al.} proposed two heuristic algorithms for entangling and swapping phases separately \cite{chen2022heuristic}.
Farahbakhsh \emph{et al.} developed an add-up scheme to store and forward data qubits as much as possible \cite{farahbakhsh2022opportunistic}.
Zeng \emph{et al.} studied how to simultaneously maximize the number of quantum-user pairs and their expected throughput \cite{zeng2023entanglement}.
Zeng \emph{et al.} utilized the properties of $n$-fusion to help increase the success probability of constructing a long entangled pair for quantum networks \cite{zeng2023entanglement2}.
Li \emph{et al.} exploited purification to meet the fidelity requirements for multiple SD pairs as many as possible \cite{li2022fidelity}.
Chakraborty \emph{et al.} gave an LP formulation to compute the maximum total entanglement distribution rate \cite{chakraborty2020entanglement}.
Pouryousef \emph{et al.} attempted to stockpile entangled pairs in advance when the traffic demand is low while using them otherwise \cite{pouryousef2022quantum}.
Zhao \emph{et al.} considered a room-size network scenario, where a longer entangled pair can be established by sending one of its photons via all-optical switching without swapping at intermediate nodes \cite{zhao2023segmented}.
However, none of the above works considers the effect of decoherence on fidelity.
To this end, Cicconetti \emph{et al.} studied different policies of path selection and request scheduling and then measured the resulting link fidelity under the influence of dephasing and depolarizing noises (i.e., decoherence) \cite{cicconetti2021request}.
Ghaderibaneh \emph{et al.} further considered time decoherence when determining the swapping order for each SD pair to minimize entangled pair generation latency \cite{ghaderibaneh2022efficient}.
However, all of the above works use long time slot systems and conduct entangling and swapping sequentially, which may block short requests to wait for long requests, resulting in more memory idle time.

To address the limitations of long time slot systems, some studies \cite{huang2023socially, yang2023asynchronous} have explored short time slot systems or even asynchronous protocols to improve efficiency.
Huang \emph{et al.} employed a short time slot protocol while opportunistically forwarding data qubits via social nodes to ensure security. However, they neglected decoherence over time, fidelity loss due to swapping processes, and freed qubits after swapping processes \cite{huang2023socially}. 
Yang \emph{et al.} proposed an asynchronous model to freely generate entangled pairs or conduct swapping processes without time slot limitations.
It helps establish end-to-end connections more quickly while utilizing more resources, as resources can be freed immediately \cite{yang2023asynchronous}. However, it did not consider decoherence or fidelity loss.  
In contrast, our short time slot system considers more practical fidelity losses due to time decoherence and swapping processes. It further leverages diverse numerologies for requests to maximize the fidelity sum, achieving efficient entangling and swapping scheduling.

Although some short time slot approaches have been proposed to enhance efficiency, they still overlook fidelity degradation over time due to decoherence and swapping processes. To this end, the following works \cite{Haldar2024Fast, Inesta2023Optimal, Haldar2024Reducing, Kenneth2024On} consider decoherence when constructing remote entangled pairs along a specific path.
Haldar \emph{et al.} employed the $Q$-learning reinforcement-learning (RL) algorithm to discover policies that optimize both average waiting time and fidelity for a single request \cite{Haldar2024Fast}.
I\~{n}esta \emph{et al.} proposed a method based on Markov decision processes-based method with value and policy iteration to minimize the expected needed time to achieve end-to-end entanglement for a single request \cite{Inesta2023Optimal}.
Haldar \emph{et al.} proposed the quasi-local multiplexing policies, following the SWAP-ASAP approach and incorporating entanglement purification, to optimize both average waiting time and fidelity for a single request in linear chain quantum networks \cite{Haldar2024Reducing}.
Goodenough \emph{et al.} analyzed the value of fidelity up to $25$ segments over time using a generating function. The approach can be applied to find cut-off policies in \ac{QKD} \cite{Kenneth2024On}.
Since the above studies focused on handling single requests, their policies will bind sufficient resources for each accepted request, enabling re-entangling and swapping until a successful end-to-end entangled link is established. 
Thus, they may perform well in single-request scenarios but are less effective when dealing with multiple requests. 

Table \ref{tab: related work} summarizes the related works based on their setups (e.g., time slot length, fidelity decoherence, scheduling method, path selection) and optimization objectives (e.g., throughput, fidelity, latency, waiting time).

\begin{table*}[ht]
    \small
    \setcounter{table}{0}
    \centering
    \caption{Comparison of related works}
    \label{tab: related work}
    \begin{tabular}{|>{\centering\arraybackslash}m{5em}|>{\centering\arraybackslash}m{5em}|>{\centering\arraybackslash}m{5em}|>{\centering\arraybackslash}m{5.5em}|>{\centering\arraybackslash}m{5em}|>{\centering\arraybackslash}m{8.5em}|>{\centering\arraybackslash}m{13.5em}|}
    \hline
    \textbf{Literature} & \textbf{Time slot}  & \textbf{Fidelity decohere over time} & \textbf{Multi-request scheduling} & \textbf{Path selection} & \textbf{Objective} &\textbf{Solution strategy}\\
    \hline 
    \cite{pant2019routing} & Long & No & No & Yes & Max throughput & Greedy \\ 
    \hline 
    \cite{shi2020concurrent} & Long & No & No & Yes & Max throughput & Dijkstra-based \\
    \hline
    \cite{zhao2021redundant} & Long & No & No & Yes & Max throughput & LP rounding + Heuristic \\
    \hline
    \cite{zhao2022e2e} & Long & No & No & Yes & Max throughput & Dijkstra-based + LP rounding \\
    \hline
    \cite{chen2022heuristic} & Long & No & No & Yes & Max throughput & Heuristic + DP \\
    \hline
    \cite{farahbakhsh2022opportunistic} & Long & No & No & No & Min waiting time & Greedy \\
    \hline
    \cite{zeng2023entanglement} & Long & No & No & Yes & Max user pairs \& Max throughput & Dijkstra-based + LP rounding + Branch and bound algorithm \\
    \hline
    \cite{zeng2023entanglement2} & Long & No & No & Yes & Max throughput & Dijkstra-based \\
    \hline
    \cite{li2022fidelity} & Long & No & No & Yes & Max throughput & Dijkstra-based + Greedy \\
    \hline
    \cite{chakraborty2020entanglement} & Long & No & No & Yes & Max throughput & Multicommodity flow-based algorithm \\
    \hline
    \cite{pouryousef2022quantum} & Long & No & No & Yes & Max throughput \& Min delay & LP \\
    \hline
    \cite{zhao2023segmented} & Long & No & No & Yes & Max throughput & LP rounding + Heuristic \\
    \hline
    \cite{cicconetti2021request} & Long & No & Yes & Yes & Max throughput & Heuristic + Dijkstra-based \\
    \hline
    \cite{ghaderibaneh2022efficient} & Long & Yes & No & No & Min latency & DP \\ 
    \hline
    \cite{huang2023socially} & Short & No & No & Yes & Min waiting time & Greedy \\
    \hline
    \cite{yang2023asynchronous} & Async & No & No & Yes & Max network efficiency & Dijkstra-based + Primal-Dual-based algorithm \\
    \hline
    \cite{Haldar2024Fast} & Short & Yes & No & No & Min waiting time \& Max fidelity & $Q$-learning reinforcement-learning\\ 
    \hline
    \cite{Inesta2023Optimal} & Short & Yes & No & No & Min delivery time & Markov decision process-based algorithm \\
    \hline
    \cite{Haldar2024Reducing} & Short & Yes & No & No & Min waiting time \& Max fidelity & SWAP-ASAP-based algorithm \\
    \hline
    Ours & Short & Yes & Yes & Yes & Max fidelity \& Max throughput & Combinatorial algorithm with DP-based separation oracle + LP rounding \& Greedy + DP \\
    \hline 
    \end{tabular}
\end{table*}

\section{Background and Assumptions}
\label{sec: background}

\subsection{Fidelity and Decoherence of Entangled Pairs}
\label{subsec: Fidelity and Decoherence curve}
Fidelity is a classic index to qualify whether an entangled pair is good. An entangled pair with fidelity $F$ can be written in a density matrix as $\rho= F\ket{\Phi^+}\bra{\Phi^+}+b\ket{\Phi^-}\bra{\Phi^-}+c\ket{\Psi^+}\bra{\Psi^+}+d\ket{\Psi^-}\bra{\Psi^-}$, where $F+b+c+d=1$ and $\ket{\Phi^+}$ is our wanted state.
In this paper, we consider our entangled pairs are Werner state \cite{van2014quantum,werner1989quantum,panigrahy2022capacity}, which has the relation $b=c=d$. 
A quantum state will interact with the environment and gradually lose its quantum properties (i.e., the fidelity decreases), called decoherence.
As $F=b=c=d=0.25$, it loses all its quantum properties and behaves like a classic random behavior \cite{van2014quantum}.
The decoherence speed depends on the type of quantum memory. The general decoherence can be described as follows:
\begin{align}
\label{deco formula}
F_d(t)= A+B{\times}e^{-(t/\mathcal{T})^{\kappa}},
\end{align}
which is an empirical formula fitting to the experimental data \cite{abobeih2018one,bradley2019ten}, as plotted in Fig. \ref{fig: decocuv}, while the unit of $t$ is second.
Note that $A$, $B$, $\mathcal{T}$, and $\kappa$ are the constants according to the adopted technique.
In addition, there is a one-to-one mapping between wait time and fidelity because $F_d(t)$ is invertible.

\begin{figure}[t]
\centering
    \subfigure[Decoherence curve]{\label{fig: decocuv}\includegraphics[width= .276\textwidth]{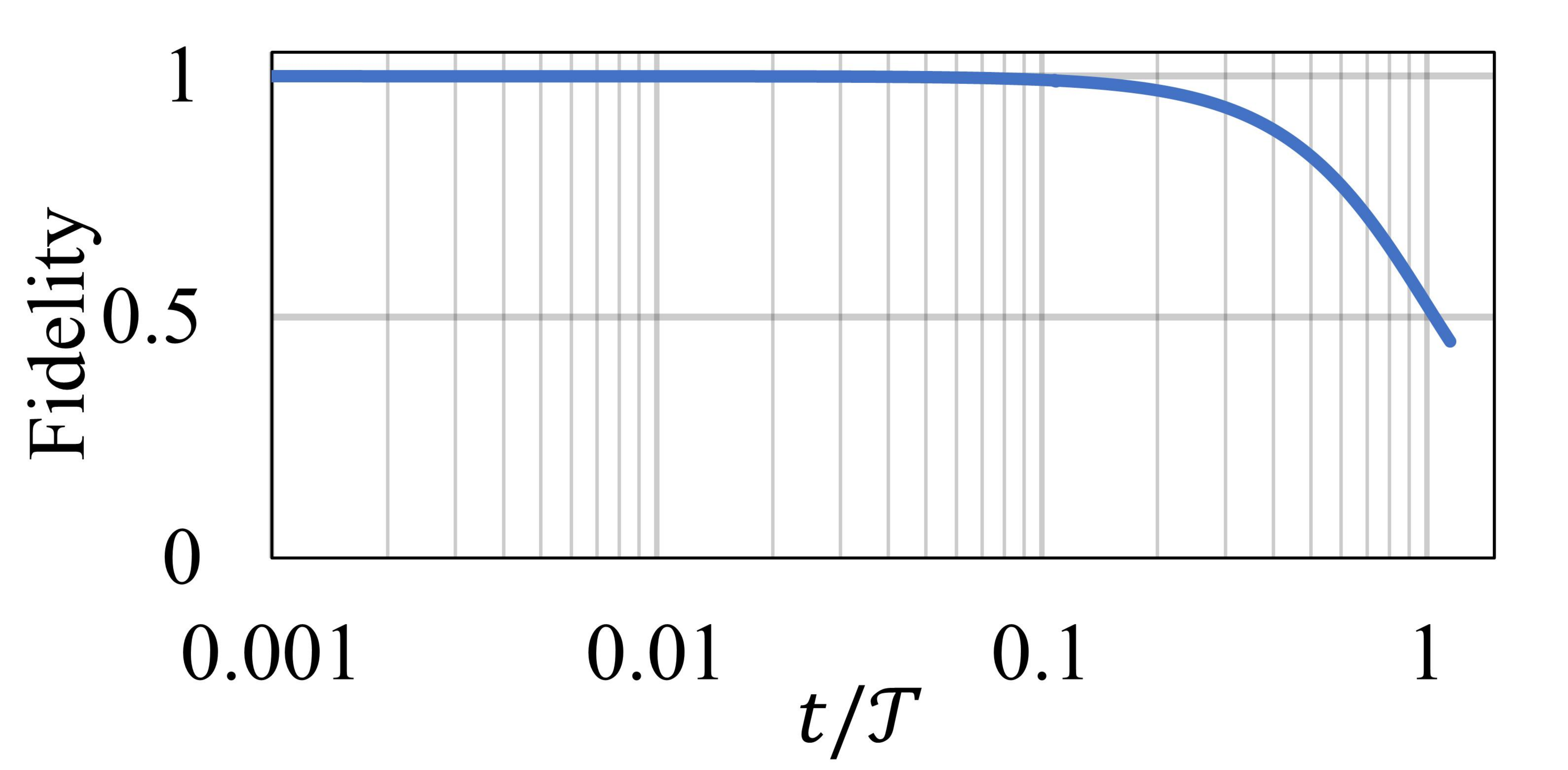}}
    \subfigure[Fidelity of swapping]{\label{fig: Fswap}\includegraphics[width= .204\textwidth]{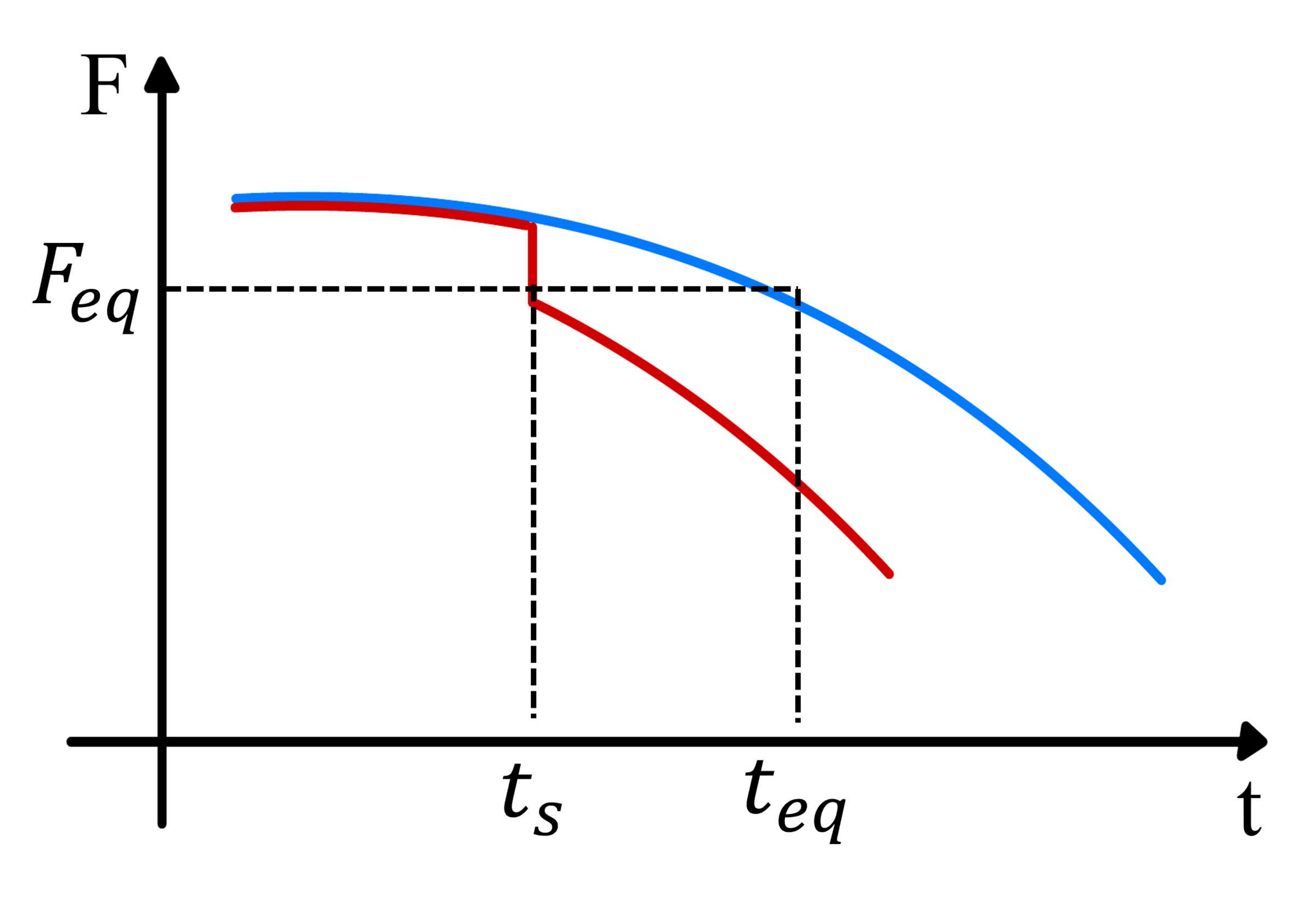}}
    \caption{Fidelity loss due to elapsed time and swapping.}
    \label{fig: Fdeco}
\end{figure}

\subsection{Entanglement Swapping}
\label{subsec: Swapping process}

Entanglement swapping causes extra fidelity loss \cite{van2014quantum}.
Swapping two Werner states with the fidelity of $F_1$ and $F_2$ will lead to the fidelity as follows:
\begin{align}
\label{swapping formula}
F_{s}(F_1,F_2)= F_1{\times}F_2+\frac{1}{3}(1-F_1)\times(1-F_2).
\end{align}
Fig. \ref{fig: Fswap} shows how the fidelity changes after swapping.
The blue curve represents the fidelity change of an entangled pair due to decoherence; on the other hand, the red curve depicts the fidelity of the entangled pair that conducts swapping at $t_s$. 
After swapping, there is a sudden drop of fidelity to $F_{eq}$, and it keeps decohering after the drop.
It can be viewed as a shift from the blue curve, i.e., the red curve after $t_s$ behaves the same as the blue curve after $t_{eq}$. 
We call $t_{eq}$ the equivalent wait time after swapping, which can be calculated by $t_{eq}=F_{d}^{-1}(F_{eq})$.

The success probability of entangling decreases exponentially with channel distance \cite{shi2020concurrent}.
Specifically, the entangling
probability between any \emph{adjacent} nodes $u$ and $v$ is defined as:
\begin{align} \label{eq: link prob}
    \Pr(u,v) = 1 - (1 - e^{-\lambda \cdot l(u,v)})^\xi,
\end{align}
where $\xi = \lfloor \frac{\tau}{\mathfrak{T}} \rfloor$ represents the number of entangling attempts, $\tau$ is the length of a short time slot, $\mathfrak{T}$ is the entangling time, $l(u,v)$ is the length of the quantum channel between $u$ and $v$, and $\lambda$ is a constant determined by the optical fiber material \cite{sangouard2011quantum}.
Then, swapping also has a different success probability at each node $v$, denoted as $\Pr(v)$ \cite{shi2020concurrent}.
Thus, the success probability of a path $p$ between the source $s$ and the destination $d$ can be expressed as follows:
\begin{align} \label{eq: path prob}
    \Pr(p) = \prod_{(u,v) \in p}\Pr(u,v) \prod_{v\in p\setminus\{s, d\}} \Pr(v).
\end{align}%
\vspace{-5mm}

\subsection{Assumptions}
\label{subsec: Assumptions}
For ease of presentation, this paper has the following assumptions:
1) Each entangled pair is described by a Werner state \cite{van2014quantum,werner1989quantum,panigrahy2022capacity}, a mixture of a specific pure state and white noise. The white noise is represented by the normalized identity operator in the Hilbert space corresponding to the state \cite{sen2005entanglement}.
2) This paper assumes that the fiber resource is always sufficient because it is relatively cheaper than quantum memory.
In other words, we focus on quantum memory consumed by the entangling process while assuming all entangling processes will not be blocked by fiber availability.
3) Following \cite{zhao2021redundant}, all nodes communicate with a central controller and share global knowledge.
4) The entangled pairs between the same pair of two nodes have identical initial fidelity. Besides, all the entangled pairs decohere as the same relation, as shown in Fig. \ref{fig: decocuv}, because of using the same quantum memory technique (i.e., constants $A$, $B$, $\mathcal{T}$, and $\kappa$ are the same for all nodes).
5) For synchronization, the time slot length should be at least the duration required for a single swapping process (e.g., around $1.1\ \textrm{ms}$ \cite{pompili2021realization}).
Since an entangling process generally requires less time than swapping (e.g., $0.25\ \textrm{ms}$ \cite{pompili2021realization}), it can repeatedly attempt entangling within a time slot to improve the entangling success probability \cite{shi2020concurrent}, as described in Eq. (\ref{eq: link prob}).

\section{System Model \& Problem Formulation}
\label{sec: problem}

\subsection{System Model}
\label{subsec: the scenario}

We consider a \ac{QN} with multiple quantum nodes, each with limited quantum memory.
Nodes connected through a bunch of fibers are adjacent nodes.
Data qubits are transmitted via end-to-end entangled pairs in \ac{QN}.
A one-hop entangled pair can be directly constructed via the fiber between two adjacent nodes $u$ and $v$ with the success probability of $\Pr(u,v)$ and the initial fidelity of $F(u,v)$.
Plus, a long entangled pair $(u,w)$ can be constructed by performing the swapping process at a repeater $v$ to consume two short entangled pairs $(u,v)$ and $(v,w)$ with the success probability of $\Pr(v)$.
The quality of an entangled pair is measured by fidelity. 
A generated entangled pair will decohere over time, as defined by Eq. (\ref{deco formula}) in Section \ref{subsec: Fidelity and Decoherence curve}.

A node can perform an entanglement swapping process when it possesses two qubits, each from different entangled pairs.
Subsequently, the memory occupied by these two qubits after swapping will be freed. 
The fidelity of the resulting longer entangled pair after swapping follows Eq. (\ref{swapping formula}). 

We consider a short time slot protocol where a single time slot only accommodates a single process for each memory unit.
Each node can freely execute one process within a time slot for each memory unit, such as entangling, swapping, or even idling.
The protocol periodically performs all necessary computations in advance, producing a complete operation plan before any entangling or swapping begins. The controller then instructs all quantum nodes to synchronously execute these operations within a fixed batch of time slots.
Any request that fails within a batch is deferred to the next batch for reattempt.
For clarity, let $T$ be the set of time slots in the current batch.

To consider the resource distribution, we define the strategy tree, numerology, and fidelity formulas in our system.
\begin{definition}
\label{defi: Strategy Tree}
A strategy tree $\gamma$ is an activity-on-vertex tree structure, describing an entangling and swapping strategy between node $v_1$ and node $v_n$ on path 
$p=\{v_1,v_2,...,v_n\}$.
It consists of two parts; one is the entangling part ($\gamma_e$), and the other is the swapping part ($\gamma_s$), as shown by the green and blue lines in Fig. \ref{fig: treeskew}, respectively. 
The entangling part $\gamma_e$ includes all the external pairs. 
Each external pair denotes a quantum pair $(v_{i}, v_{i+1})$ being entangled, where $i\in \{1, 2, ..., n-1\}$, as shown by the green pair in Fig. \ref{fig: treeskew}.
On the other hand, the swapping part $\gamma_s$ is an edge-weighted binary tree with a root pair $(v_1, v_n)$, where
each leaf pair represents an entangled pair $(v_i,v_{i+1})$ for each $i\in \{1,2,...,n-1\}$ and each edge's weight denotes the number of elapsed time slots for the corresponding entangling or swapping processes.
Besides, each non-leaf pair $(v_i,v_j)$ has exactly two child pairs, $(v_i, v_k)$ and $(v_k, v_j)$, where $i< k< j$, denoting a long pair created by consuming two pairs.
Last, each leaf pair $(v_i,v_{i+1})$ in $\gamma_s$ is connected to a corresponding external pair $(v_i,v_{i+1})$ in $\gamma_e$ by an edge with a cost to form $\gamma$.
\end{definition}

For a given strategy tree $\gamma$, each edge in $\gamma$ has a positive cost (i.e., the number of elapsed time slots for the corresponding entangling or swapping processes). Let $\rho_{r}$ denote the root pair of the strategy tree $\gamma$, and $\rho_{a}$ denote any pair in $\gamma$. We define the cost sum $\delta(\rho_{r}, \rho_{a})$ as the total cost along the simple path from $\rho_{a}$ to $\rho_{r}$ in $\gamma$, i.e., the total number of elapsed time slots from $\rho_{a}$ to $\rho_{r}$.
Following the definition above, for a given strategy tree $\gamma$, the root pair $\rho_{r}$ is assigned a designated time slot $t_r \in T$, where $T = \{1, 2, \dots, |T|\}$ represents the set of time slots in the current batch. Therefore, the corresponding time slot for a pair $\rho_a$ in $\gamma$ can be expressed as $t_a = t_r - \delta(\rho_{r}, \rho_{a})$. 
A strategy tree $\gamma$ is \emph{feasible} if satisfying the condition:
The time slot $t_a$ of each pair $\rho_a\in \gamma$ is within $T$ (i.e., $t_a\in T$). This condition is essential; otherwise, $t_a$ may fall outside the set $T$.

\begin{figure}[t]
\centering
    \subfigure[Three possible strategy trees for the path $v_1 \rightarrow v_2 \rightarrow v_3 \rightarrow v_4$]{\label{fig: strategy tree}\includegraphics[width= .48\textwidth]{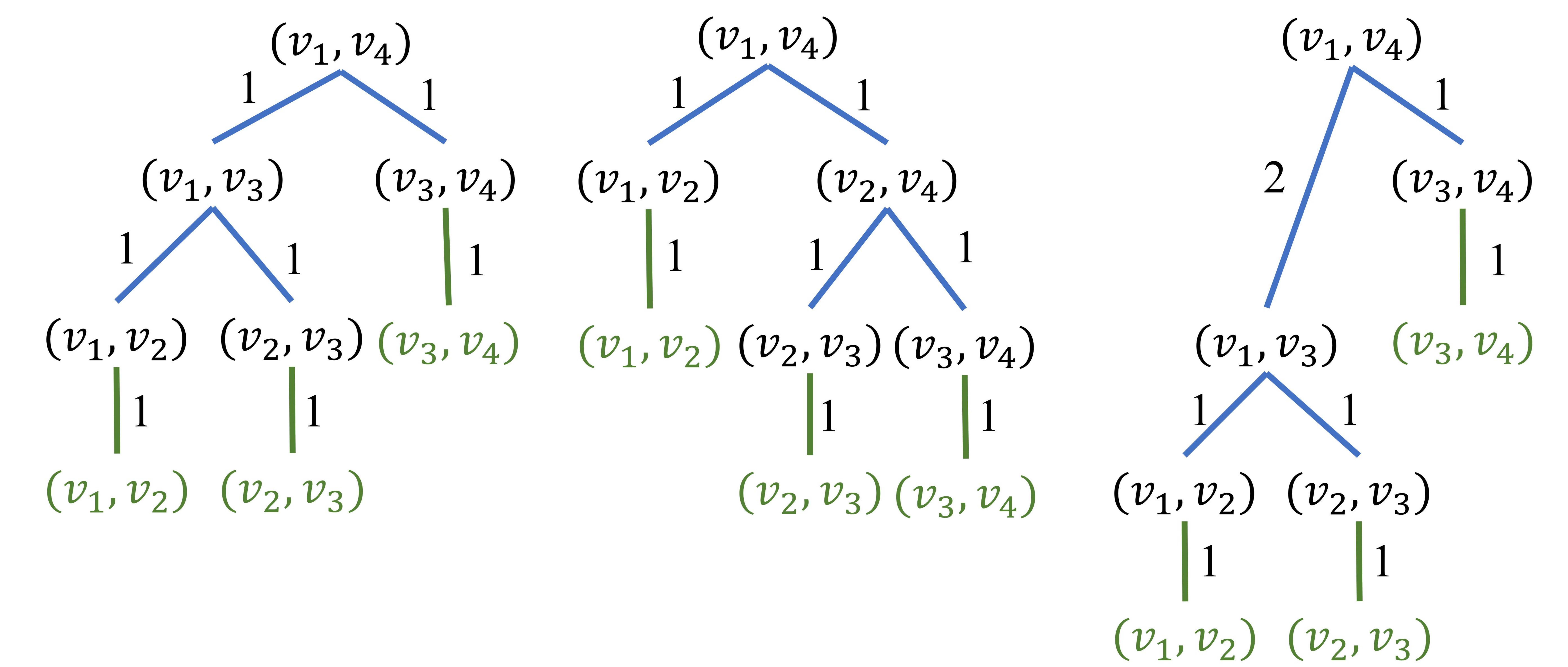}}
    \subfigure[The corresponding numerologies for the strategy trees in (a)]{\label{fig: numerology}\includegraphics[width= .48\textwidth]{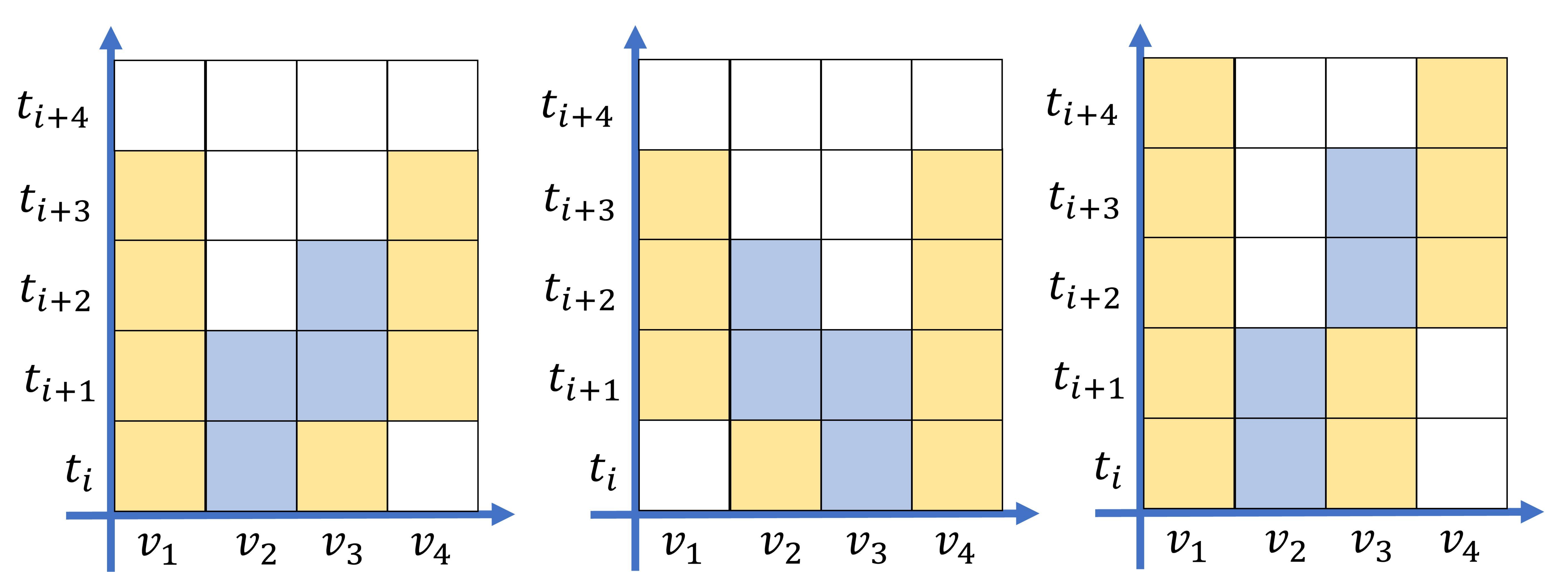}}
    \caption{Numerologies induced by different strategy trees.}
    \label{fig: nume}
\end{figure}

\begin{figure*}[t]
\centering
    \subfigure[Scanning numerology from time slot $t_5$ to $t_4$]{\label{fig: prof1}\includegraphics[width= 0.3577\textwidth]{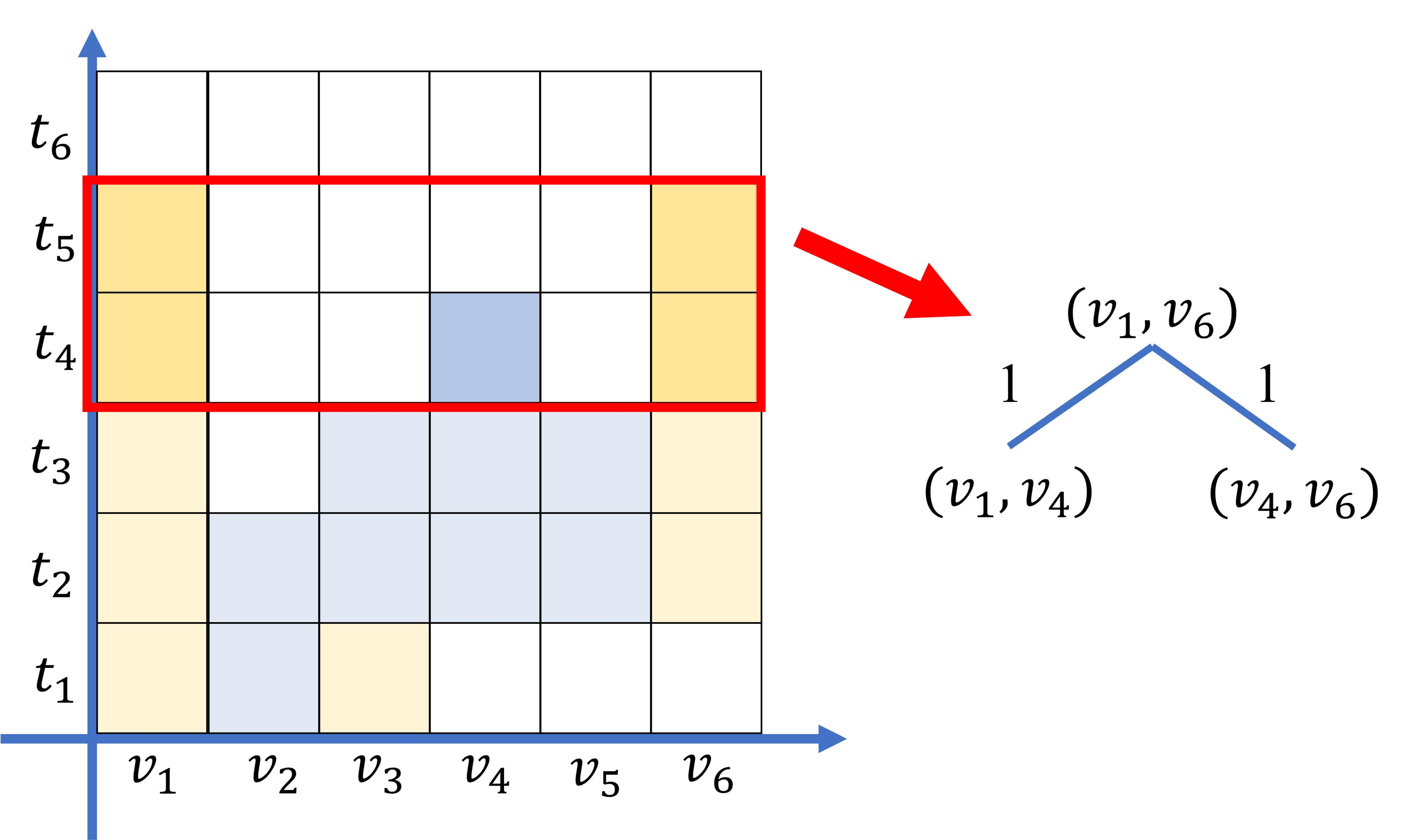}}
    \hfill
    \subfigure[Scanning numerology from time slot $t_4$ to $t_3$]{\label{fig: prof2}\includegraphics[width= 0.5918\textwidth]{graph/prof2.pdf}}
    \caption{An illustrating example of transforming a numerology into a strategy tree.}
    \label{fig: oneoneprof}
\end{figure*}

For example, consider the right strategy tree in Fig. \ref{fig: strategy tree}. Suppose the root pair $\rho_{r} = (v_1, v_4)$ is with $t_r = 3$. Now, taking the external pair $(v_1, v_2)$ as $\rho_{a}$, the corresponding time slot $t_a$ is calculated as $t_a = 3 - \delta(\rho_{r}, \rho_{a}) = 3 - 4 = -1$, which means $t_a \notin T$.
Thus, if $t_r = 3$, the right strategy tree is infeasible.

\begin{definition}
\label{defi: numerology}
A numerology $m$ represents the resource distribution of a specific strategy tree $\gamma$. Specifically, let $\rho_p = (v_i, v_j)$ denote a pair at time slot $t_p$ in $\gamma$, with a left child pair $\rho_{c_l} = (v_i, v_k)$ at time slot $t_{c_l}$ and a right child pair $\rho_{c_r} = (v_k, v_j)$ at time slot $t_{c_r}$, where $t_p > t_{c_l}$ and $t_p > t_{c_r}$. Consequently, the left child pair $\rho_{c_l}$ occupies one memory unit of $v_i$ and $v_k$ during time slots $T(\rho_{c_l}) = \{t\in \mathbb{Z} \mid t_{c_l} \le t < t_p\}$, while the right child pair $\rho_{c_r}$ occupies one memory unit of $v_k$ and $v_j$ during time slots $T(\rho_{c_r}) = \{t \in \mathbb{Z} \mid t_{c_r} \le t < t_p\}$.
Furthermore, if $\rho_p$ has only a child pair $\rho_c=(v_i, v_j)$ (i.e., the external pair in $\gamma$) at time slot $t_c$, this child pair occupies one memory unit of $v_i$ and $v_j$ during time slots $T(\rho_{c}) = \{t \in \mathbb{Z} \mid t_{c} \le t < t_p\}$.
Plus, $T(\rho_{r}) = \{ t_r\}$ for the root pair $\rho_{r}$ at time slot $t_r$ in $\gamma$. 
Therefore, numerology $m$ captures the resource distribution across the nodes in the strategy tree $\gamma$ and can be formally defined as $m=\{(\rho_{a}, T(\rho_{a})) \mid \rho_{a}\in\gamma\}$, where each pair $\rho_a\in\gamma$ and its associated time slots are grouped in a tuple $(\rho_a, T(\rho_a))$.
The amount of memory on node $v$ occupied by numerology $m$ at time slot $t$ is $\theta_{m}(t,v) = \big|\{(\rho_{a}, T(\rho_{a}))\in m \mid v\in \rho_{a}, t\in T(\rho_{a})\}\big|$. Note that $\theta_{m}(t,v) \in \{0, 1, 2\}$.
\end{definition}

We further illustrate the example in Fig. \ref{fig: nume}.
Each strategy tree in Fig. \ref{fig: strategy tree} corresponds to a distinct numerology shown in Fig. \ref{fig: numerology}.
In this figure, yellow squares indicate a node using one memory unit during that time slot, while blue squares represent a node using two memory units.
For the numerology $m$ on the right in Fig. \ref{fig: numerology}, we observe that $\theta_{m}(t_{i+1},v_{2})$ equals $2$, as node $v_2$ simultaneously uses one memory unit for each of the pairs $(v_1, v_2)$ and $(v_2, v_3)$ during time slot $t_{i+1}$.
Additionally, it is worth noting that for any given path, feasible strategy trees have a one-to-one and onto mapping to feasible numerologies.

\begin{figure}[t]
\centering
    \includegraphics[width= .48\textwidth]{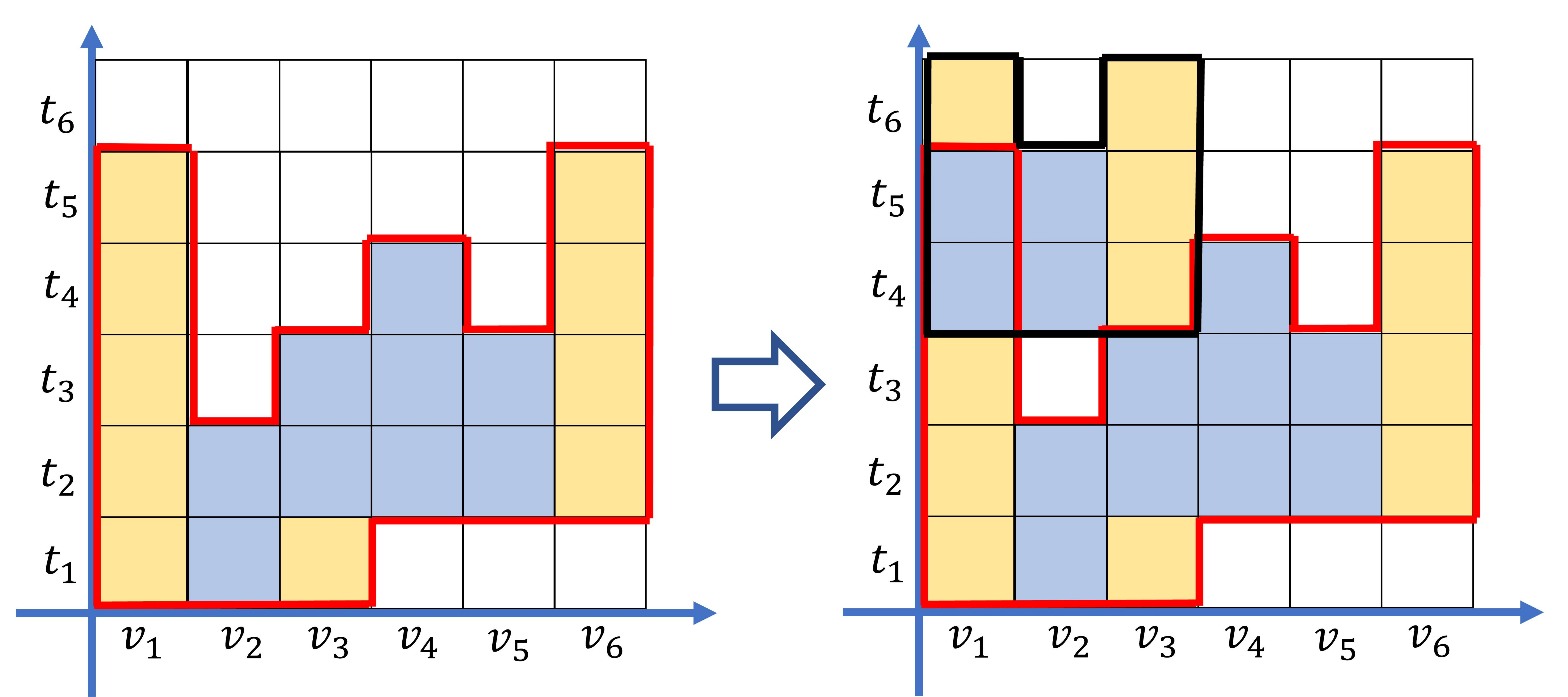}
    \caption{Overlapping numerologies for two accepted requests.}
    \label{fig: puz}
\end{figure}

\begin{lemma}
For any given path $p$, there exists a one-to-one and onto mapping $f$, which maps a feasible strategy tree $\gamma$ to a feasible numerology $m$, while $f^{-1}$ denotes the inverse mapping, i.e., $f(\gamma)=m$ and $f^{-1}(m)=\gamma$.
\end{lemma}%
\begin{proof}
    Definition \ref{defi: numerology} indicates that each strategy tree has a corresponding numerology, i.e., $f(\gamma)=m$.
    Then, function $f$ is \emph{onto} since only the numerologies derived from feasible strategy trees are feasible.
    It suffices to show that for any two different strategy trees $\gamma_1$ and $\gamma_2$, their corresponding numerologies $m_1$ and $m_2$ are different, i.e., $m_1=f(\gamma_1)\neq f(\gamma_2)=m_2$.
    We prove it by contradiction. Assume that there exist two different strategy trees mapping to the same numerology for the given path $p=\{v_1, v_2, v_3, ..., v_n\}$, i.e., $f(\gamma_1)=f(\gamma_2)=m$. 
    Let $t_r$ denote the time slot of root pair $r$ in $\gamma_1$.
    Therefore, $m$ uses one memory unit of nodes $v_1$ and $v_n$ at time $t_r$, respectively, denoted by $(v_1,v_n)$ in the strategy tree. By Definition \ref{defi: Strategy Tree}, we know that a non-leaf node $(v_i,v_j)$ in the swapping part must have exactly two children $(v_i,v_k)$ and $(v_k,v_j)$, where $i<k<j$. 
    This means that if we start scanning $m$ from time slot $t_r$ in decreasing order, we can find exactly one node $k$ which spends two memory units at some time slot.
    Take Fig. \ref{fig: prof1} as an example. We start scanning from root pair $(v_1,v_6)$ at time $t_5$.
    Then, when reaching $t_4$, we can find node $v_4$ with two occupied memory units, implying two children $(v_1,v_4)$ and $(v_4,v_6)$.
    Similarly, we keep scanning the remaining time slots from $t_4$ to $t_3$.
    Note that the resource distribution at $t_4$ is the combination of resource distributions of pairs $(v_1,v_4)$ and $(v_4,v_6)$, each of which also meets Definition \ref{defi: Strategy Tree}, as shown in Fig. \ref{fig: prof2}.
    Thus, we can divide $m$ into two sub-distributions from $v_4$, find their children individually, and construct the sub-strategy tree recursively.
    By doing so, we can reconstruct only one possible strategy tree, implying a contradiction. Then, $f$ is a one-to-one and onto mapping, and the lemma follows.
\end{proof}%

\begin{definition}
In the short time slot protocol, with Eq. (\ref{deco formula}), the fidelity of an entangled pair after one-time-slot idling becomes 
    \begin{align}
    \label{eq: deco over time}
    F^{\tau}(F) = F_{d}(F_{d}^{-1}(F)+ \tau),
    \end{align}
where $\tau$ is the actual duration (i.e., length) of one time slot and $F$ is the fidelity before decay.
As two short entangled pairs with fidelity $F_1$ and $F_2$ conduct swapping, with Eqs. (\ref{swapping formula}) and (\ref{eq: deco over time}), the fidelity of the long entangled pair becomes%
    \begin{align} \label{eq: swap deco over time}
    F^{\tau}_s(F_1,F_2) = F_s(F^{\tau}(F_1),F^{\tau}(F_2)).
    \end{align}
\end{definition}

We continue to demonstrate the example in Fig. \ref{fig: nume}.
The left strategy tree in Fig. \ref{fig: strategy tree} (or the left numerology in Fig. \ref{fig: numerology}) and its mirror strategy tree (i.e., the middle case) will have the same fidelity if all the entangled pairs have the same initial fidelity.
The right case is similar to the left one but useful as $v_2$ and $v_4$ have no available memory at $t_{i+2}$ and $t_{i+1}$, respectively (i.e., the left and middle cases cannot be chosen). 
However, it costs more resources and has lower fidelity than the left case.
Note that requests' numerologies could overlap, as shown in Fig. \ref{fig: puz}, where two requests from $v_1$ to $v_6$ and from $v_1$ to $v_3$ are denoted by the red- and black-frame numerologies.

As the numerology indicates the strategy tree in the form of quantum resource distribution, the \ac{QN} strategy tree selection problem can map to a quantum resource allocation problem.

\begin{figure*}[t]
\centering
    \subfigure[\ac{MIS} instance $\mathcal{G}$]{\label{fig: NPH1}\includegraphics[width= .1514\textwidth]{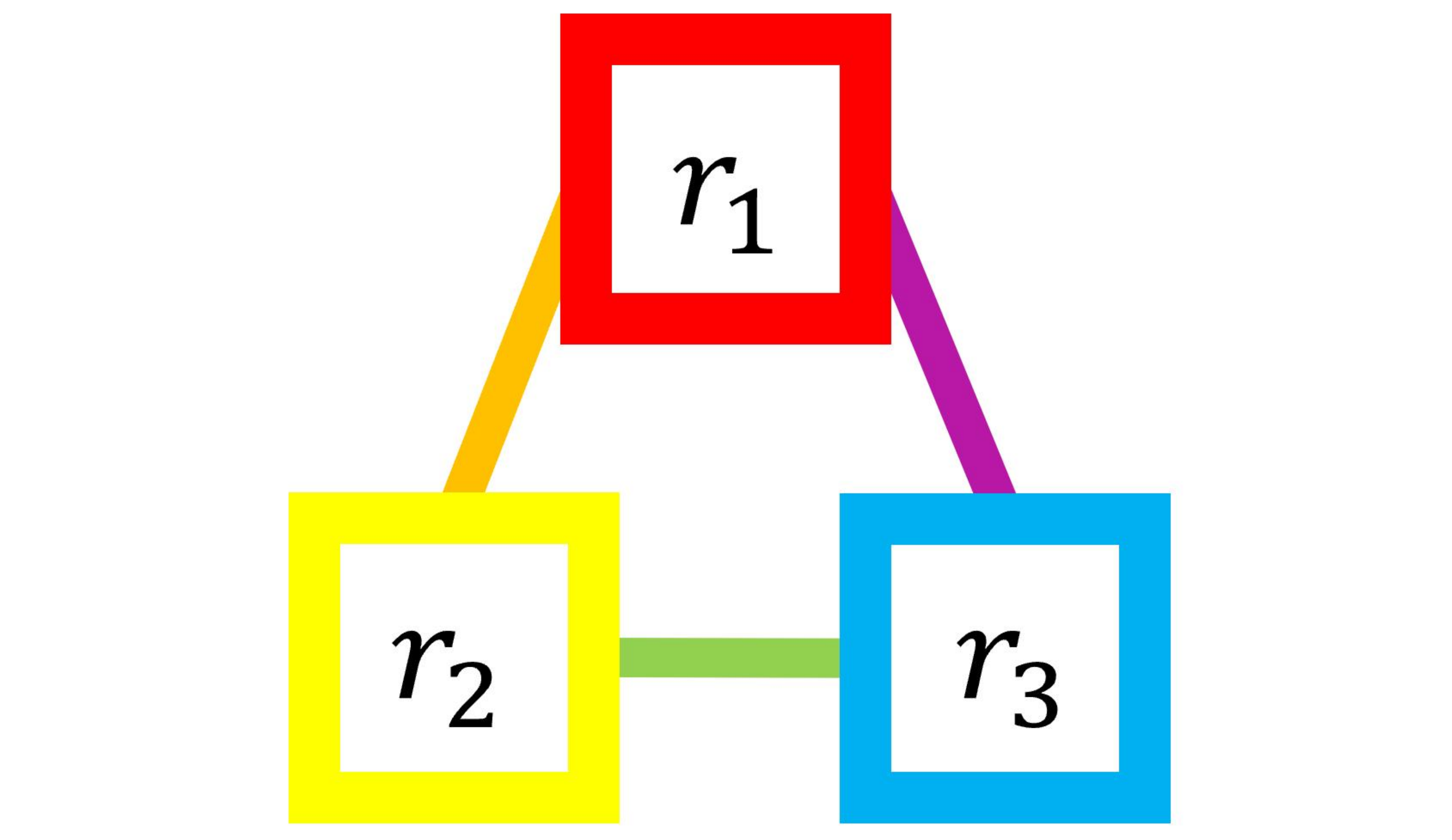}}
    \hfill
    \subfigure[Path construction for $G$]{\label{fig: NPH2}\includegraphics[width= 0.2142\textwidth]{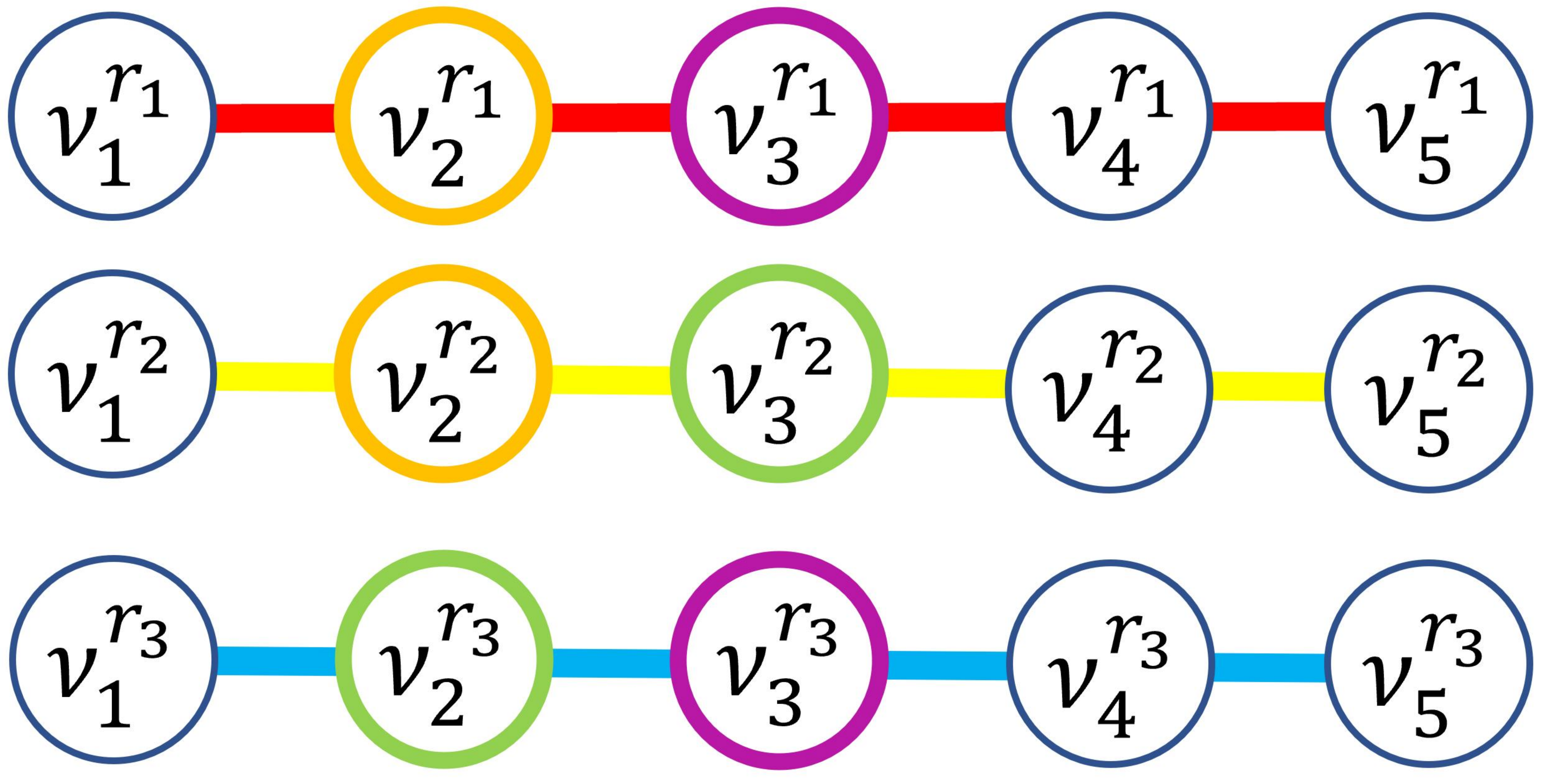}}
    \subfigure[{\ProblemNameAbbrN} instance $G$]{\label{fig: NPH3}\includegraphics[width= .2142\textwidth]{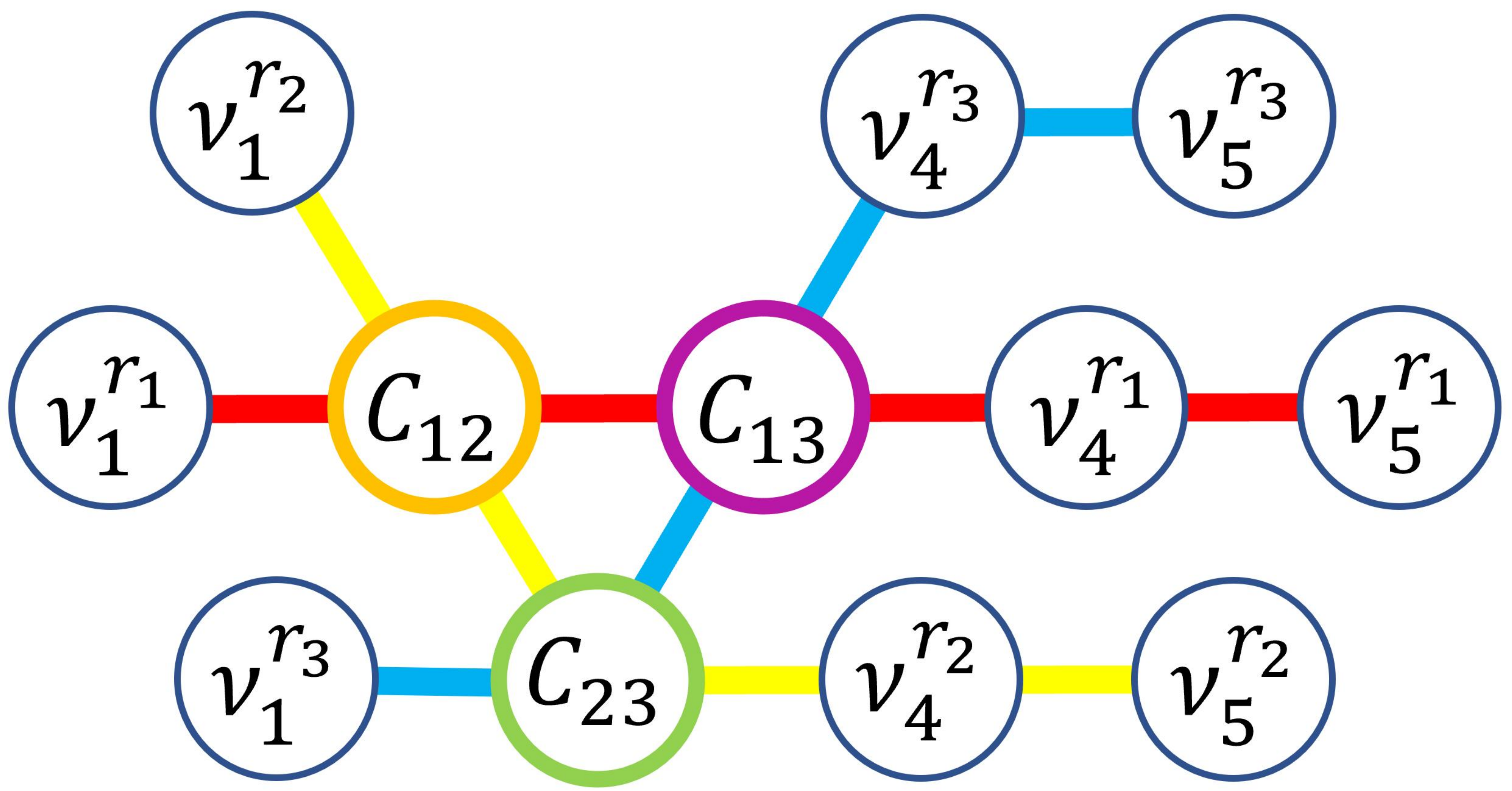}}
    \subfigure[Numerology of $r_1$]{\label{fig: NPH4}\includegraphics[width= .1525\textwidth]{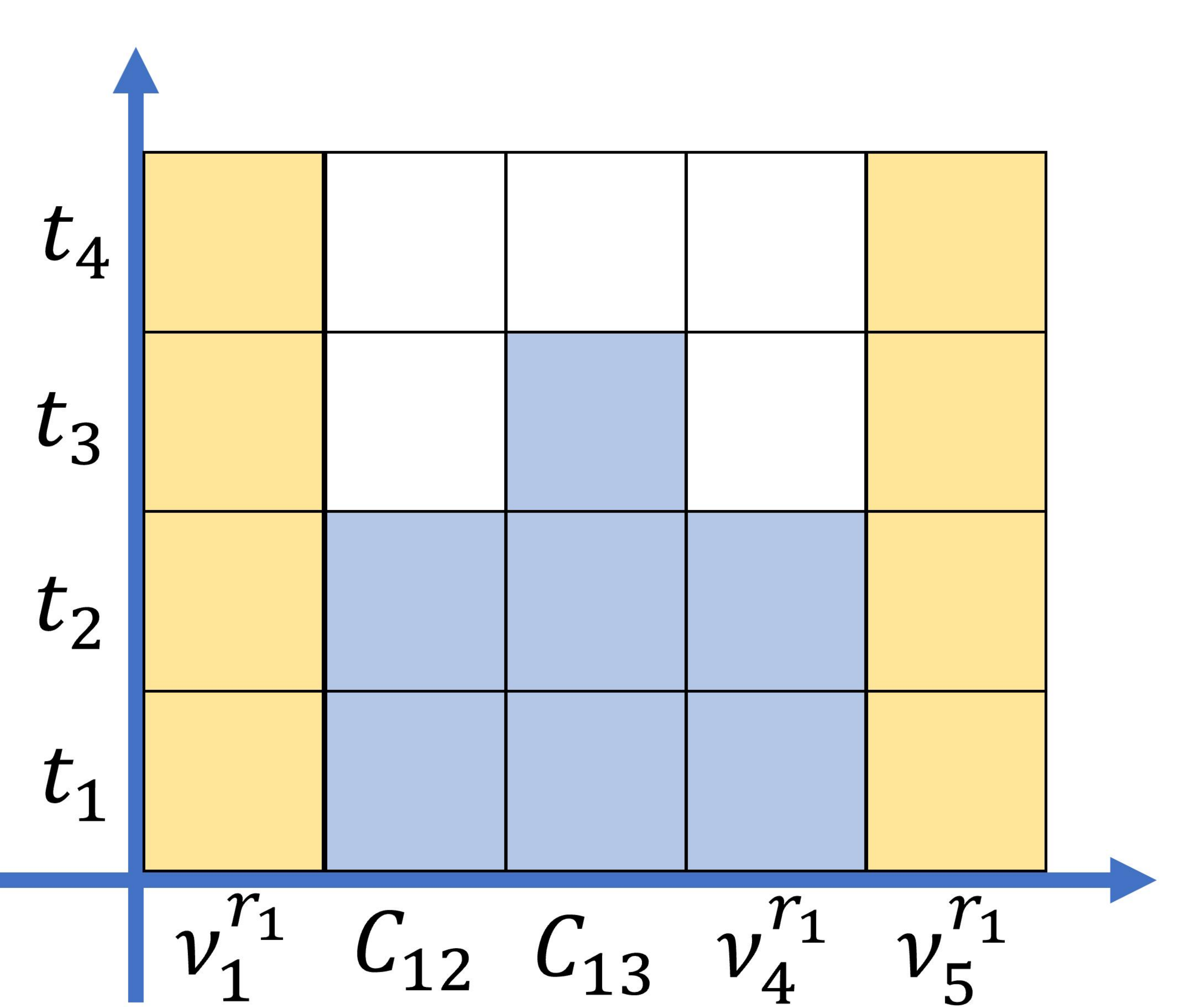}}
    \subfigure[Resource contention for $r_2$]{\label{fig: NPH5}\includegraphics[width= .228\textwidth]{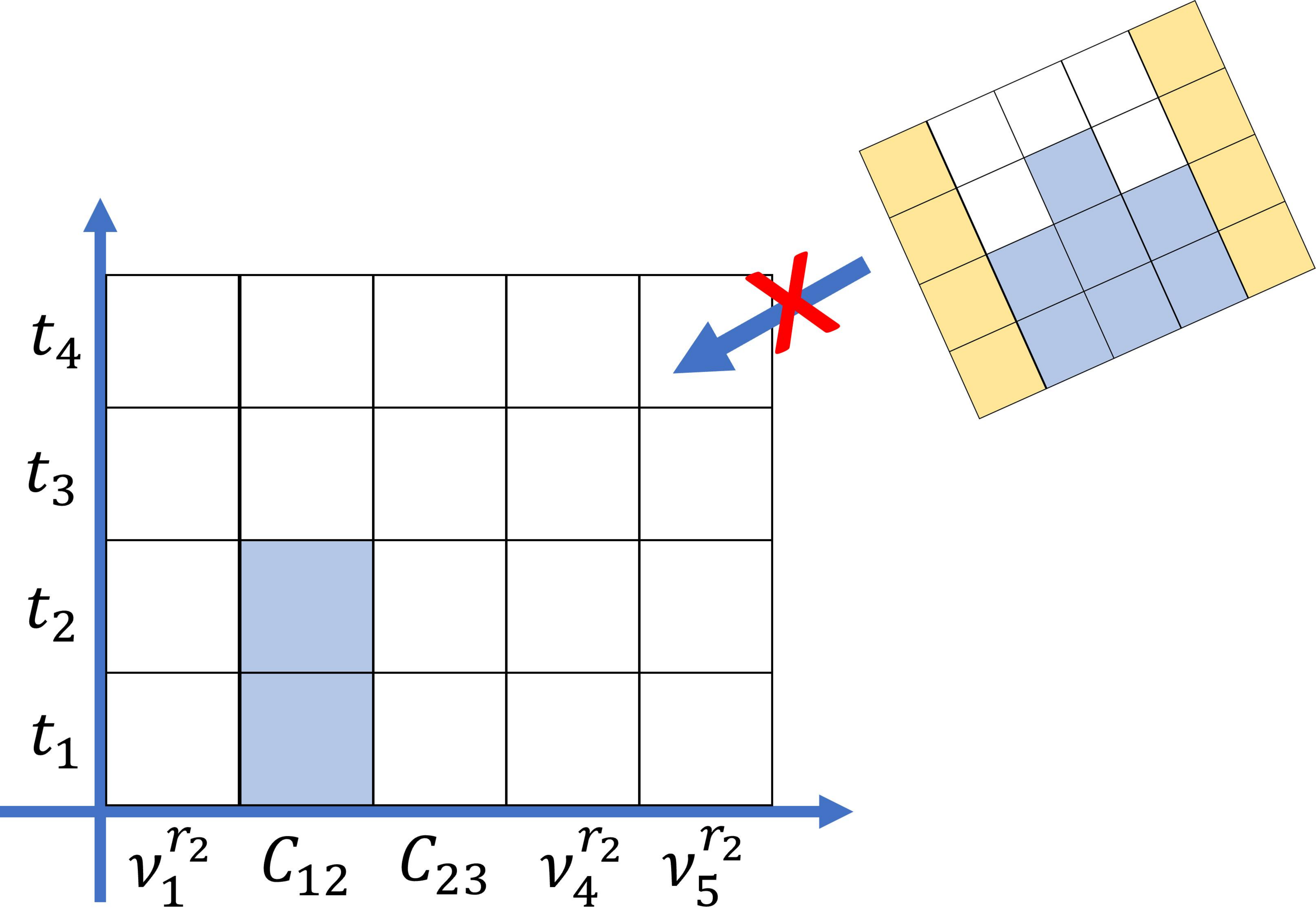}}
    \caption{An example of the reduction from the \ac{MIS} to the {\ProblemNameAbbr}.}
    \label{fig: NPH}
 \end{figure*}

\subsection{Problem Formulation}
\label{subsec: the problem}

We formulate the problem based on the above scenario.
\begin{definition} \label{defi: problem} 
Consider a network $G = (V, E)$, where each node $v \in V$ has a memory limit $c^t(v) \in \mathbb{Z}^+$ for each time slot $t \in T = \{1, 2, ..., |T|\}$. The network also includes a set of \ac{SD} pairs $I$. Each pair $i \in I$ has a predefined path set $P(i)$ derived by an unspecified routing algorithm.%
    \footnote{There are many algorithms for finding routing paths for \ac{SD} pairs \cite{pant2019routing,shi2020concurrent,zhao2021redundant} to derive the predefined path set $P(i)$ for each pair $i\in I$.\label{footnote: routing algorithm}}
Each path $p\in P(i)$ has the success probability $\Pr(p)$.
Given the fidelity threshold $\widehat{F}$, 
the {\ProblemName} (\ProblemNameAbbr) aims to maximize the expected fidelity sum of the chosen numerology for each \ac{SD} pair, subject to the following constraints.
\begin{enumerate}
    \item At most, a numerology associated with one path from the path set $P(i)$ can be selected for each \ac{SD} pair.\footref{footnote: routing algorithm}
    \item The amount of occupied memory on each node $v$ does not exceed its memory limit $c^t(v)$ at each time slot $t$.
    \item The selected numerology's fidelity should be no less than the threshold $\widehat{F}$ to ensure high-quality teleportation.
\end{enumerate}
\end{definition}

Let $M_p(i)$ denote the set of all numerologies that can be implemented from the path $p \in P(i)$ within $T$ (i.e., the period from time slot 1 to time slot $|T|$) for an \ac{SD} pair $i\in I$, and let $F(m)$ denote the fidelity of an entangled pair constructed with numerology $m \in M_p(i)$.
Following Definition \ref{defi: numerology}, let $\theta_m(t, v) \in \{0, 1, 2\}$ denote the required amount of memory on node $v\in V$ at time slot $t \in T$ when implementing a numerology $m\in M_p(i)$ for \ac{SD} pair $i\in I$.
In this way, {\ProblemNameAbbr} can be formulated as the following \ac{ILP}, where the binary variable $x^{ip}_m$ indicates whether a numerology $m \in M_p(i)$ is chosen ($x^{ip}_m=1$) or not ($x^{ip}_m =0$) for \ac{SD} pair $i\in I$ and path $p\in P(i)$.
\begin{subequations} \label{eq: primal ILP}
\begin{align}
    \label{eq: primal objective}
    \text{maximize } 
    & \sum_{i\in I} \sum_{p\in P(i)} \sum_{m\in M_p(i)} \Pr(p) \cdot F(m) \cdot x^{ip}_m
    \\
    \label{eq: only select one constraint}
    \text{subject to }
    & \sum_{p\in P(i)} \sum_{m\in M_p(i)} x^{ip}_m \le 1, \quad \forall i \in I
    \\
    \begin{split}
        \label{eq: memory constraint}
        & \sum_{i\in I} \sum_{p\in P(i)} \sum_{m\in M_p(i)} \theta_m(t, v) \cdot x^{ip}_m \le c^t(v), \\ 
        & \qquad\qquad\qquad\qquad\qquad\;\; \forall v \in V, \forall t \in T
    \end{split}\\
    \begin{split}
        \label{eq: fidelity threshold constraint}
        & ( F(m) - \widehat{F} ) \cdot x^{ip}_m \geq 0, \\
        & \qquad\qquad\; \forall i \in I, \forall p\in P(i), \forall m \in M_p(i)
    \end{split}\\
    \begin{split}
        \label{eq: x range constraint}
        & x^{ip}_m \in \{0, 1\}, \\
        & \qquad\qquad\; \forall i \in I, \forall p\in P(i), \forall m \in M_p(i)
    \end{split}
\end{align}
\label{straincomponent}%
\end{subequations}
The objective (\ref{eq: primal objective}) maximizes the expected fidelity sum of all \ac{SD} pairs, where $\Pr(p)$ is the success probability of path $p$. Constraint (\ref{eq: only select one constraint}) ensures that at most one numerology can be chosen from a single
path $p$ within the path set $P(i)$ for each \ac{SD} pair $i\in I$.
Constraint (\ref{eq: memory constraint}) guarantees that the total amount of memory occupied on each node $v\in V$ at every time slot $t\in T$ does not exceed its memory limit.
Constraint (\ref{eq: fidelity threshold constraint}) ensures that the selected numerology fidelity should be no less than the threshold $\widehat{F}$ the controller sets based on its policy.

\subsection{NP-hardness and Inapproximability}

We prove the hardness of {\ProblemNameAbbr} by demonstrating that even without the constraint (\ref{eq: fidelity threshold constraint}), the problem, termed {\ProblemNameAbbrN}, is very challenging. 
Specifically, we prove the NP-hardness and inapproximability of {\ProblemNameAbbrN} in Theorem \ref{theo: np-hard} by reducing the well-known NP-hard problem, the \ac{MIS} \cite{MIS-hard}, to the {\ProblemNameAbbrN}.
Therefore, the {\ProblemNameAbbr}, as a general case of the {\ProblemNameAbbrN}, must be at least as hard as {\ProblemNameAbbrN}, thereby implying Corollary \ref{coro: np-hard}.
Note that the \ac{MIS} asks for a maximum set of vertices in a graph, no two of which are adjacent.

\begin{theorem} \label{theo: np-hard}
    The {\ProblemNameAbbrN} is NP-hard and cannot be approximated within any factor of $|I|^{1-\epsilon}$ unless $NP = P$ for any fixed $\epsilon>0$, where $|I|$ is the number of \ac{SD} pairs.
\end{theorem}
\begin{proof}
    The proof idea is to create an \ac{SD} pair and a dedicated path and then add them to the {\ProblemNameAbbrN} instance for each node in the \ac{MIS} instance such that the corresponding \ac{SD} pairs of any two neighboring nodes in the \ac{MIS} instance cannot be satisfied simultaneously within $T$.
    In the following, for any given \ac{MIS} instance $\{\mathcal{G}=\{\mathcal{V}, \mathcal{E}\}\}$, we show how to construct the corresponding {\ProblemNameAbbrN} instance in detail.
    
    For each node in $\mathcal{V}$, we create a corresponding \ac{SD} pair and add it to the {\ProblemNameAbbrN} instance.
    Subsequently, for each created \ac{SD} pair, we construct a dedicated path that consists of $ 2^{\lceil \log |\mathcal{V}|\rceil}+1$ nodes with $2^{\lceil \log |\mathcal{V}|\rceil}$ edges and then add the path to $G$ of the {\ProblemNameAbbrN} instance.
    Afterward, $|T|$ is set to $\lceil \log |\mathcal{V}|\rceil+2$.
    In addition, the amounts of memory of the source and destination on the path are set to $1$, and those of the other nodes (i.e., intermediate nodes) on the path are set to $2$.
    Then, the probability of each path is uniformly set to $1$ (i.e., $\Pr(p)=1$). 
    Last, for each pair of neighboring nodes in $\mathcal{G}$, we pick an intermediate node that has not been picked yet from each of their corresponding paths in $G$ and then merge the two picked nodes into a single node. Note that the node induced by merging cannot be picked for merging again.
    The construction can be done in polynomial time.

    Fig. \ref{fig: NPH} illustrates an example of instance construction. The \ac{MIS} instance $\mathcal{G}$ includes a clique with three nodes, $r_1$, $r_2$, and $r_3$, which are drawn in the red, yellow, and blue frames, respectively. 
    Then, for each node in $\mathcal{G}$, we create an \ac{SD} pair and construct a dedicated path that consists of $2^{\lceil \log 3\rceil}+1 =5$ nodes, as shown in Fig. \ref{fig: NPH2}.
    Note that the color of the edges on each path corresponds to the color of relevant node in $\mathcal{G}$, and thus the path $\{v^{r_1}_1, v^{r_1}_2, ...,v^{r_1}_5\}$ of \ac{SD} pair $r_1$ in $G$ represents the node $r_1$ in $\mathcal{G}$. 
    In $\mathcal{G}$, an orange edge is between $r_1$ and $r_2$. Therefore, $v^{r_1}_2$ and $v^{r_2}_2$ are selected from the two corresponding paths, respectively, and merged into an orange node denoted by $C_{12}$ in $G$, as shown in Fig. \ref{fig: NPH3}.
    Similarly, $C_{23}$ and $C_{13}$ are derived due to the edges $(r_2,r_3)$ and $(r_1,r_3)$, respectively.
    
    We then show that the solution of the \ac{MIS} instance is one-to-one mapping to that of the constructed {\ProblemNameAbbrN} instance. 
    In this {\ProblemNameAbbrN} instance, every \ac{SD} pair has only one path with scarce memory, and its candidate strategy trees include only the full binary tree due to the subtle setting of $|T|$, as shown in Fig. \ref{fig: NPH4}, where $|T|$ is set to $\lceil \log3\rceil+2 = 4$.
    Therefore, for any two neighboring nodes in $\mathcal{G}$, the numerologies of the two corresponding pairs in the {\ProblemNameAbbrN} instance cannot be selected at the same time since the merged node on the two paths in $G$ has no sufficient memory to serve the two pairs within $T$.
    In contrast, they can be selected simultaneously if the two corresponding nodes in $\mathcal{G}$ are not adjacent.
    Figs. \ref{fig: NPH4} and \ref{fig: NPH5} continue to show the example. 
    Assume that \ac{SD} pair $r_1$ is admitted to construct numerology on path $\{v^{r_1}_1, \mathcal{C}_{12}, \mathcal{C}_{13}, v^{r_1}_4, v^{r_1}_5\}$.
    In such a case, $r_2$ cannot construct the numerology on the path $\{v^{r_2}_1, \mathcal{C}_{12}, \mathcal{C}_{23}, v^{r_2}_4, v^{r_2}_5\}$ since the memory of $\mathcal{C}_{12}$ has been occupied by $r_1$ (i.e., resource contention), as shown in Fig. \ref{fig: NPH5}.
    Thus, the solutions of any \ac{MIS} instance and its corresponding {\ProblemNameAbbrN} instance are one-to-one mapping. Thereby, {\ProblemNameAbbrN} is NP-hard.
    
    We continue to show that the {\ProblemNameAbbrN} does not admit any approximation algorithm within a factor of $|I|^{1-\epsilon}$ for any fixed $\epsilon>0$ by contradiction. To this end, suppose there exists an $|I|^{1-\epsilon}$-approximation algorithm $A$ for the {\ProblemNameAbbrN}.
    Following the above-mentioned instance construction, any arbitrary \ac{MIS} instance has a corresponding {\ProblemNameAbbrN} instance.
    Assume the optimal solution for the \ac{MIS} instance has $k$ nodes, implying the optimal solution for the {\ProblemNameAbbrN} instance maximizes the expected fidelity sum by satisfying the corresponding $k$ pairs.
    Then, we can employ algorithm $A$ to find a solution that satisfies at least $\frac{k}{|I|^{1-\epsilon}}$ pairs, and the found solution corresponds to a solution with at least $\frac{k}{n^{1-\epsilon}}$ nodes for the \ac{MIS} instance, where $n$ denotes the number of nodes in the corresponding \ac{MIS} instance. This is because $|I|$ is set to $n$ during the instance construction. 
    However, unless $NP = P$, the \ac{MIS} does not admit any approximation ratio within $n^{1-\epsilon}$ for any fixed $\epsilon>0$ \cite{MIS-hard}.
    Thus, algorithm $A$ must not exist; otherwise,
    $A$ could be employed to derive an $n^{1-\epsilon}$-approximation algorithm for the \ac{MIS}. The theorem follows.
\end{proof}
\begin{corollary} \label{coro: np-hard}
    The {\ProblemNameAbbr} is at least as hard as the {\ProblemNameAbbrN}.
\end{corollary}%

\section{The Design of {\AlgoNametwo} ({\AlgoNameAbbrtwo})}
\label{sec: solution part 2}

In this section, we first design a bi-criteria approximation algorithm {\AlgoNameAbbrtwo} for the {\ProblemNameAbbrN} (i.e., it has an approximation ratio while relaxing the memory limit by a bounded ratio), then remedy the relaxation on the memory limit, and finally extend it to support the {\ProblemNameAbbr}. 
In the beginning, we observe the case where $\tau$ and $|T|$ are small.
In this case, the impact of time slot decoherence diminishes, therefore, all feasible numerologies within $T$ achieve similar fidelity.
To better illustrate this perspective, we conduct the simulation to observe the results, as shown in Fig. \ref{fig: FmaxFminRatio}, where the average path length of each \ac{SD} pair is about seven.
This figure shows that as the time slot length $\tau$ and the batch size of time slots $T$ decrease, the ratio of $\frac{F_{\max}}{F_{\min}}$ approaches $1$, where $F_{\max}$ denotes maximum fidelity that can be achieved among all feasible numerologies within $T$, while $F_{\min}$ denotes the minimum one.
Specifically, this is because:
1) A smaller time slot length $\tau$ reduces decoherence and fidelity loss within each time slot, thereby increasing $F_{\min}$. 2) With a smaller batch size of time slots $T$, the overall processing time decreases and further limits the feasible numerologies available for each request. Thus, these feasible numerologies yield comparable fidelity levels, as $T$ restricts the strategies for entangling and swapping.

Therefore, in this case, we can temporarily neglect $F(m)$ in the {\ProblemNameAbbrN} and consider it later when deriving the bi-criteria approximation ratio.
Then, the modified problem formulation has the objective as 
Eq. (\ref{eq: revised primal objective}) and the constraints (\ref{eq: only select one constraint}), (\ref{eq: memory constraint}), and (\ref{eq: x range constraint}).%
\begin{align}
    \label{eq: revised primal objective}
    \text{maximize } 
    \sum_{i\in I} \sum_{p\in P(i)} \sum_{m\in M_p(i)} \Pr(p) \cdot x^{ip}_m.
\end{align}%

Specifically, the {\AlgoNameAbbrtwo} employs a combinatorial algorithm with a cleverly-designed \ac{DP}-based separation oracle, a randomized rounding technique, and two heuristics for ensuring feasibility, as detailed in Sections \ref{subsec: the combinatorial algorithm}, \ref{subsec: the separation oracle}, \ref{subsec: the randomized rounding algorithm}, and \ref{subsec: heuristic algorithm for solution improvement}, respectively.
{\color{black}
Furthermore, in Section \ref{subsec: heuristic algorithm for solution improvement}, we extend the {\AlgoNameAbbrtwo} to address the constraint (\ref{eq: fidelity threshold constraint}), guaranteeing the solution feasibility for the {\ProblemNameAbbr}.}
Overall, the {\AlgoNameAbbrtwo} can approximate the optimum solution of the {\ProblemNameAbbrN} within an approximation ratio of $O(\frac{F_{\max}}{F_{\min}})$ in the price of a bounded relaxation ratio of $O(\log(|V|\cdot|T|))$ on the memory limit (i.e., bi-criteria approximation).
Then, with the two heuristics, {\AlgoNameAbbrtwo}'s solution can be improved to meet the fidelity threshold and approach the optimum solution as {\color{black} $\tau$ and $|T|$ are small enough}, i.e., $\frac{F_{\max}}{F_{\min}}\approx 1$.

\subsection{The Combinatorial Algorithm} \label{subsec: the combinatorial algorithm}

We first obtain the relaxed \ac{LP} by replacing constraint (\ref{eq: x range constraint}) with constraint (\ref{eq: primal relaxed x}) as follows:
\begin{align}
    \label{eq: primal relaxed x}
    x^{ip}_m \ge 0, && \forall i\in I, \forall p\in P(i), \forall m \in M_p(i).
\end{align}
However, the number of variables in our \ac{ILP} grows exponentially with the input size since the number of feasible numerologies for each \ac{SD} pair may be exponential (i.e., the number of possible binary trees). 
Thus, the relaxed \ac{LP} (\ref{eq: revised primal objective}), (\ref{eq: only select one constraint}), (\ref{eq: memory constraint}), (\ref{eq: primal relaxed x}) cannot be solved by an \ac{LP} solver in polynomial time.
On the other hand, by \cite{introducetoalgorithm}, if an \ac{LP} is in standard form and has a polynomial number of constraints (except for non-negativity constraints of variables), the number of variables in its dual \ac{LP} is also polynomial, as described in Definition \ref{defi: linear program}.
Then, by combinatorial algorithms, we can obtain a near-optimum solution to the primal \ac{LP} in polynomial time if the following properties are satisfied \cite{primal_dual}.
\begin{enumerate}
    \item Each coefficient on the left-hand side is not greater than the constant on the right-hand side for every inequality of the constraint in the primal \ac{LP} (except for the non-negativity constraints of variables).
    \item In the dual \ac{LP}, all coefficients on the left-hand-side of inequality and the constant on the right-hand-side of inequality are positive in each constraint (except for the non-negativity constraints of variables).
    \item A separation oracle exists to tell whether there is a violated constraint for the dual \ac{LP}.
\end{enumerate}

\begin{figure}[t]
\centering
\includegraphics[width=.48\textwidth]{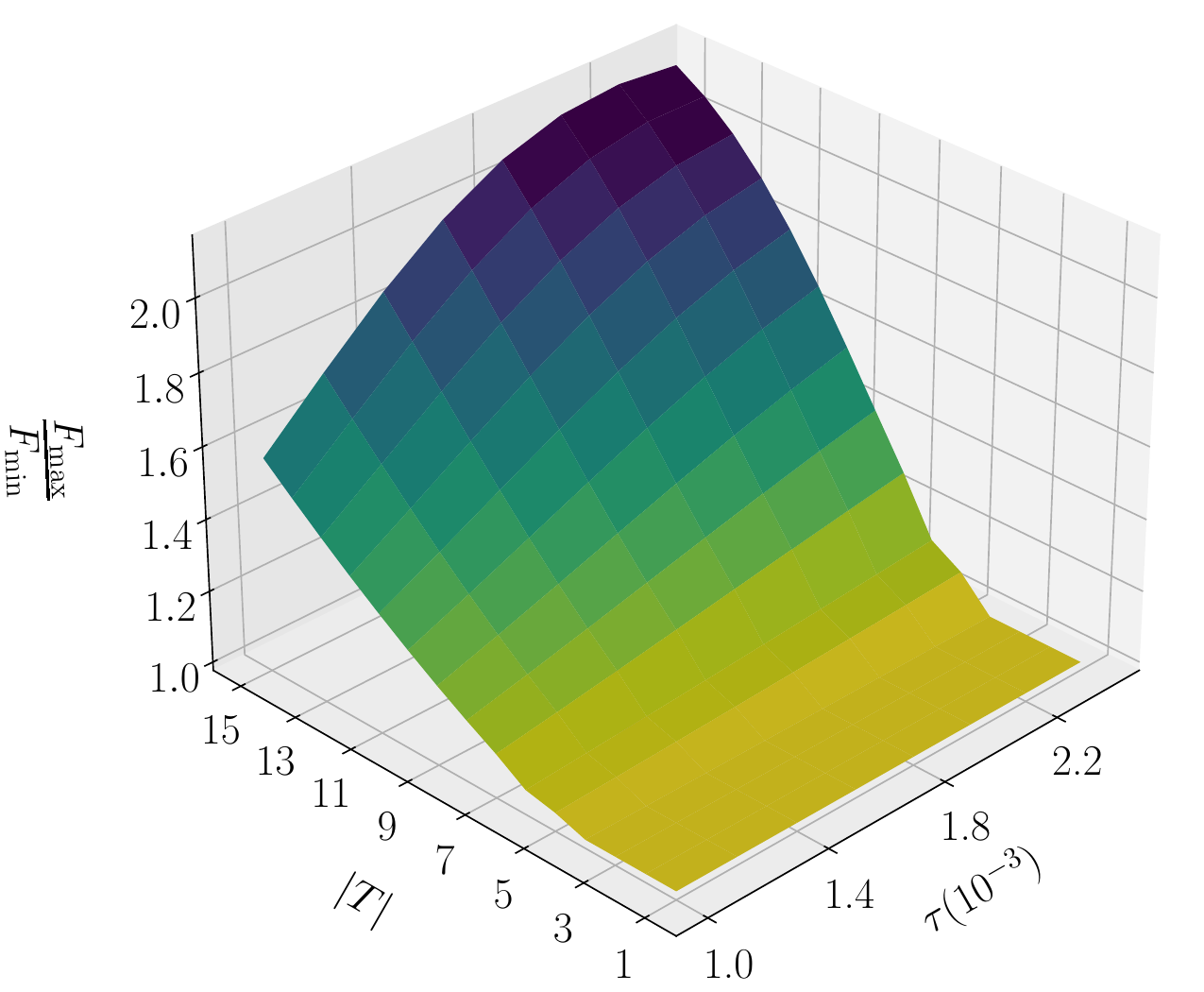}
    \caption{The ratio of $\frac{F_{\max}}{F_{\min}}$ for different values of $\tau$ and $|T|$.}
    \label{fig: FmaxFminRatio}
\end{figure}

\begin{definition} \label{defi: linear program}
\cite{introducetoalgorithm} The \ac{LP} standard form is $\max \{c^T x| Ax\le b, x \ge 0\}$, where $x$ is an $n \times 1$ variable vector, and $c$, $b$, and $A$ denote $n \times 1$, $m \times 1$, $m \times n$ constant vectors, respectively.
The corresponding dual \ac{LP} is $\min\{b^T y| A^T y \ge c, y \ge 0\}$, where $y$ is an $m \times 1$ variable vector.
Note that $x \ge 0$ and $y\ge 0$ are the non-negativity constraints of $x$ and $y$.
\end{definition}

Clearly, the relaxed primal \ac{LP} (\ref{eq: revised primal objective}), (\ref{eq: only select one constraint}), (\ref{eq: memory constraint}), and (\ref{eq: primal relaxed x}) meets the first property.
By Definition \ref{defi: linear program}, we associate dual variables $\alpha_i$ and $\beta^t_v$ with each constraint in (\ref{eq: only select one constraint}) and (\ref{eq: memory constraint}), respectively. Thus, we obtain the corresponding dual \ac{LP} (\ref{eq: dual objective})$-$(\ref{eq: dual beta constraint}).
It can be seen that the dual \ac{LP} also satisfies the second property.
\begin{subequations}
\begin{align}
    \label{eq: dual objective}
    \text{minimize \ \ } 
        & \sum_{i\in I} \alpha_i+\sum_{v\in V}\sum_{t\in T} c^t(v)\cdot\beta^t_v
    \\
    \begin{split}
        \label{eq: dual length constraint}
        \text{subject to \ \ }
        & \alpha_i + \sum_{v\in V}\sum_{t\in T} \theta_m(t,v) \cdot \beta^{t}_v\ge \Pr(p),\\
        & \qquad\qquad\, \forall i \in I, \forall p\in P(i), \forall m \in M_p(i)
    \end{split}\\
    \label{eq: dual beta constraint}
    & \alpha_i, \beta^{t}_{v} \ge 0, \quad\qquad \forall i\in I, \forall v \in V, \forall t \in T
\end{align}
\label{dual linear programmming}%
\end{subequations}%
Then, we design a separation oracle for the dual \ac{LP} and get the fractional solution for the relaxed primal \ac{LP} in Section \ref{subsec: the separation oracle}. However, the fractional solution may be infeasible for the \ac{ILP} (\ref{eq: revised primal objective}), (\ref{eq: only select one constraint}), (\ref{eq: memory constraint}), and (\ref{eq: x range constraint}). 
To this end, we propose an \ac{LP} rounding algorithm and improve the solution in Sections \ref{subsec: the randomized rounding algorithm} and \ref{subsec: heuristic algorithm for solution improvement}. Last, we analyze its performance in Section \ref{subsec: the analysis of randomized rounding algorithm}.

\subsection{The Separation Oracle}
\label{subsec: the separation oracle}

Given an arbitrary (fractional) solution ($\alpha, \beta$) to the dual \ac{LP}, the separation oracle tells whether a violated constraint exists.
It is easy to examine every constraint in (\ref{eq: dual beta constraint}) in polynomial time, but the challenge arises in identifying whether a violated constraint exists in (\ref{eq: dual length constraint}).
Note that the number of \ac{SD} pairs and the size of path set for each \ac{SD} pair are both polynomial. 
Therefore, it suffices to identify the most violated constraint in (\ref{eq: dual length constraint}) for each pair $i \in I$ and each path $p\in P(i)$. This can be done by computing
\begin{align}
    \label{eq: shortest length function}
    \min_{m\in M_p(i)} \sum_{v\in V}\sum_{t\in T}\theta_m(t, v)\cdot\beta^{t}_v.
\end{align}%

To this end, we subtly design a \ac{DP} algorithm to iteratively solve a larger subproblem by examining its smaller subproblems.
Specifically, any numerology $m\in M_p(i)$ has a one-to-one and onto mapping to a strategy tree through path $p$ with a root represented by the \ac{SD} pair $(s_i,d_i)$, as shown in Fig. \ref{fig: nume}. 
We introduce the function $\hat{f}(T,(s,d))$ to find a numerology $m$ that can generate an entangled pair from node $s$ to node $d$ through the path $p$ within $T$, while minimizing $\sum_{v\in V}\sum_{t\in T}\theta_m(t,v)\cdot\beta^t_v$.
In this way, the desired numerology for a given \ac{SD} pair $i$ and path $p \in P(i)$ can be found by using $\hat{f}(T,(s_i,d_i))$.

To compute $\hat{f}(T,(s,d))$, we examine all local minimum numerologies, each corresponding to a strategy tree rooted at a time slot $t\in T$.
We know that any feasible numerology $m$ requires $\theta_m(t,v)$ memory units on node $v$ at time slot $t$, and a strategy tree can be built by combining two sub-strategy trees.
Thus, we define a function $\hat{g}(t,(s,d),(\sigma_s,\sigma_d))$ to find a numerology $m$ with a (sub-)strategy tree rooted at $(s,d)$ at time slot $t$, which can minimize $\sum_{v\in V}\sum_{t'\in T}\theta_m(t',v)\cdot\beta^{t'}_v$ while ensuring the memory limit $c^t(s)\ge \sigma_s$ and $c^t(d)\ge \sigma_d$, where $\sigma_s,\sigma_d\in \{1,2\}$.
Then, Eq. (\ref{eq: dp root case}) derives $\hat{f}(T,(s,d))$.
\begin{align}
    \hat{f}(T,(s,d))
    =
    \min_{t\in T}
        \big\{\hat{g}(t,(s,d),(1,1))\big\}    
    \label{eq: dp root case}
\end{align}%

\begin{table*}[ht]
\small
\centering
\begin{minipage}{1\textwidth}

\begin{align} 
    & \hspace{-0.30cm} \hat{g}(t,(s,d),(\sigma_s,\sigma_d)) =  
    \begin{cases}
    \beta^t_s + \beta^t_d + \beta^{t-1}_s + \beta^{t-1}_d
        ~
        & \mbox{if } \substack{(s,d)\in E, \, t\ge 2, \\  c^t(s) \text{ and } c^{t-1}(s)\ge \sigma_s, \\  c^t(d) \text{ and } c^{t-1}(d)\ge \sigma_d;} \\
    \beta^t_s + \beta^t_d + 
    \min \bigg\{\hat{g}(t-1,(s,d),(\sigma_s,\sigma_d)), \displaystyle\min_{k \in \mathcal{I}(s, d)} \big\{\hat{h}(t-1,(s,k,d),(\sigma_s,\sigma_d))\big\}\bigg\}
        ~
        & \mbox{else if } \substack{t\ge 2, \, c^t(s)\ge \sigma_s, \\ c^t(d)\ge \sigma_d;} \\
    \infty 
        & \mbox{otherwise.} 
    \end{cases} \notag
    \\[-7.5ex] \label{eq: dp two cases}
    \\[6ex]
    & \hspace{-0.30cm} \hat{h}(t,(s,k,d),(\sigma_s,\sigma_d)) = \min \bigg\{ 
    \hat{g}(t,(s,k),(\sigma_s,2)) + \hat{g}(t,(k,d),(1,\sigma_d)), \; 
    \hat{g}(t,(s,k),(\sigma_s,1)) + \hat{g}(t,(k,d),(2,\sigma_d))
    \bigg\}.
    \label{eq: dp state case}
\end{align}

\hrule
\end{minipage}
\end{table*}

We derive $\hat{g}(t,(s,d),(\sigma_s,\sigma_d))$ according to the three cases.
\begin{enumerate}
    \item \emph{Leaves of strategy tree}. If $(s,d)\in E$ and $t\ge 2$, then we reach a pair of leaves in a strategy tree. Thus, it returns the sum of values of these two leaves and their external nodes as long as the memory limits of $s$ and $d$ are sufficient.
    \item \emph{Non-leaves of strategy tree}. If $(s,d)\notin E$ and $t\ge 2$, then we have two possible cases at time slot $t-1$. 1) Both $s$ and $d$ idle. 2) $s$ and $d$ conduct swapping to consume two entangled links $(s,k)$ and $(k,d)$ to acquire an entangled link $(s,d)$, where $k$ is an intermediate node between $s$ and $d$ on path $p$.
    It examines the values of $s$ and $d$ plus the value of every case and then returns the minimum if the memory limits of $s$ and $d$ are sufficient.
    \item \emph{Wrong time or no capacity}. If the time slot $t\le 1$, it is impossible to have an entangled link $(s,d)$ at time slot $t$. Moreover, if the memory limits of $s$ and $d$ do not meet the requirements, the numerology does not exist. Thus, these cases result in an infeasible solution.
\end{enumerate}%
Based on the above cases, to derive $\hat{g}(t,(s,d),(\sigma_s,\sigma_d))$, we additionally define the function $\hat{h}(t,(s,k,d),(\sigma_s,\sigma_d))$ in Eq. (\ref{eq: dp state case}) to represent the minimum $\sum_{v\in V}\sum_{t'\in T} \theta_m(t',v)\cdot \beta^{t'}_v$ of the entangled pair $(s,d)$ that can be generated by limiting the amount of used memory on node $k$ in either the left or right subproblem.
Specifically, if the amount of memory on node $k$ is limited in the left (or right) subproblem, then we set $\sigma_d=2$ (or $\sigma_s=2$) to check whether the amount of memory of $k$ is at least $2$.
Otherwise, the left (or right) subproblem has no limit, and we set $\sigma_d=1$ (or $\sigma_s=1$).
Then, the recurrence relation of $\hat{g}(t,(s,d),(\sigma_s,\sigma_d))$ can be expressed as Eq. (\ref{eq: dp two cases}), where $\mathcal{I}(s,d)$ denotes the set of intermediate nodes on the path $p$ between $s$ and $d$ (i.e., the nodes on $p$ but excluding $s$ and $d$).
Eqs. (\ref{eq: dp root case})$-$(\ref{eq: dp state case}) can help find the numerology $m$ that minimize $\sum_{v\in V}\sum_{t\in T}\theta_m(t,v)\cdot\beta^{t}_{v}$ for any given $i\in I$ and $p\in P(i)$.

With Eqs. (\ref{eq: dp root case})$-$(\ref{eq: dp state case}), we can efficiently identify the most violated constraint in (\ref{eq: dual length constraint}), satisfying the third property.
As all three properties are satisfied, we can solve the relaxed primal \ac{LP} and obtain a fractional solution $\hat{x}^{ip}_m$ using the combinatorial algorithm in \cite{primal_dual} that incorporates our \ac{DP} separation oracle.

\subsection{The Randomized Rounding}
\label{subsec: the randomized rounding algorithm}

Given a fractional solution $\hat{x}^{ip}_{m}$ of our primal \ac{LP}, we design a randomized rounding algorithm to obtain the solution for {\ProblemNameAbbrN}. 
Specifically, let $\hat{M}_p(i)$ denote the set of numerologies with $\hat{x}^{ip}_m > 0$ for each \ac{SD} pair $i\in I$ and path $p\in P(i)$. 
Then, let $\hat{M}(i) = \hat{M}_{p_1}(i) \cup \hat{M}_{p_2}(i) \cup \dots \cup \hat{M}_{p_n}(i)$, where $\{p_1,p_2,\dots,p_n\} = P(i)$.
It is worth noting that the size of $\hat{M}(i)$ is polynomial since the combinatorial algorithm terminates in polynomial time \cite{primal_dual}.
We then interpret $\hat{x}^{ip}_m$ as the probability of selecting the numerology $m$ for each \ac{SD} pair $i$. 
For example, assume that \ac{SD} pair $i$ has three numerologies $m_1$, $m_2$, and $m_3$ in $\hat{M}(i)$, and they are $\hat{x}^{ip_1}_{m_1} = 0.1$, $\hat{x}^{ip_1}_{m_2} = 0.2$, and $\hat{x}^{ip_2}_{m_3} = 0.3$, respectively. 
Thus, the probabilities of selecting $m_1$, $m_2$, and $m_3$ are $0.1$, $0.2$, and $0.3$, respectively, while the probability of not selecting any numerology for \ac{SD} pair $i$ is $0.4$.

\subsection{Bi-Criteria Approximation Ratio and Time Complexity}
\label{subsec: the analysis of randomized rounding algorithm}

We first analyze the feasibility and the relaxation bound of Eqs. (\ref{eq: only select one constraint}) and (\ref{eq: memory constraint}) after randomized rounding, and the time complexity in Lemmas \ref{lemma: demand violet probability}, \ref{lemma: memory violet probability}, and \ref{lemma: time complexity}, respectively. We then derive the bi-criteria approximation ratio in Theorem \ref{theo: QR ratio}. To this end, we use the Chernoff bound in Theorem \ref{chernoff bound} as follows.
\begin{theorem} [Chernoff bound \cite{chernoff_bound}] \label{chernoff bound}
    There is a set of $n$ independent random variables $x_1,...,x_n$, where $x_i \in [0, 1]$ for each $i\in[1,n]$. Let $X=\sum_{i=1}^n x_i$ and $\mu = \mathbb{E}{[X]}$. Then,
    \begin{equation}
        \text{Pr}\left[\sum_{i=1}^n x_i \ge (1+\epsilon)\mu\right] \le e^{\frac{-\epsilon^2\mu}{2+\epsilon}}.
    \end{equation}
\end{theorem}

\begin{lemma}
    \label{lemma: demand violet probability}
     The {\AlgoNameAbbrtwo} will satisfy constraint (\ref{eq: only select one constraint}).
\end{lemma}
\begin{proof}
The {\AlgoNameAbbrtwo} selects at most one numerology from the union of all sets $M_p(i)$ for each \ac{SD} pair $i$ using the randomized rounding in Section \ref{subsec: the randomized rounding algorithm}. Thus, the lemma holds.
\end{proof}

\begin{lemma}
    \label{lemma: memory violet probability}
    The probability that the {\AlgoNameAbbrtwo} relaxes constraint (\ref{eq: memory constraint}) for any node $v$ at any time slot $t$ by more than a factor of $(1 + 4 \ln(|V|\cdot|T|))$ is at most $\frac{1}{|V|^2\cdot|T|^2}$, which is negligible.
\end{lemma}
\begin{proof}
For each node $v$ at time slot $t$, we define a random variable $z^{i}_{t,v}$ that denotes the amount of memory occupied on node $v$ at time slot $t$ for each \ac{SD} pair $i$. Note that a quantum node at time slot $t$ may require $\theta_m(t, v)$ units of memory to perform numerology $m$. Thus, $z^{i}_{t,v}$ can be expressed as:
\begin{align} \notag
\displaystyle z^{i}_{t,v} &= 
    \begin{cases}
        \theta_m(t, v) & \text{with probability} \displaystyle \ \hat{x}^{ip}_m;\\
        0 & \text{otherwise}.
    \end{cases}
\end{align}
Note that their sum $Z_{t,v}= \sum_{i\in I} z^{i}_{t,v}$ is exactly the amount of memory needed on $v$ at time slot $t$ after rounding.
Then, we derive the upper bound of the expectation of the amount of memory occupied on node $v$ at time slot $t$ as follows:
\begin{align}
    \mathbb{E}[Z_{t,v}] & = \mathbb{E}\left[\sum_{i\in I}z^{i}_{t,v}\right] = \sum_{i\in I}\mathbb{E}[z^{i}_{t,v}] \notag \\
    & = \sum_{i \in I}\bigg(\sum_{p\in P(i)} \sum_{m\in M_p(i)}\theta_m(t,v) \cdot \hat{x}^{ip}_m\bigg) 
    \le  c^t(v). \notag
\end{align}
Note that the last inequality directly follows the memory limit constraint (\ref{eq: memory constraint}).
Afterward, we can prove the lemma by Chernoff bound as follows:
\begin{align}
    &~ \Pr(\exists v\in V, \exists t\in T: Z_{t,v} \ge (1+4\ln(|V|\cdot|T|))\cdot c^t(v)) \notag\\
    \le &~ \sum_{v \in V}\sum_{t \in T}\Pr(\frac{Z_{t,v}}{c^t(v)} \ge 1 + 4\ln(|V|\cdot|T|)) \notag \\
    \le &~ \sum_{v \in V}\sum_{t \in T}\Pr(\sum_{i\in I}\frac{z^{i}_{t,v}}{c^t(v)} \ge 1 + 4 \ln(|V|\cdot|T|)) \notag \\
    \le &~ |V|\cdot |T|\cdot e^{\frac{-(4\ln(|V|\cdot|T|))^2}{2+ 4\ln(|V|\cdot|T|)}}
    \le |V|^{-2}\cdot |T|^{-2}.
    \notag
\end{align}
Note that $\frac{z^{i}_{t,v}}{c^t(v)}\le 1$ since the separation oracle adopts Eq. (\ref{eq: dp two cases}) to examine the memory limit, meeting the condition of the Chernoff bound (i.e., every variable $\frac{z^{i}_{t,v}}{c^t(v)}$ ranges from $0$ to $1$).
Besides, the last inequality holds when $|V|\cdot|T|\ge 5$.
\end{proof}

\begin{lemma} \label{lemma: time complexity}
    The {\AlgoNameAbbrtwo} can be executed in polynomial time.
\end{lemma}
\begin{proof}
According to \cite{primal_dual}, the combinatorial algorithm executes $O(\omega^{-2} \eta_1\log \eta_1)$ times of separation oracle, where $\omega$ is a constant error bound of primal \ac{LP} solution and $\eta_1$ is the number of constraints (except for the non-negativity constraints), i.e., $\eta_1 = (|V|\cdot|T| + |I|)$. 
Thus, it outputs $O(\omega^{-2} \eta_1 \log \eta_1)$ numerologies for randomized rounding.
In addition, each time of separation oracle takes $\eta_2 = O( |V|^3 \cdot |T| \cdot |I|)$.
Then, the combinatorial algorithm takes $O(\omega^{-2}\eta_1 \eta_2 \log \eta_1)$, and thus
the overall time complexity of {\AlgoNameAbbrtwo} is $O(\omega^{-2} \eta_1 \eta_2 \log \eta_1)$. 
\end{proof}

\begin{theorem}
    \label{theo: QR ratio}
    \color{black}
    The bi-criteria approximation ratio of the {\AlgoNameAbbrtwo} for the {\ProblemNameAbbrN} is $(O(\frac{F_{\max}}{F_{\min}}), O(\log(|V|\cdot|T|))$.
\end{theorem}
\begin{proof}
With Lemmas \ref{lemma: demand violet probability}, \ref{lemma: memory violet probability}, and \ref{lemma: time complexity}, we can know that {\AlgoNameAbbrtwo} relaxes the memory limit by a ratio of at most $O(\log(|V|\cdot|T|))$ with high probability and runs in polynomial time.
It suffices to focus on the gap between the optimal solution and the {\AlgoNameAbbrtwo}'s solution.
For ease of presentation, let $OPT$, $OPT^{nF}$, and $OPT^{nF}_{PL}$ be the optimum values of the {\ProblemNameAbbrN} (i.e., the \ac{ILP} (\ref{eq: primal ILP}) with no constraint (\ref{eq: fidelity threshold constraint})), the {\ProblemNameAbbrN} without considering fidelity (i.e., the \ac{ILP} (\ref{eq: revised primal objective}), (\ref{eq: only select one constraint}), (\ref{eq: memory constraint}), (\ref{eq: x range constraint})), and the primal \ac{LP} (\ref{eq: revised primal objective}), (\ref{eq: only select one constraint}), (\ref{eq: memory constraint}), (\ref{eq: primal relaxed x}), respectively.

Clearly, $OPT \le OPT^{nF}\cdot F_{\max}$ because $OPT^{nF}$ has the maximum number of \ac{SD} pairs with an allocated numerology.
Besides, $OPT^{nF} \le OPT^{nF}_{PL}$ since the primal \ac{LP} (\ref{eq: revised primal objective}), (\ref{eq: only select one constraint}), (\ref{eq: memory constraint}), (\ref{eq: primal relaxed x}) is a \ac{LP} relaxation of the \ac{ILP} (\ref{eq: revised primal objective}), (\ref{eq: only select one constraint}), (\ref{eq: memory constraint}), (\ref{eq: x range constraint}).
{\color{black}
In addition, let $\hat{x}^{ip}_m$ denote the solution output by the combinatorial algorithm.
Then, let $\bar{x}^{ip}_{m}$ be the rounded solution.}
Since the expected value of the rounded solution is equal to the optimal solution of the primal \ac{LP}, the expected fidelity sum of the rounded solution can be bound as follows:
\begin{align}
    &~
    \mathbb{E}\left[\sum_{i\in I}\sum_{p\in P(i)}\sum_{m\in M_p(i)} \Pr(p)\cdot F(m)\cdot \bar{x}^{ip}_m \right] \notag \\ 
    \ge &~ \sum_{i\in I}\sum_{p\in P(i)}\sum_{m\in M_p(i)} \mathbb{E}\left[\Pr(p)\cdot F_{\min}\cdot \bar{x}^{ip}_m \right] \notag \\ 
    = &~ F_{\min} \cdot \sum_{i\in I}\sum_{p\in P(i)}\sum_{m\in M_p(i)} \Pr(p)\cdot \hat{x}^{ip}_{m}  \notag \\ 
    = &~ F_{\min} \cdot APP^{nF}_{PL} \ge \frac{F_{\min}}{1 + \omega}\cdot OPT^{nF}_{PL} \notag \\ 
    \ge &~ \frac{F_{\min}}{1 + \omega}\cdot OPT^{nF} \ge \frac{F_{\min}}{F_{\max}} \cdot \frac{OPT}{1+\omega}, \notag 
\end{align}
where $APP^{nF}_{PL}$ is the near-optimal value of primal \ac{LP} from the combinatorial algorithm with a user-defined constant error bound $\omega>0$, and the second inequality holds due to \cite{primal_dual}.
\end{proof}%

\subsection{Heuristic Algorithms for Solution Improvement}
\label{subsec: heuristic algorithm for solution improvement}

\subsubsection{Remedy for the Memory Limit Relaxation}
\label{subsubsec: heuristic algorithm for solution improvement 1}
Since the solution acquired by the {\AlgoNameAbbrtwo} after rounding may relax the memory limit on nodes with a bounded ratio, the following steps are adopted to deal with the memory limit relaxation.
Let $\bar{x}^{ip}_m$ and $\bar{M}_p(i)$ denote the variables in the rounded solution and the set of numerologies with $\bar{x}^{ip}_m>0$ for each $i\in I$ and $p\in P(i)$, respectively. The heuristic has two phases as follows.

In the first phase, for each node $v$ and each time slot $t$ where the memory limit is exceeded, we sort the numerologies occupying $v$'s memory at time slot $t$ (i.e., $\bar{M}(v,t)=\{m~|~ \theta_m(t,v)\cdot \bar{x}^{ip}_m > 0, \forall m\in \bar{M}_p(i), \forall p\in P(i),\forall i\in I\} $) in non-decreasing order of their expected fidelity, and then iteratively remove numerologies in $\bar{M}(v,t)$ based on this order until the memory limit of node $v$ at time slot $t$ is satisfied to make the solution feasible.

In the second phase, we utilize the residual resources as much as possible. 
We iteratively allocate the numerology $m$ for each pair $i$ in non-increasing order of the value of those positive variables $\hat{x}^{ip}_m$ (i.e., the fractional solution of \ac{LP} (\ref{eq: revised primal objective}), (\ref{eq: only select one constraint}), (\ref{eq: memory constraint}), (\ref{eq: primal relaxed x})) until no numerology can be chosen or accommodated by the residual network.
The heuristic is a polynomial-time algorithm since the number of positive variables $\hat{x}^{ip}_m$ is polynomial.

\subsubsection{Extension to Support the {\ProblemNameAbbr}}
\label{subsubsec: heuristic algorithm for solution improvement 2}

We design the {\AlgoNameAbbrtwo} based on the observation of numerology characteristics in the context of small $\tau$ and $|T|$. 
However, as $\tau$ and $|T|$ increase, the {\AlgoNameAbbrtwo} may not guarantee solution feasibility for the {\ProblemNameAbbr}. 
To address this limitation, we extend the {\AlgoNameAbbrtwo} to support the {\ProblemNameAbbr} by introducing a heuristic as follows.
Following the heuristic proposed in Section \ref{subsubsec: heuristic algorithm for solution improvement 1}, we additionally add one step to remove the numerology from $\bar{M}(v,t)$ if its fidelity is less than the threshold $\widehat{F}$ in the first phase.
In the second phase, we only add the numerology $m$ with adequate fidelity for each pair $i$ in non-increasing order of the value of those positive variables $\hat{x}^{ip}_m$ until no $m$ can be chosen or accommodated.
This heuristic can ensure solution feasibility, and later, the numerical results show the modified solution can approach the optimum.

\section{The Design of {\AlgoNameone} ({\AlgoNameAbbrone})}
\label{sec: solution part 1}

Although the {\AlgoNameAbbrtwo} can approximate the optimum when $\tau$ and $|T|$ are small, its performance might decrease when $\tau$ or $|T|$ increases since it neglects fidelity loss.
However, it is non-trivial to find the optimal solution for the {\ProblemNameAbbr} since every \ac{SD} pair could have an exponential number of numerologies for selection. 
Fortunately, the separation oracle in the {\AlgoNameAbbrtwo} serves as inspiration for determining the optimal numerology. Following this inspiration and the intrinsic properties of {\ProblemNameAbbr}, we first design a \ac{DP} algorithm that can find the optimal numerology with the maximum expected fidelity on the selected path $p\in P(i)$ for any given \emph{single} \ac{SD} pair $i$.

Intuitively, we can derive a solution for the {\ProblemNameAbbr} by iteratively adding a new request based on the available resource by invoking the \ac{DP} algorithm until no more request can be satisfied.
Nonetheless, such a greedy algorithm tends to fall into a local optimum trap since it lacks an overview of current resource utilization for load balancing unless a proper index is provided for the greedy algorithm to estimate the efficiency of the current allocation.
Inspired by \cite{greedyIndicator}, we subtly design the \acl{REI} (\acs{REI}) to evaluate a numerology.
With \ac{REI}, our greedy algorithm can mitigate severe network congestion and increase the fidelity sum as {\color{black}$\tau$ or $|T|$ grows.}

\subsection{Numerology with the Max. Exp. Fidelity for Single Request}
\label{subsec: DP}

Similar to Section \ref{subsec: the separation oracle}, we exploit \ac{DP} to iteratively solve a larger subproblem by examining the optimum solution of each smaller subproblem to derive the optimal solution for \ac{ILP} (\ref{eq: primal objective})$-$(\ref{eq: x range constraint}) when the \ac{SD} pair $i$ is single.
Hence, we introduce the recursive function $f(T,(s,d))$ to return the maximum fidelity that can be achieved between node $s$ and node $d$ on a specific path $p$ within $T$ as well as the corresponding numerology.
Similar to Eq. (\ref{eq: dp root case}), $f(T,(s,d))$ requires an additional function $g(t,(s,d),(\sigma_s,\sigma_d))$ to derive its value, as shown in Eq. (\ref{eq: greedy dp root case}).
\begin{align}
f(T,(s,d)) = \max_{t\in T} \big\{g(t,&(s,d),(1,1)) \mid \notag \\
& g(t,(s,d),(1,1)) \ge \widehat{F} \big\}.
\label{eq: greedy dp root case}
\end{align}
Note that if the derived maximum fidelity is less than the fidelity threshold $\widehat{F}$, the corresponding path will not be considered.
Following Eq. (\ref{eq: dp two cases}), to derive $g(t,(s,d),(\sigma_s,\sigma_d))$, we define the function $h(t,(s,k,d),(\sigma_s,\sigma_d))$ in Eq. (\ref{eq: greedy dp state case}), which represents the maximum fidelity of the entangled pair $(s,d)$ that can be generated by limiting the amount of used memory on node $k$ in either the left or right subproblem.
The details are omitted due to the similarity, and the recurrence relation of $g(t,(s,d),(\sigma_s,\sigma_d))$ can be expressed as Eq. (\ref{eq: greedy dp two cases}).
In this way, we can identify the numerology with the maximum fidelity on any specific path for any given single \ac{SD} pair.
Subsequently, it is possible to find the numerology with the maximum expected fidelity among all paths in $P(i)$ since the size of predefined path set $P(i)$ is polynomial. To this end, we calculate $\Pr(p)\cdot F(\hat{m})$ of the numerology $\hat{m}$ returned by Eq. (\ref{eq: greedy dp root case}) for each $p\in P(i)$ and choose the numerology with the highest value.

\begin{table*}[t]
\small
\centering
\begin{minipage}{1\textwidth}
\begin{align} 
    & \hspace{-0.6cm} g(t,(s,d),(\sigma_s,\sigma_d))= 
    \begin{cases}
    F(s,d)
        ~
        & \quad\;\; \mbox{if } \substack{(s,d)\in E, \, t\ge 2, \\ c^t(s) \text{ and } c^{t-1}(s)\ge \sigma_s, \\ c^t(d) \text{ and } c^{t-1}(d)\ge \sigma_d;} \\
    \max \bigg\{
    F^{\tau}\big(g(t-1,(s,d),(\sigma_s,\sigma_d))\big), 
    \displaystyle\max_{k\in\mathcal{I}(s,d)} \big\{h(t-1,(s,k,d),(\sigma_s,\sigma_d))\big\}
    \bigg\}
        ~
        & \quad\;\; \mbox{else if } \substack{t\ge 2, \, c^t(s)\ge \sigma_s, \\ c^t(d)\ge \sigma_d;} \\
    -\infty 
        & \quad\;\; \mbox{otherwise.} 
    \end{cases} \notag
    \\[-7.5ex] \label{eq: greedy dp two cases}
    \\[6ex]
    & \hspace{-0.6cm} h(t,(s,k,d),(\sigma_s,\sigma_d)) = \max \bigg\{ 
    F^{\tau}_{s}\big(g(t,(s,k),(\sigma_s,2)), g(t,(k,d),(1,\sigma_d))\big), \;
    F^{\tau}_{s}\big(g(t,(s,k),(\sigma_s,1)), g(t,(k,d),(2,\sigma_d))\big)
    \bigg\}.
    \label{eq: greedy dp state case}
\end{align}
\begin{align}
    & \hspace{-0.60cm} \bar{f}(T,(s,d)) 
    = \min_{t\in T}\big\{\bar{g}(t,(s,d),(1,1))\big\}.
    \label{eq: greedy 2 dp root case} \\
    & \hspace{-0.60cm} \bar{g}(t,(s,d),(\sigma_s,\sigma_d))= 
    \begin{cases}
    r^t_s+r^t_d + r^{t-1}_s+r^{t-1}_d
        ~
        & \mbox{if } \substack{(s,d)\in E, \, t\ge 2, \\ c^t(s) \text{ and } c^{t-1}(s)\ge \sigma_s, \\ c^t(d) \text{ and } c^{t-1}(d)\ge \sigma_d;} \\
    r^t_s+r^t_d+
    \min \bigg\{\bar{g}(t-1,(s,d),(\sigma_s,\sigma_d)), \displaystyle\min_{k \in \mathcal{I}(s, d)} \big\{\bar{h}(t-1,(s,k,d),(\sigma_s,\sigma_d))\big\}\bigg\}
        ~
        & \mbox{else if } \substack{t\ge 2, \, c^t(s)\ge \sigma_s, \\ c^t(d)\ge \sigma_d;} \\
    \infty 
        & \mbox{otherwise.} 
    \end{cases} \notag
    \\[-7.5ex] \label{eq: greedy 2 dp two cases}
    \\[6ex]
    & \hspace{-0.60cm} \bar{h}(t,(s,k,d),(\sigma_s,\sigma_d)) = \min \bigg\{ 
    \bar{g}(t,(s,k),(\sigma_s,2)) + \bar{g}(t,(k,d),(1,\sigma_d)), \; 
    \bar{g}(t,(s,k),(\sigma_s,1)) + \bar{g}(t,(k,d),(2,\sigma_d))
    \bigg\}.
    \label{eq: greedy 2 dp state case}
\end{align}
\hrule
\end{minipage}
\end{table*}

\subsection{Greedy Algorithm with \ac{DP} Technique}
\label{subsec: Greedy DP}

We can design a greedy algorithm to solve the \ac{ILP} (\ref{eq: primal objective})$-$(\ref{eq: x range constraint}) by utilizing the \ac{DP} algorithm in Section \ref{subsec: DP}.
That is, we iteratively choose the numerology with the maximum expected fidelity by the \ac{DP} algorithm among all \ac{SD} pairs and then allocate the resources to the corresponding \ac{SD} pair until no more numerologies can be chosen.
However, such a naive algorithm may cause congestion.
To remedy the side effect, we design the \ac{REI} to help us choose a numerology while avoiding hot spots.
Specifically, the \ac{REI} is defined as follows:
\begin{align}
    \label{eq: CP}
    \mathfrak{R}(i,m) = \frac{\Pr(p)\cdot F(m)}{\sum_{v\in V}\sum_{t\in T}\frac{\theta_m(t, v)}{c^t(v)}},
\end{align}
where $p\in P(i)$ and $m\in M_p(i)$ for \ac{SD} pair $i$.
Note that the numerator and denominator denote the expected fidelity and resource cost of numerology $m$, respectively.
However, finding the numerology with the maximum $\mathfrak{R}(i,m)$ among all possible numerologies is computationally-intensive.
Thus, {\AlgoNameAbbrone} focuses on two types of candidate numerologies for each \ac{SD} pair $i$ as follows:
1) the numerology with the maximum expected fidelity on each path $p\in P(i)$, denoted by $m^{ip}_1$, and
2) the numerology with the minimum resource cost on each path $p\in P(i)$, denoted by $m^{ip}_2$.
Subsequently, {\AlgoNameAbbrone} selects the numerology with the highest $\mathfrak{R}(i,m)$ from all candidates across all \ac{SD} pairs, i.e., $\max_{i\in I, p\in P(i)}\{\mathfrak{R}(i,m^{ip}_1),\mathfrak{R}(i,m^{ip}_2)\}$.

Specifically, {\AlgoNameAbbrone} invokes the \ac{DP} algorithm in Section \ref{subsec: DP} to determine the $m^{ip}_1$ for each path $p\in P(i)$.
In addition, $m^{ip}_{2}$ for each $p$ can be found using a similar \ac{DP} algorithm designed in Eqs. (\ref{eq: greedy 2 dp root case})$-$(\ref{eq: greedy 2 dp state case}).
These equations describe the recurrence relation to search the numerology with the minimum resource cost, where $r^t_v$ denotes the resource cost of node $v$ at time slot $t$ and is set to $\frac{1}{c^t(v)}$ based on Eq. (\ref{eq: CP}).
In brief, $\bar{f}(T,(s,d))$ finds the numerology that can generate an entangled pair from node $s$ to node $d$ on specific path $p$ with the minimum resource cost within $T$, 
while $\bar{g}(t,(s,d),(\sigma_s,\sigma_d))$ and $\bar{h}(t,(s,k,d),(\sigma_s,\sigma_d))$ are additional functions similar to Eqs. (\ref{eq: greedy dp two cases}) and (\ref{eq: greedy dp state case}).
Note that the numerology will not be considered if the fidelity is less than the fidelity threshold $\widehat{F}$, and the \ac{DP} details are omitted due to the similarity.

Overall, {\AlgoNameAbbrone} iteratively chooses an \ac{SD} pair with the numerology that has the maximum \ac{REI} among all unsaturated \ac{SD} pairs and then allocates the resources asked by that numerology to saturate the \ac{SD} pair until no more \ac{SD} pairs can be served.
Then, we analyze the time complexity of {\AlgoNameAbbrone}. Both the \ac{DP} algorithm Eqs. (\ref{eq: greedy dp root case}) and (\ref{eq: greedy 2 dp root case}) take $O(|V|^3\cdot |T|)$ for each \ac{SD} pair $i\in I$ on specific path $p\in P(i)$. Thus, it needs $O(|V|^3\cdot |T|\cdot |P|)$ for each \ac{SD} pair $i\in I$ calculating for all paths, where $|P|$ denotes the maximum size of path set $P(i)$ between all \ac{SD} pairs. 
Since {\AlgoNameAbbrone} chooses one \ac{SD} pair among all pairs for each round and the number of chosen \ac{SD} pairs is at most $|I|$, the overall time complexity is $O(|V|^3\cdot |T|\cdot |P|\cdot |I|^2)$.


\section{Performance Evaluation} \label{sec: evaluation}

\subsection{Simulation Settings}
The model and \emph{default} parameters are as follows.
The \ac{QN} topology is generated using the Waxman model \cite{Waxman}, with $|V| = 100$ nodes deployed over a $150 \times 300$ km region, and an average edge length of $30$ km.
For clarity, a network topology example is given in Appendix \ref{app: topo} of the supplementary material.
We select $|I| = 50$ random \ac{SD} pairs within a batch of $|T| = 13$ time slots, where each time slot length is $\tau = 2 \text{ ms}$. This setup ensures that the total duration of a batch ($\tau\times|T|$) does not exceed $40\textrm{ ms}$, consistent with recent advances in quantum memory \cite{bradley2019ten}.
The choice of $\tau=2\textrm{ ms}$ is reasonable since a single swapping process requires approximately $1.1 \textrm{ ms}$ \cite{pompili2021realization}.
For each edge, we set $\lambda = 0.045/\textrm{km}$ \cite{sangouard2011quantum} and the entangling time $\mathfrak{T} = 0.25 \textrm{ ms}$ \cite{pompili2021realization}. Thus, the average entangling probability is around $0.9$, with the number of entangling attempts per time slot given by 
$\xi = \lfloor \frac{\tau}{\mathfrak{T}} \rfloor = 8$.
The initial fidelity for each edge is randomly sampled from the range $[0.7,0.98]$ \cite{zhao2022e2e}.
The average memory limit of a node (i.e., $c^t(v)$) is set between $6$ to $14$ units, aligned with \cite{Haldar2024Reducing} ($10$ to $16$ per node).
The swapping probability is set to $0.9$ \cite{shi2020concurrent, zhao2021redundant}.
Following \cite{bradley2019ten}, the parameters in Eq. (\ref{deco formula}) are set as $A = 0.25$, $B = 0.75$, $\mathcal{T} = 40 \text{ ms}$, and $\kappa = 2$. 
The fidelity threshold $\widehat{F}$ is set to $0.5$.
Then, we apply GREEDY \cite{pant2019routing}, Q-CAST \cite{shi2020concurrent}, and REPS \cite{zhao2021redundant} to generate the predefined path set $P(i)$ for each pair $i\in I$. For simplicity, they are denoted by G, Q, and R.
Each result is averaged over $50$ trials.

We compare {\AlgoNameAbbrtwo} and {\AlgoNameAbbrone} with the following methods.
\begin{enumerate}
    \item {\Merge} \cite{sangouard2011quantum} first conducts entangling for any two adjacent nodes until all entangled links on the path are created. After that, it greedily maximizes the number of swapping processes for each time slot until the end-to-end entangled link is built. Thus, {\Merge} has a near-balanced tree structure.
    \item {\Linear} \cite{farahbakhsh2022opportunistic} performs swapping operations one by one, starting from the source towards the destination, after all entangled links on the path are created.
    Thus, the numerology that {\Linear} tends to use is a near-biased tree structure.
    \item {\ASAP} \cite{Haldar2024Reducing} performs swapping as soon as two adjacent entangled links are successfully generated. {\ASAP} needs to bind the resources for re-entangling and swapping until a successful end-to-end entangled link is achieved.
    \item {\UB} is the fractional solution of \ac{LP} (\ref{eq: revised primal objective}), (\ref{eq: only select one constraint}), (\ref{eq: memory constraint}), (\ref{eq: primal relaxed x}), derived by the combinatorial algorithm with the separation oracle in {\AlgoNameAbbrtwo} to know the upper bound of the optimum. Note that it is not necessarily feasible for {\ProblemNameAbbr}.
\end{enumerate}%

\begin{figure}
\centering
    \subfigure[\#Requests vs Exp. Fid. Sum]{\label{fig: 8a}\includegraphics[width=.234\textwidth]{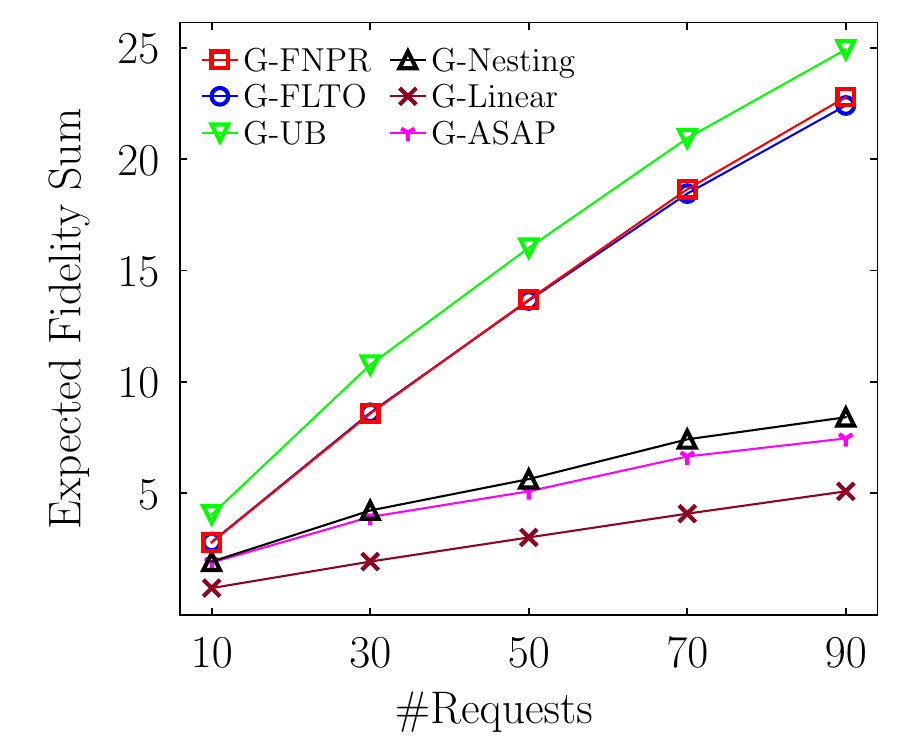}}
    \subfigure[\#Requests vs \#Acc. Requests]{\label{fig: 8b}\includegraphics[width=.234\textwidth]{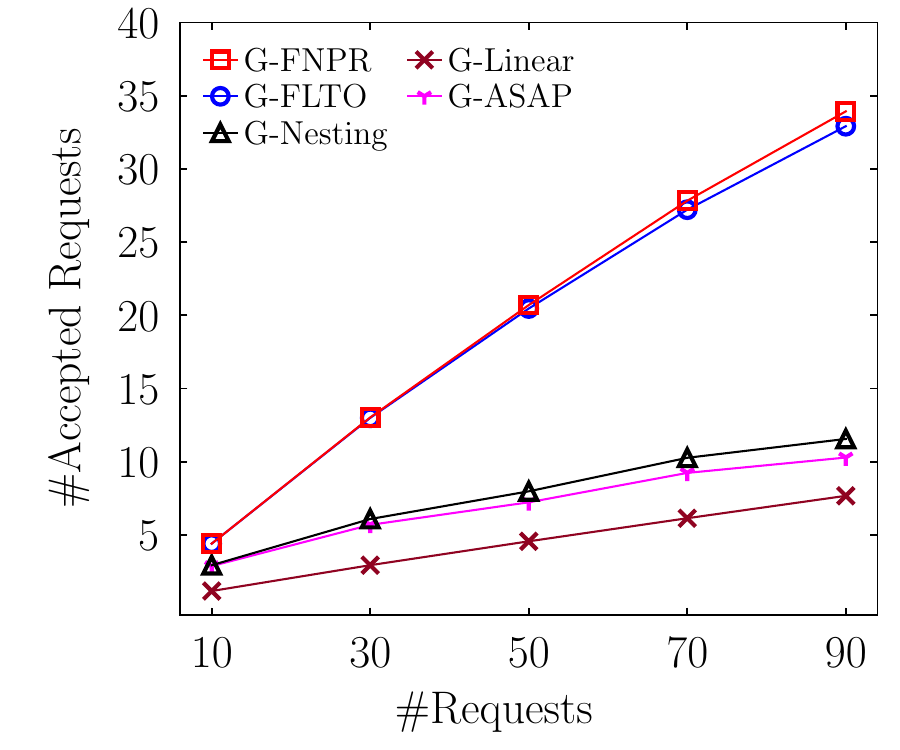}}
    \subfigure[\#Requests vs Exp. Fid. Sum]{\label{fig: 8c}\includegraphics[width=.234\textwidth]{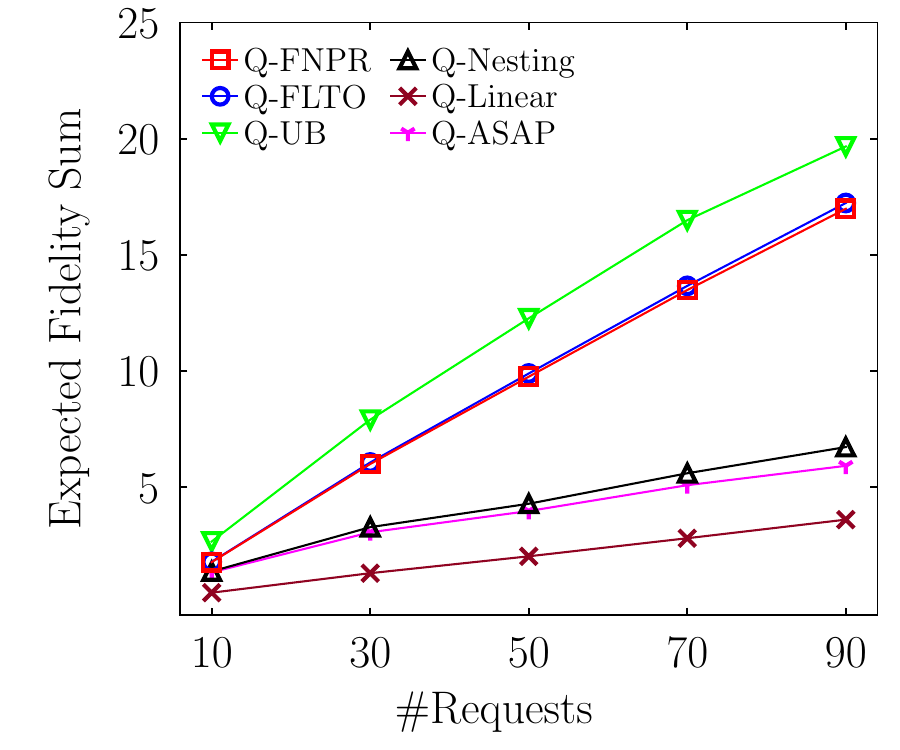}}
    \subfigure[\#Requests vs \#Acc. Requests]{\label{fig: 8d}\includegraphics[width=.234\textwidth]{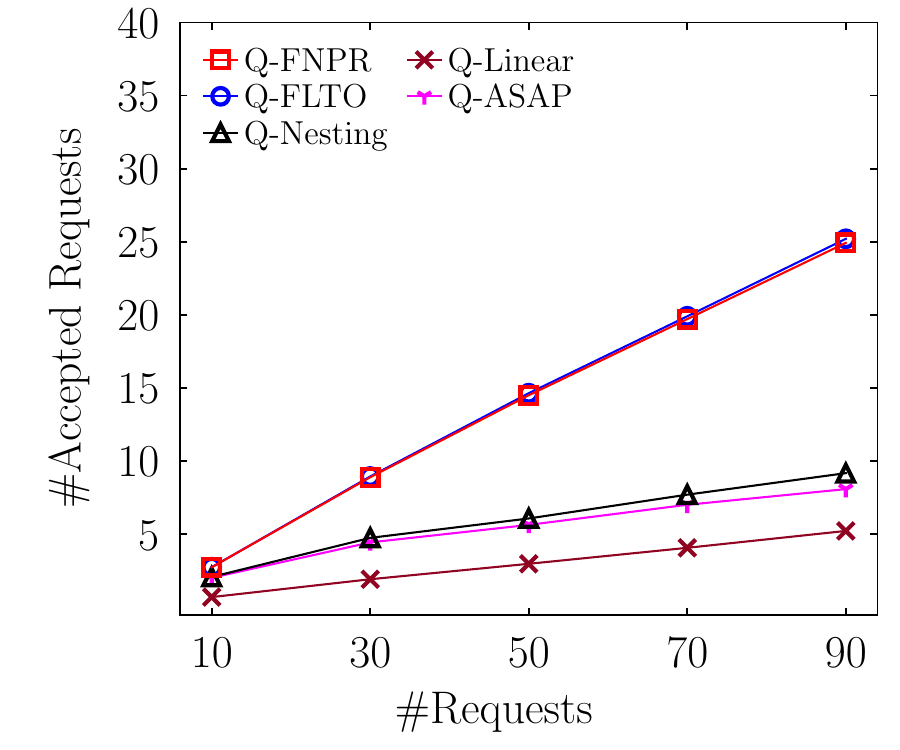}}
    \subfigure[\#Requests vs Exp. Fid. Sum]{\label{fig: 8e}\includegraphics[width=.234\textwidth]{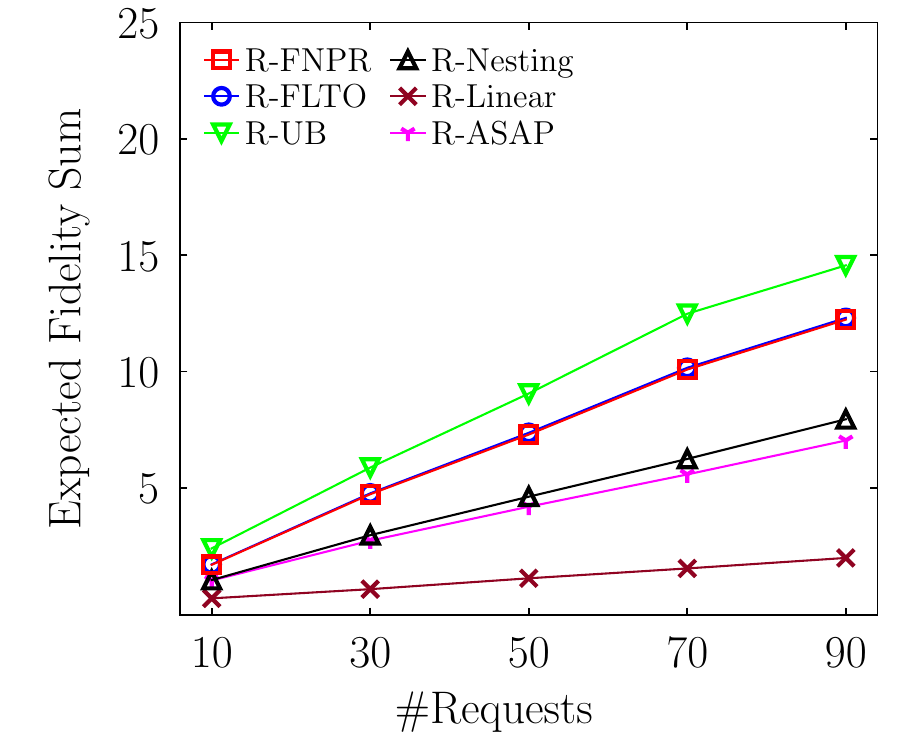}}
    \subfigure[\#Requests vs \#Acc. Requests]{\label{fig: 8f}\includegraphics[width=.234\textwidth]{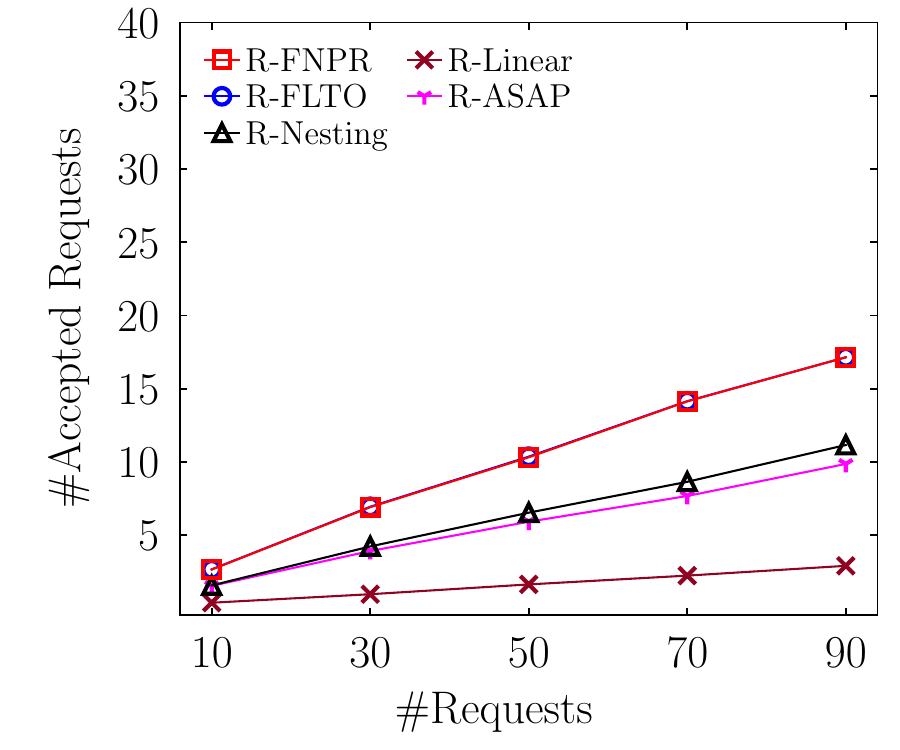}}
    \caption{Effect of different parameters and predefined paths on various metrics.%
    }
\label{fig: experiment result}
\end{figure}
\begin{figure}
\centering
    \subfigure[Swap. Prob. vs Exp. Fid. Sum]{\label{fig: 9a}\includegraphics[width=.234\textwidth]{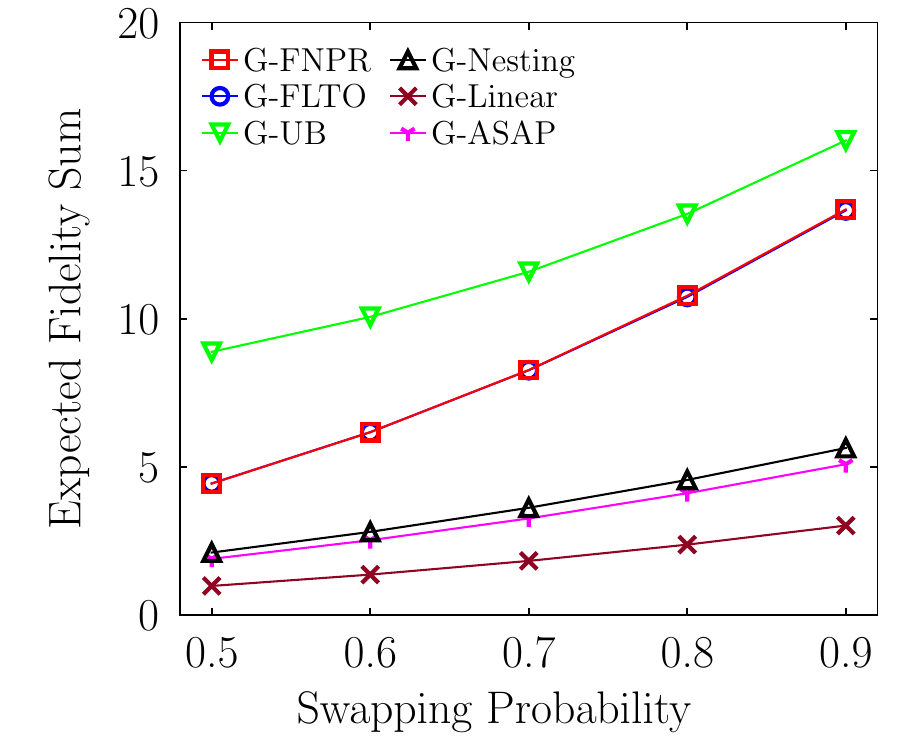}}
    \subfigure[Swap. Prob. vs \#Acc. Requests]{\label{fig: 9b}\includegraphics[width=.234\textwidth]{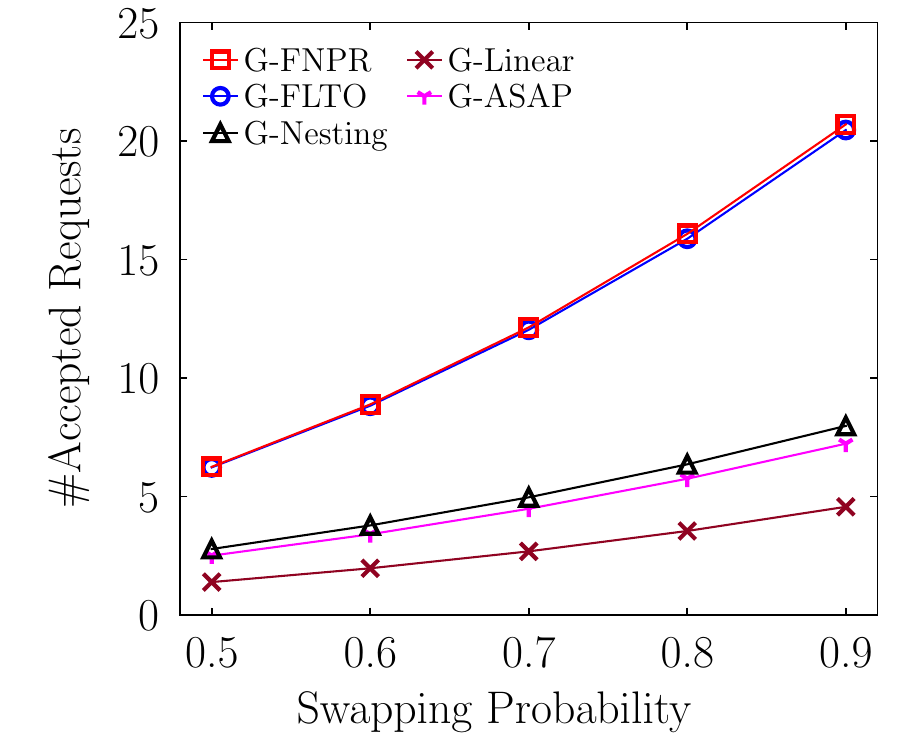}}
    \subfigure[Ent. Time vs Exp. Fid. Sum]{\label{fig: 9c}\includegraphics[width=.234\textwidth]{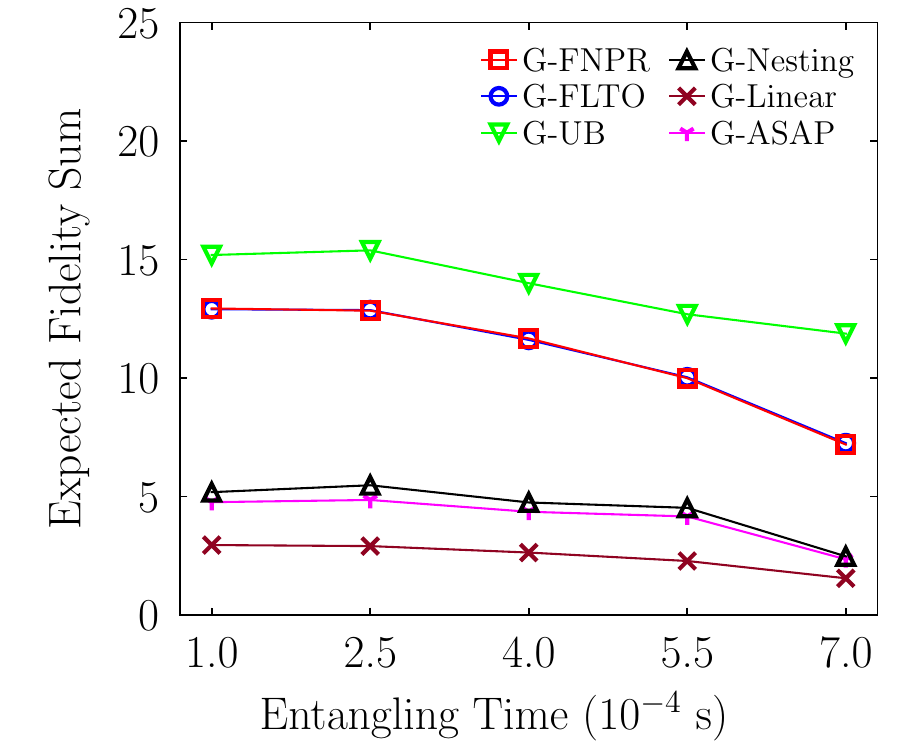}}
    \subfigure[Ent. Time vs \#Acc. Requests]{\label{fig: 9d}\includegraphics[width=.234\textwidth]{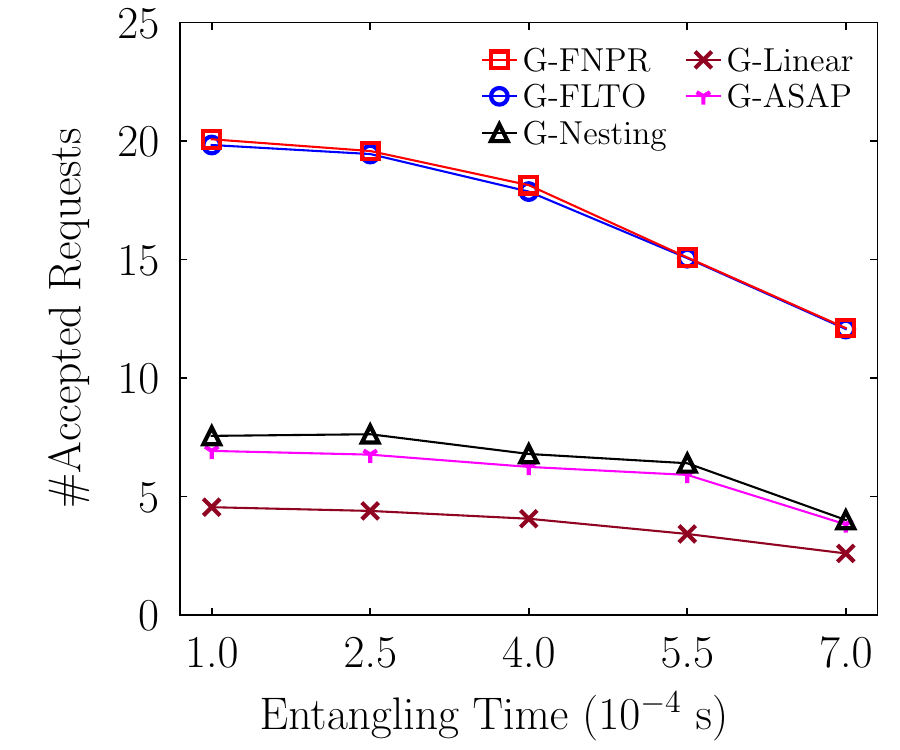}}
    \subfigure[$\tau$ vs Exp. Fid. Sum]{\label{fig: 9e}\includegraphics[width=.234\textwidth]{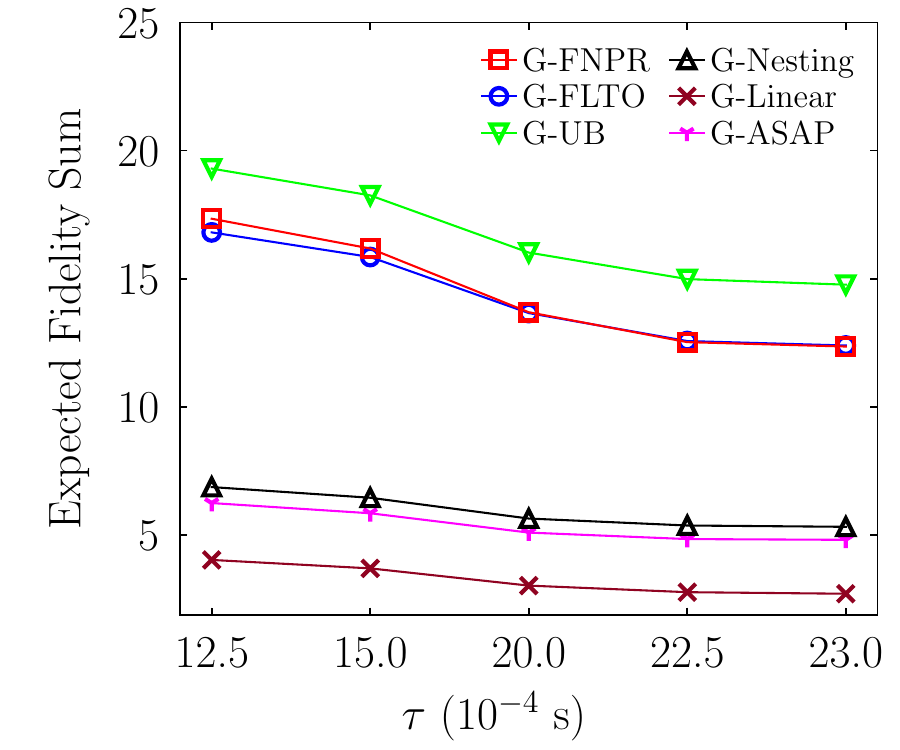}}
    \subfigure[$\tau$ vs \#Acc. Requests]{\label{fig: 9f}\includegraphics[width=.234\textwidth]{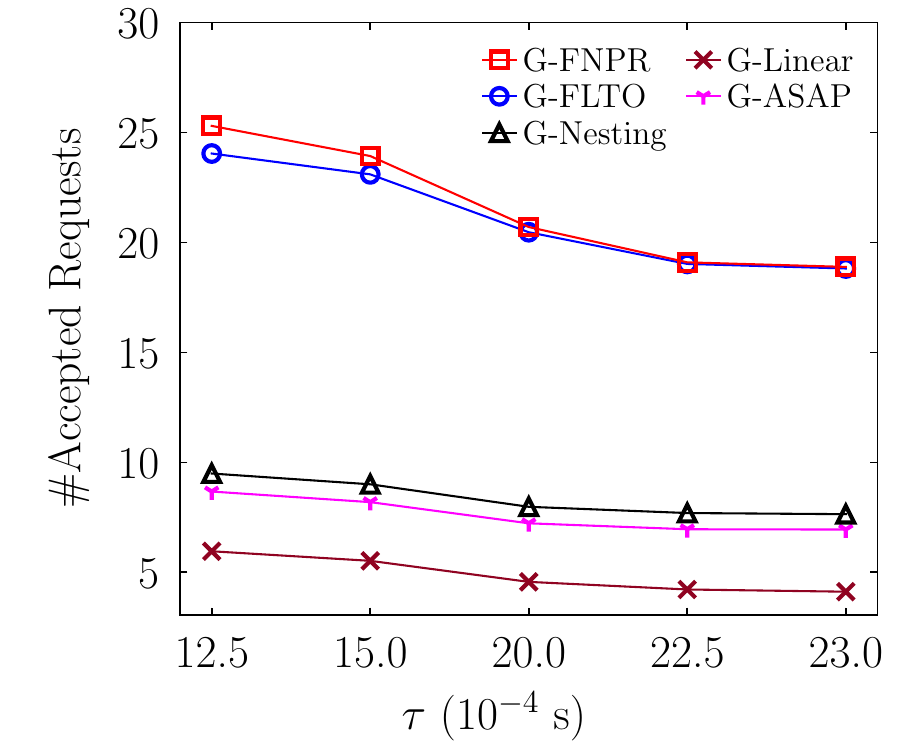}}
    \caption{Effect of different parameters on various metrics.}
\label{fig: experiment result2}
\end{figure}

\subsection{Numerical Results}

Figs. \ref{fig: experiment result}$-$\ref{fig: experiment result4} illustrate the expected fidelity sum and the number of accepted requests under different parameters. 
Overall, both {\AlgoNameAbbrtwo} and {\AlgoNameAbbrone} consistently demonstrate superior performance across all parameter settings, effectively enhancing fidelity and throughput while efficiently utilizing the resources of the quantum network.

\subsubsection{Effect of Number of Requests}

Fig. \ref{fig: experiment result} shows the expected fidelity sum and number of accepted requests across varying numbers of requests. 
As the number of requests increases, both the expected fidelity sum and the number of accepted requests generally exhibit an upward trend since more numerologies can be considered for selection to fully utilize the network resources.
No matter which routing algorithm is employed, {\AlgoNameAbbrtwo} and {\AlgoNameAbbrone} significantly outperform {\Merge}, {\Linear}, and {\ASAP}.
The results show that {\AlgoNameAbbrtwo} (and {\AlgoNameAbbrone}) averagely outperforms {\Merge}, {\Linear}, and {\ASAP} on expected fidelity sum by up to {\NMOutformMerge} ({\FLTOOutformMerge}), {\NMOutformLinear} ({\FLTOOutformLinear}), and {\NMOutformASAP} ({\FLTOOutformASAP}), respectively.
This is because {\AlgoNameAbbrtwo} can achieve near-optimum under a moderate $\tau$ and $|T|$ by the combinatorial algorithm with the separation oracle,
while {\AlgoNameAbbrone} uses \ac{REI} to choose numerologies, balancing resource utilization and fidelity.

\subsubsection{Effect of Swapping Probability}
In the literature, the swapping probability for simulations is typically between $0.5$ and $1$. To make the model more comprehensive and realistic, we compare the performance of algorithms under different swapping probabilities ranging from $0.5$ to $0.9$.
Figs. \ref{fig: 9a} and \ref{fig: 9b} show the expected fidelity sum and the number of accepted requests under a different swapping probability. 
Generally, as the swapping probability increases, both the expected fidelity sum and the number of accepted requests tend to rise. Additionally, {\AlgoNameAbbrtwo} and {\AlgoNameAbbrone} consistently outperform the other algorithms across all swapping probabilities, indicating they
are more robust and efficient in managing resources, leading to superior performance regardless of the swapping probability.

\subsubsection{Effect of Entangling Time}
The timing relationship between entangling and swapping may be affected by distance in reality, causing $\xi$ to vary under certain conditions. To achieve more comprehensive comparisons, we consider variations in the entangling time length ranging from $0.1 \text{ ms}$ to $0.7 \text{ ms}$, with corresponding $\xi$ values ranging from $8$ to $2$.
Figs. \ref{fig: 9c} and \ref{fig: 9d} show the expected fidelity sum and the number of accepted requests for different entangling time lengths. Generally, as entangling time length increases, both the expected fidelity sum and the number of accepted requests tend to decrease. This is because a smaller $\xi$ allows for fewer entangling attempts, leading to a lower entangling probability.
Additionally, {\AlgoNameAbbrtwo} and {\AlgoNameAbbrone} outperform the other algorithms.

\subsubsection{Effect of $\tau$}

Figs. \ref{fig: 9e} and \ref{fig: 9f} illustrate the expected fidelity sum and the number of accepted requests under a different $\tau$. Note that varying $\tau$ affects $\xi$, as the number of entangling attempts depends on the time slot length.
Although a larger $\tau$ allows more entangling attempts, leading to higher entangling probability, the results show that a larger $\tau$ conversely results in a lower expected fidelity sum and fewer accepted requests.
This is because the impact of decoherence over a longer time slot outweighs the benefits of increased entangling probability in our setting. 
Thus, some numerologies will not be admitted due to their low fidelity.
Fig. \ref{fig: 9e} further confirms that {\AlgoNameAbbrtwo} outperforms {\AlgoNameAbbrone} as $\tau$ is small, while {\AlgoNameAbbrone} is slightly better when $\tau$ is large, as described in Section VI.

\begin{figure}
\centering
    \subfigure[$|T|$ vs Exp. Fid. Sum]{\label{fig: 10a}\includegraphics[width=.234\textwidth]{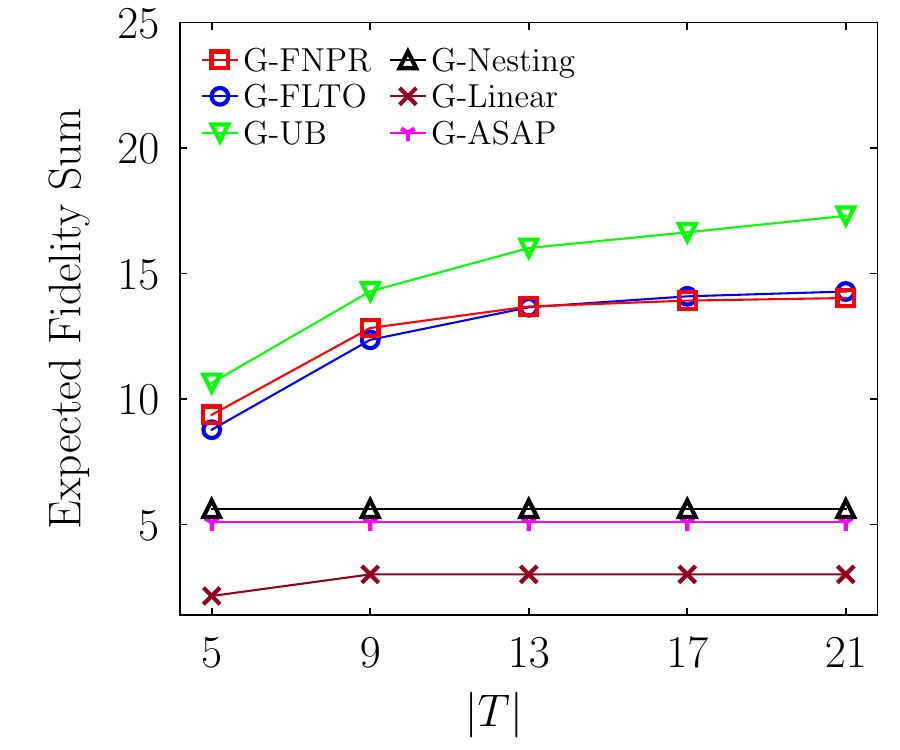}} 
    \subfigure[$|T|$ vs \#Acc. Requests]{\label{fig: 10b}\includegraphics[width=.234\textwidth]{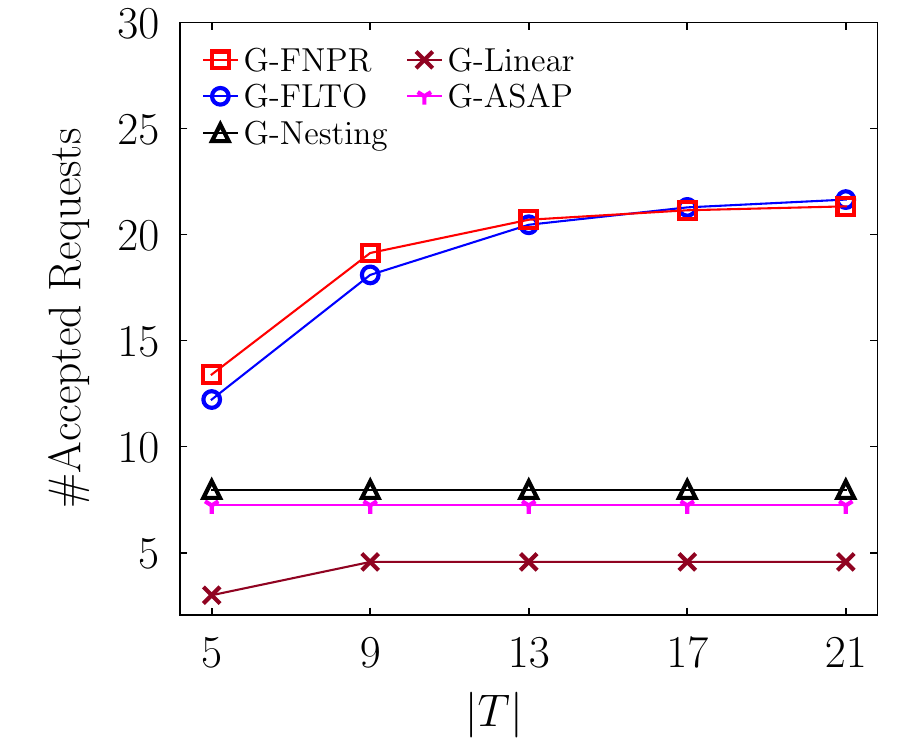}} 
    \subfigure[Avg. \#Mem. vs Exp. Fid. Sum]{\label{fig: 10c}\includegraphics[width=.234\textwidth]{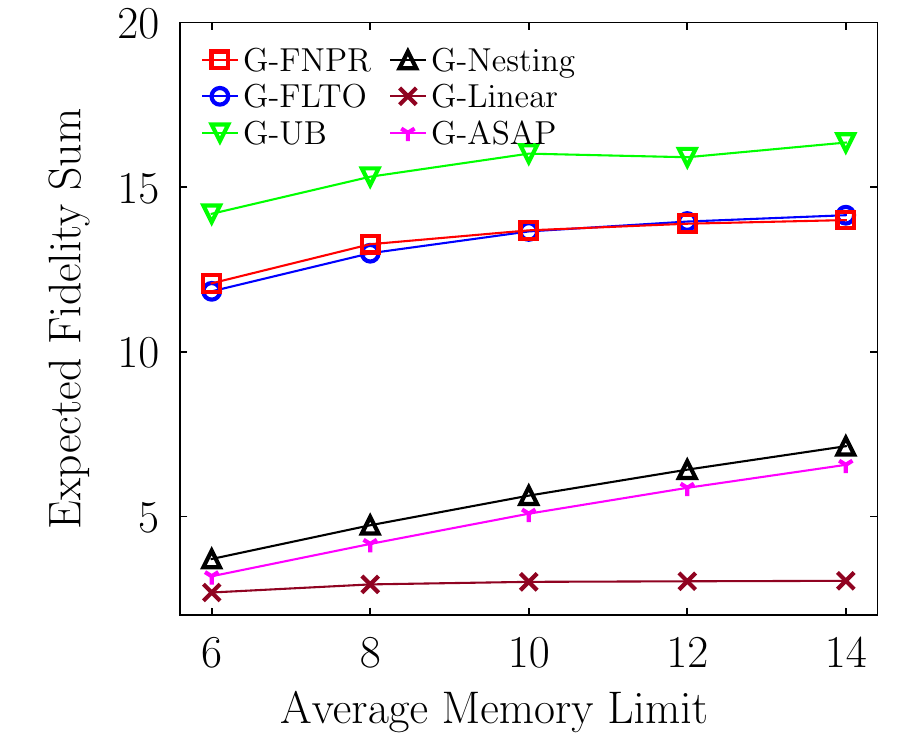}} 
    \subfigure[Avg. \#Mem. vs \#Acc. Requests]{\label{fig: 10d}\includegraphics[width=.234\textwidth]{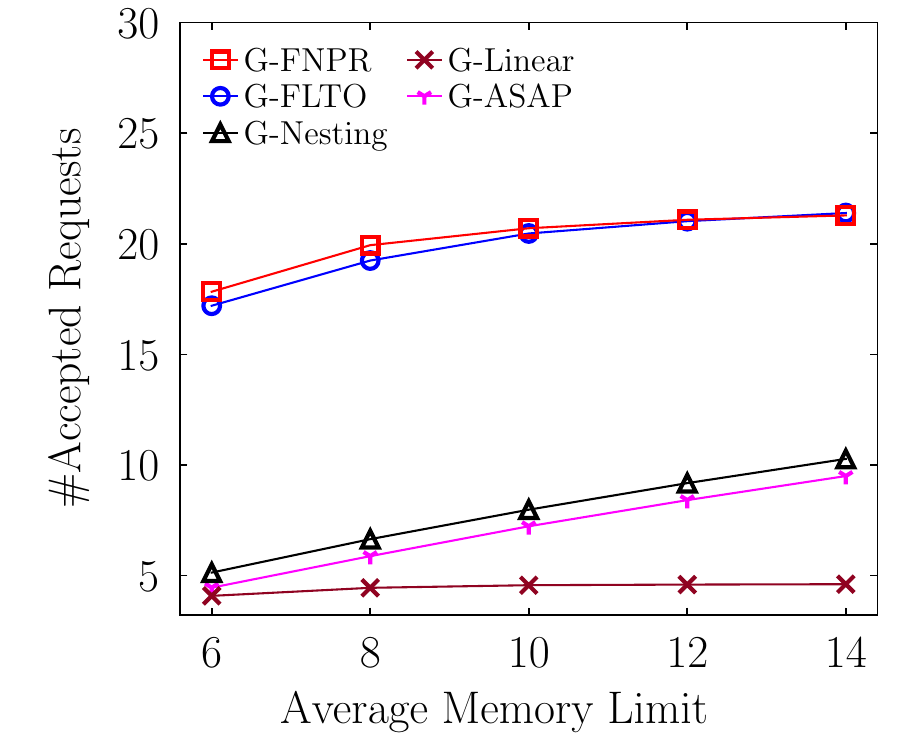}}
    \subfigure[Min. Init. Fid. vs Exp. Fid. Sum]{\label{fig: 10e}\includegraphics[width=.234\textwidth]{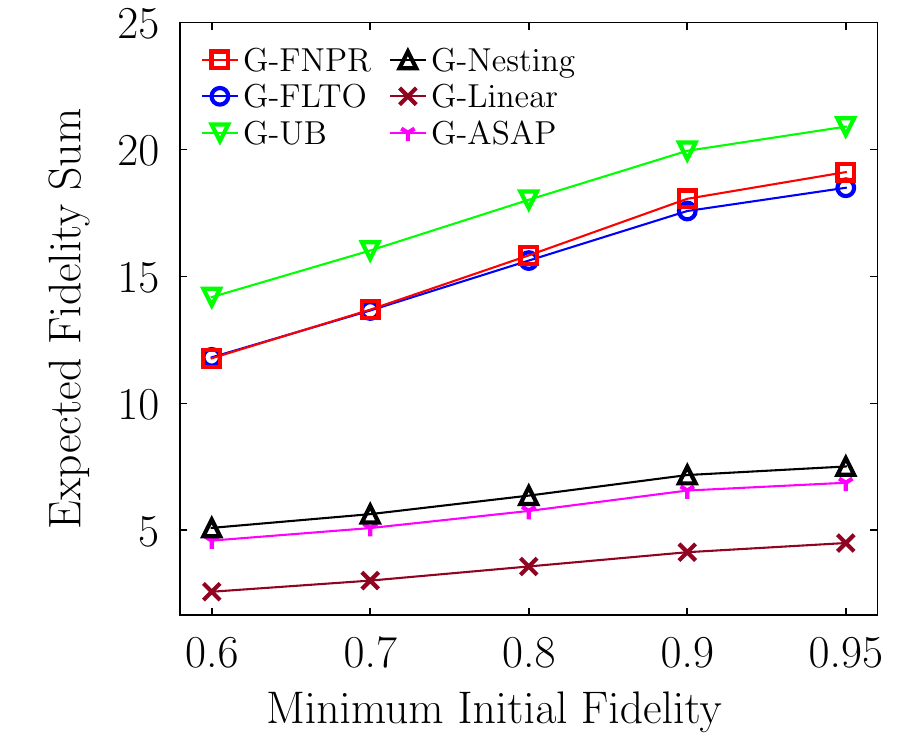}}  
    \subfigure[Min. Init. Fid. vs \#Acc. Requests]{\label{fig: 10f}\includegraphics[width=.234\textwidth]{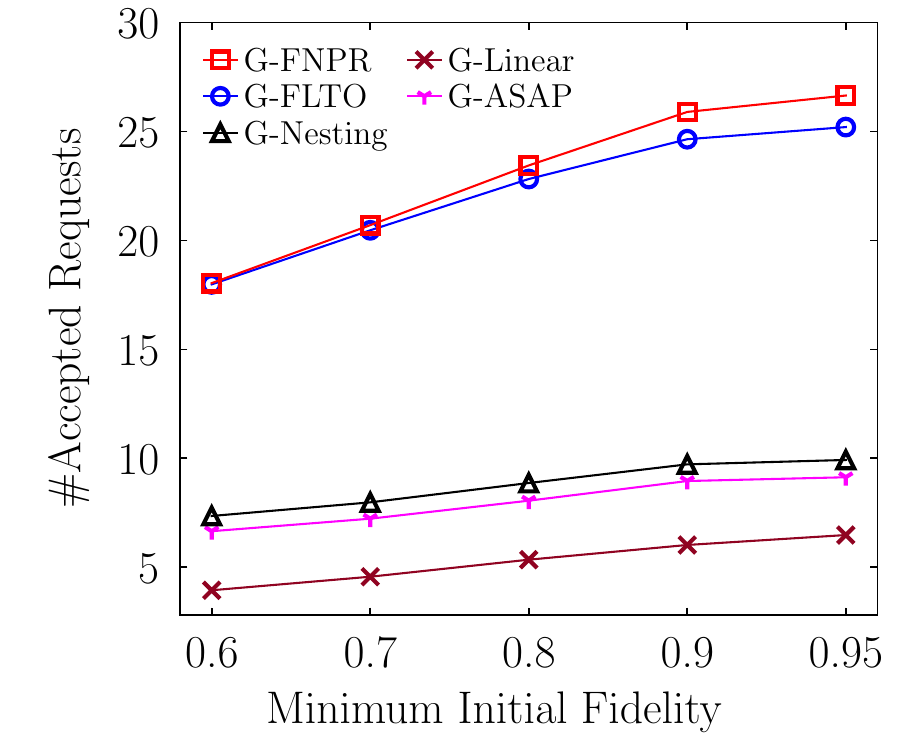}}
    \subfigure[Fid. Thr. vs Exp. Fid. Sum]{\label{fig: 10g}\includegraphics[width=.234\textwidth]{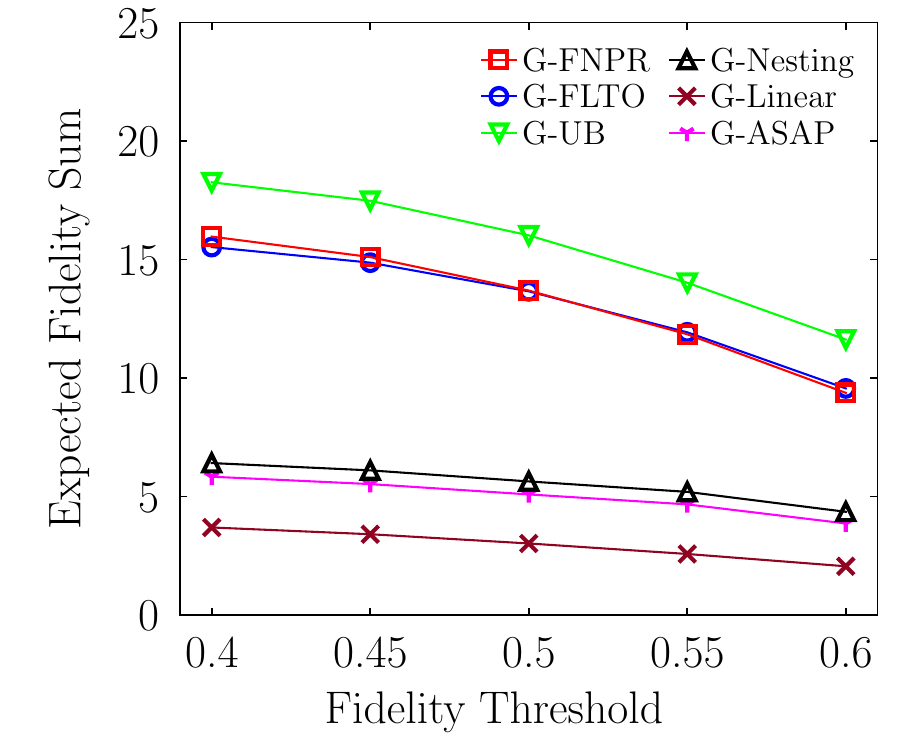}} 
    \subfigure[Fid. Thr. vs \#Acc. Requests]{\label{fig: 10h}\includegraphics[width=.234\textwidth]{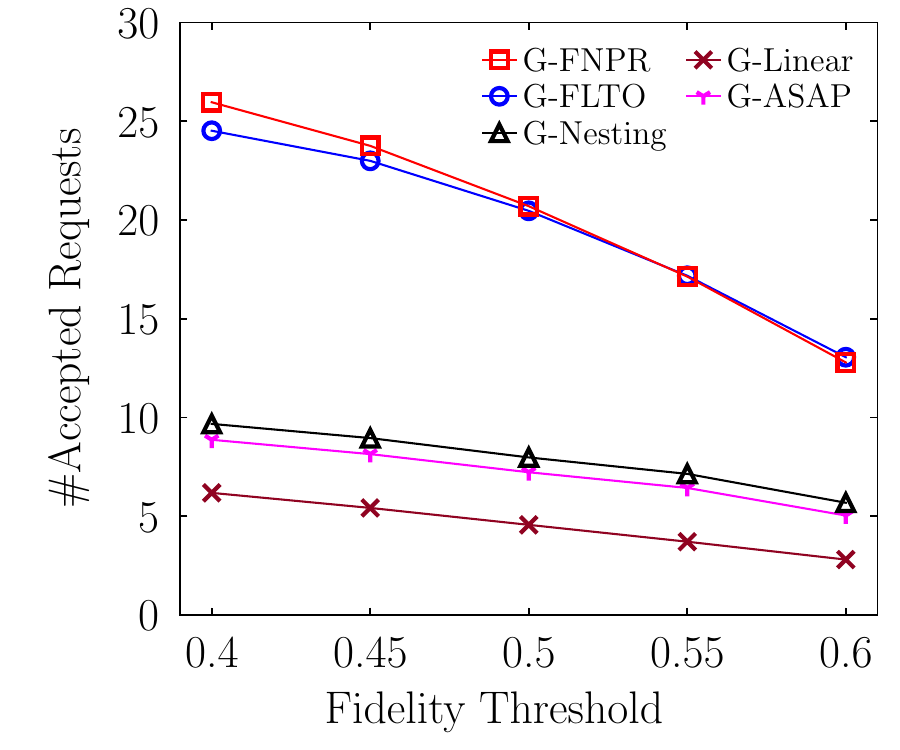}}
    \caption{Effect of different parameters on various metrics.}
\label{fig: experiment result3}
\end{figure}

\subsubsection{Effect of Resource Allocation}

Figs. \ref{fig: 10a}$-$\ref{fig: 10d} show the expected fidelity sum and the number of accepted requests under different $|T|$ and average memory limits. 
Generally, as $|T|$ or the amount of memory increases, both the fidelity sum and the number of accepted requests tend to increase, benefiting from sufficient resources.
In Figs. \ref{fig: 10a} and \ref{fig: 10b}, {\AlgoNameAbbrtwo} slightly outperforms {\AlgoNameAbbrone} as $|T|$ is small since {\AlgoNameAbbrtwo} can be close to the optimum solution as $\frac{F_{\max}}{F_{\min}}\approx 1$. 
However, in scenarios with a large $|T|$, {\AlgoNameAbbrone} proves more suitable since it can avoid selecting inefficient numerologies via our \ac{DP} algorithm with \ac{REI} when the resources are sufficient. The above results show that each has its own merits.

\subsubsection{Effect of Initial Fidelity and Threshold}

Figs. \ref{fig: 10e}$-$\ref{fig: 10f} and \ref{fig: 10g}$-$\ref{fig: 10h} show the expected fidelity sum and the number of accepted requests at varying minimum initial fidelity levels and fidelity thresholds. 
Generally, as the minimum initial fidelity increases, both the expected fidelity sum and the number of accepted requests increase since most numerologies can achieve high fidelity.
{\AlgoNameAbbrtwo} performs the best under higher minimum initial fidelity because it can approach the optimum solution as $\frac{F_{\max}}{F_{\min}}\approx 1$.
Besides, as the fidelity threshold increases, the expected fidelity sum tends to decrease inevitably; however, {\AlgoNameAbbrtwo} and {\AlgoNameAbbrone} still outperform others in all cases.

\begin{figure}
\centering
    \subfigure[Ent. Prob. vs Exp. Fid. Sum]{\label{fig: app_a}\includegraphics[width=.234\textwidth]{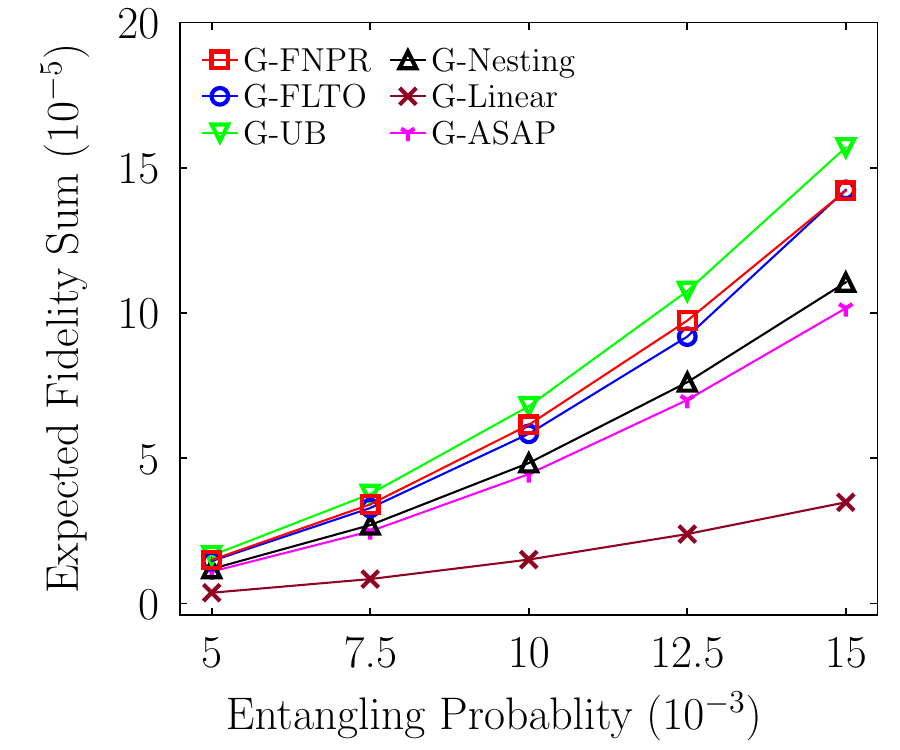}} 
    \subfigure[Ent. Prob. vs \#Acc. Requests]{\label{fig: app_b}\includegraphics[width=.234\textwidth]{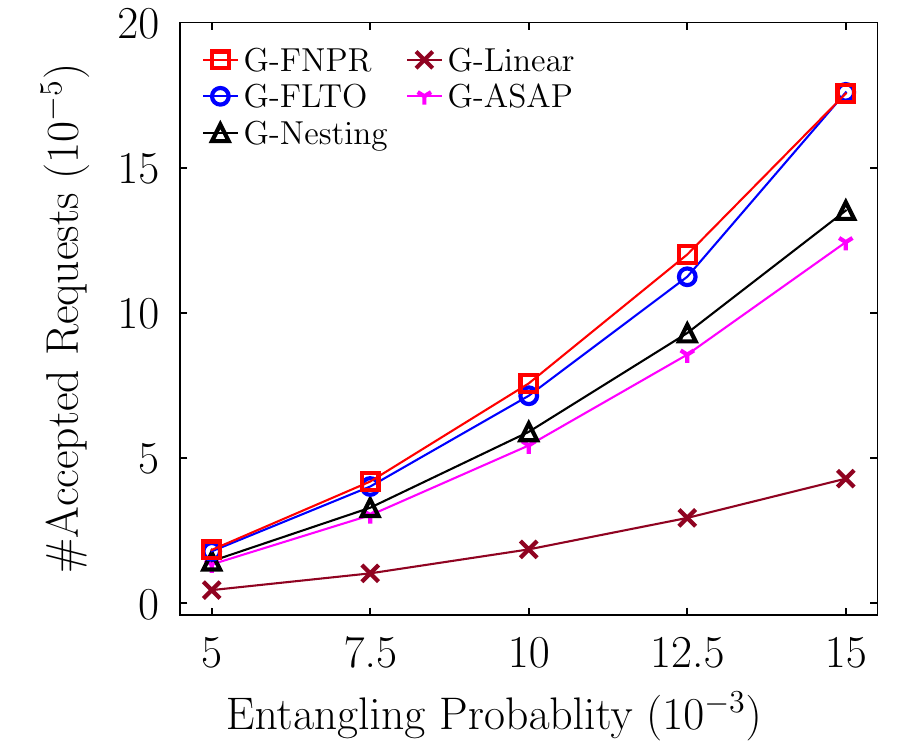}}
    \subfigure[Ent. Prob. vs Exp. Fid. Sum]{\label{fig: app_c}\includegraphics[width=.234\textwidth]{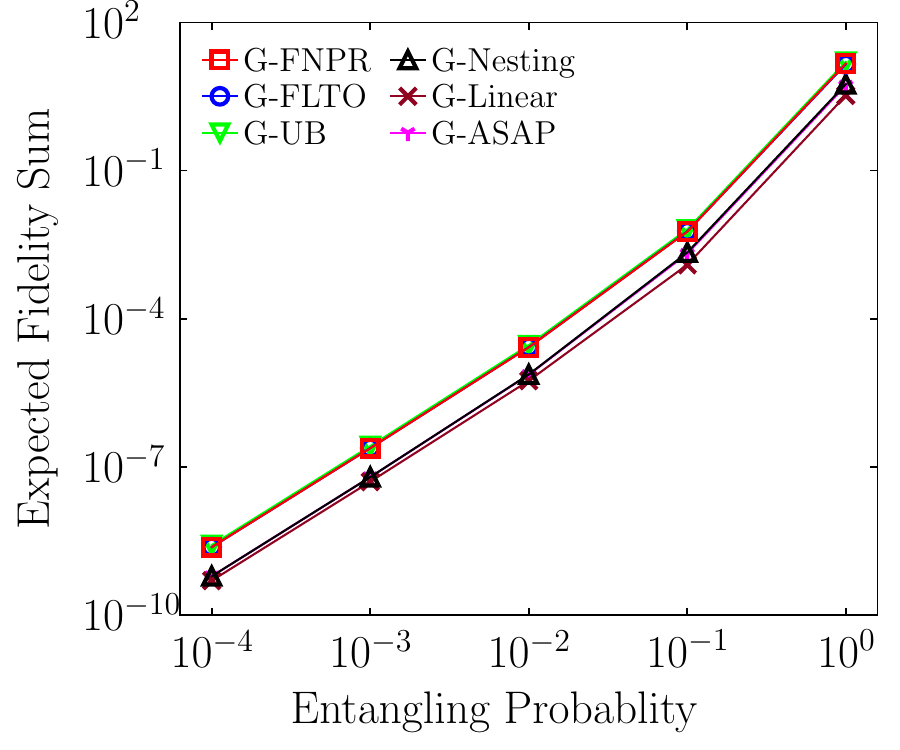}} 
    \subfigure[Ent. Prob. vs \#Acc. Requests]{\label{fig: app_d}\includegraphics[width=.234\textwidth]{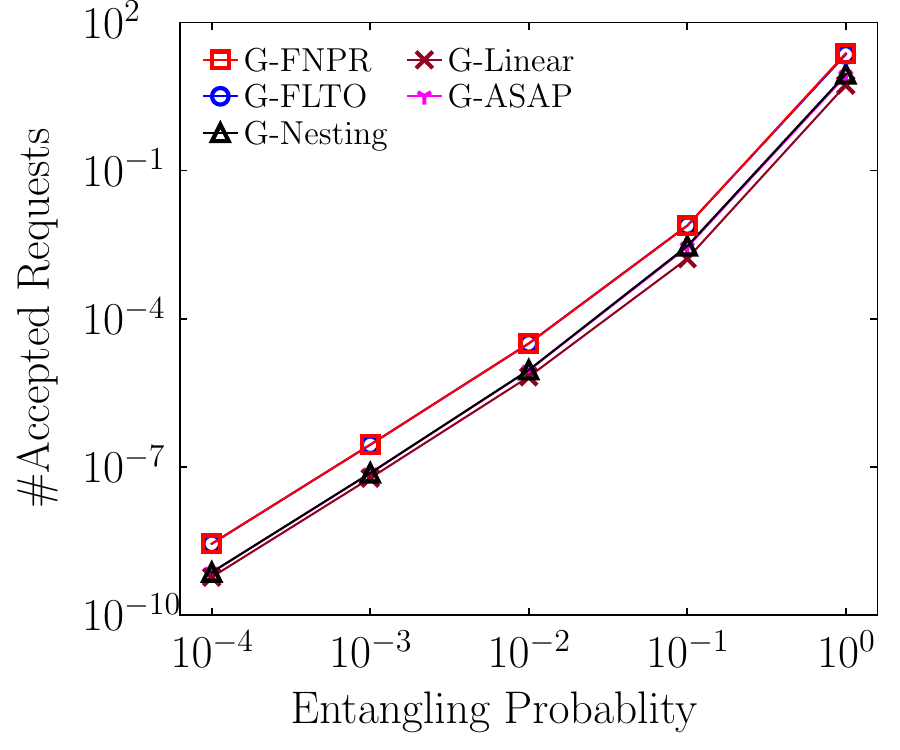}}
    \caption{Effect of different entangling probabilities on various metrics.}
\label{fig: experiment result4}
\end{figure}

\subsubsection{Effect of Entangling Probability}
To further observe the effect of low entangling probability, we conduct experiments by treating the entangling probability (i.e., $\Pr(u,v)$) as an abstract input parameter, as shown in Fig. \ref{fig: experiment result4}.
Figs. \ref{fig: app_a} and \ref{fig: app_b} show the expected fidelity sum and number of accepted requests across varying entangling probabilities. 
In this simulation, the entangling probability is set to range from $0.005$ to $0.015$, which covers the realistic parameter range discussed in \cite{metro-solid-state}.
Due to the low entangling probability, the fidelity sum and number of accepted requests are low under this setting for all simulated methods. 
Even under such adverse conditions, the proposed {\AlgoNameAbbrtwo} and {\AlgoNameAbbrone} still outperform the other algorithms by $20\%$ ($18\%$) to $75\%$ ($75\%$) on average.
This is because {\AlgoNameAbbrtwo} utilizes the combinatorial algorithm with randomized rounding to achieve near-optimal solutions, while {\AlgoNameAbbrone} invokes the \ac{DP} techniques to design an effective greedy algorithm. 
Both algorithms are designed to solve the {\ProblemNameAbbr} suitably. 
To more deeply observe the performance changes under different entangling probabilities, we conducted simulations across a wider range of entangling probabilities (i.e., from $0.0001$ to $1$), as shown in Figs. \ref{fig: app_c} and \ref{fig: app_d}. The results show that our proposed algorithms consistently outperform the other algorithms across all entangling probability levels, which ensures the robustness of the proposed algorithms.

\section{Conclusion} 
\label{sec: conclusions}

This paper proposes a promising short time slot protocol for entangling and swapping scheduling with a novel optimization problem {\ProblemNameAbbr}.
The {\ProblemNameAbbr} asks for the solution that selects the numerology (strategy tree) for each accepted request while maximizing the fidelity sum of all accepted requests.
To solve the {\ProblemNameAbbr}, we design two new algorithms.
{\AlgoNameAbbrtwo} is a bi-criteria approximation algorithm by solving an \ac{LP} with a separation oracle and rounding the solution for the cases where the time slot length and the number of time slots in a batch are small.
{\AlgoNameAbbrone} is a greedy algorithm with a proper index to call two \ac{DP}-based algorithms adaptively for the other cases.
Finally, the simulation results manifest that our algorithms can outperform the existing methods by 
up to $60\sim 78\%$ in general, and by $20\sim 75\%$ even under low entangling probabilities.

\bibliographystyle{IEEEtran}
\bibliography{References}



\appendices

\section{Discussion of the Counter-Intuitive Example} \label{app: discussion1}

We conducted additional simulations and clarified the reasons behind this counter-intuitive example.
There are two main factors explaining why Fig. \ref{fig: treeskew} has better fidelity than Fig. \ref{fig: treefull}:
1) In our decoherence model (i.e., Eq. (1)), we set $\kappa = 2$. Note that $\kappa$ controls the decoherence speed. For example, in Fig. \ref{fig: kappadiff}, we compare the two strategy trees in Figs. \ref{fig: treeskew} and \ref{fig: treefull} under different values of $\kappa$ ranging from $1$ to $2.3$. As $\kappa$ increases, the fidelity of the skewed strategy tree in Fig. \ref{fig: treeskew} tends to surpass that of the complete strategy tree in Fig. \ref{fig: treefull} because a higher $\kappa$ reduces decoherence rate.
2) Most of the swapping processes in Fig. \ref{fig: treeskew} consume two entangled pairs with similar fidelity and thus suffer less fidelity loss due to swapping processes. 
From the perspective above, we can also conclude that if all the links along the path have similar initial fidelity, the complete strategy tree will conversely outperform the other. 
For example, if all the generated pairs $(v_1,v_2)$, $(v_2,v_3)$, $(v_3,v_4)$, and $(v_4,v_5)$ have initial fidelity $0.98$, then the final fidelity of the strategy trees in Figs. \ref{fig: treeskew} and \ref{fig: treefull} will be $0.889$ and $0.891$, respectively.

\begin{figure}[t]
\centering
    \includegraphics[width= .45\textwidth]{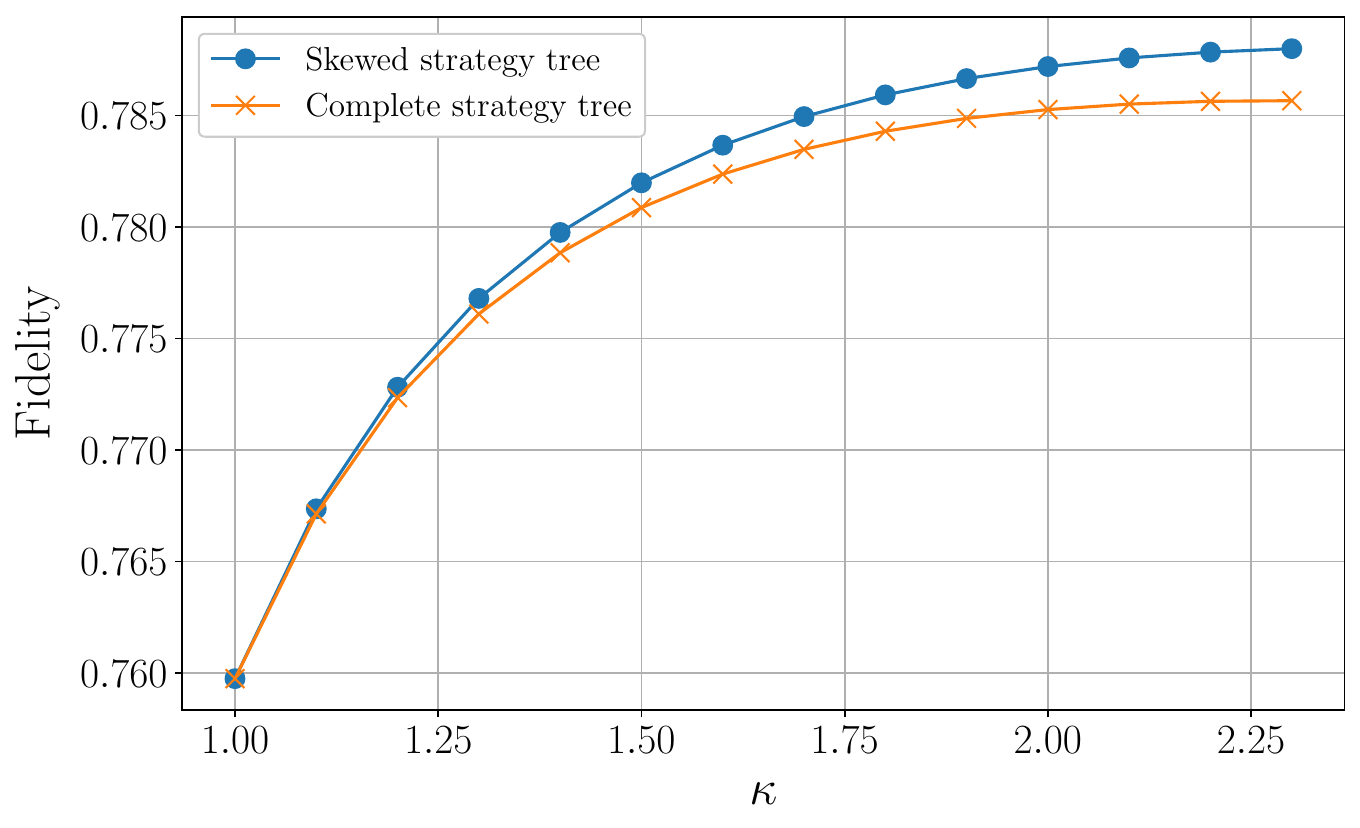}
    \caption{Comparison of fidelity between the strategy trees in Figs. \ref{fig: treeskew} and \ref{fig: treefull} under different values of $\kappa$.}
    \label{fig: kappadiff}
\end{figure}

\section{Example of Network topology} 
\label{app: topo}

Fig. \ref{fig: network_topology} shows a representative illustration of the network topology generated using the Waxman model.

\begin{figure}[t]
\centering
    \includegraphics[width=.42\textwidth]{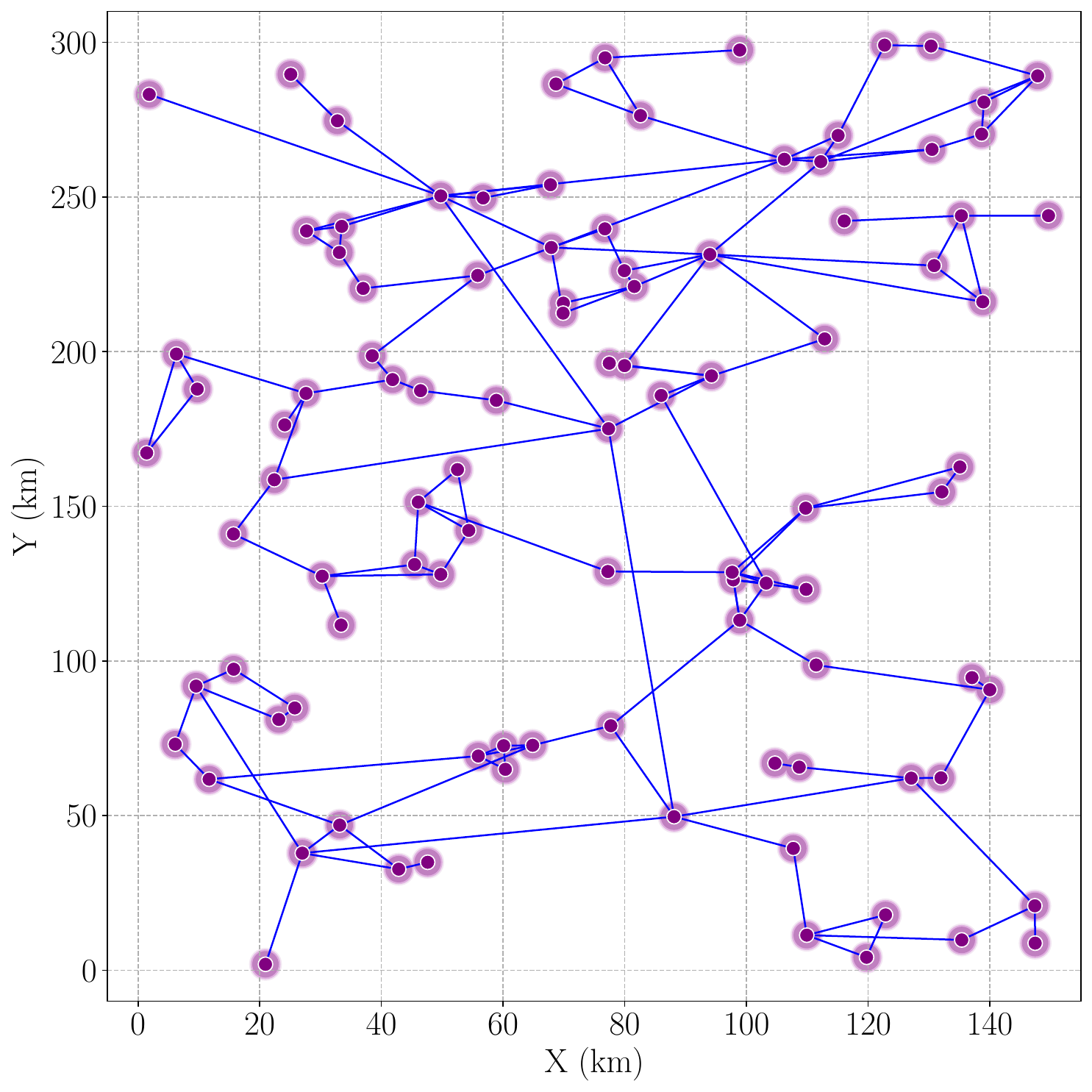}
    \caption{Network topology sample generated using the Waxman model.}
    \label{fig: network_topology}
\end{figure}
\end{document}